\documentclass[11pt]{article}
\RequirePackage{algorithm,algorithmic,graphicx,amssymb,epsf,epic,url,pifont,complexity}
\RequirePackage{fullpage}

\newtheorem{question}{Question}

\newtheorem{invariant}{Invariant}

\newcommand{\evalf}{\texttt{handle-free}}

\newcommand{\ins}{\texttt{handle-insertion}}
\newcommand{\del}{\texttt{handle-deletion}}

\newcommand{\setlvl}{\texttt{set-level}}
\newcommand{\temp}{\texttt{free-schedule}}
\newcommand{\unmatch}{\texttt{unmatch-schedule}}
\newcommand{\rise}{\texttt{rise-schedule}}
\newcommand{\shuffle}{\texttt{shuffle-schedule}}

\newcommand{\mal}{\texttt{Malicious}}
\newcommand{\shuf}{\texttt{Shuffler}}
\newcommand{\adder}{\texttt{Adder}}
\newcommand {\ignore} [1] {}

\ignore{
\usepackage{hyperref}
\hypersetup{
    colorlinks,
    citecolor=black,
    filecolor=black,
    linkcolor=black,
    urlcolor=black
}
}

\def\denseformat{
\setlength{\textheight}{9in}
\setlength{\textwidth}{6.9in}
\setlength{\evensidemargin}{-0.2in}
\setlength{\oddsidemargin}{-0.2in}
\setlength{\headsep}{10pt}
\setlength{\topmargin}{-0.3in}
\setlength{\columnsep}{0.375in}
\setlength{\itemsep}{0pt}
}

\newtheorem{theorem}{Theorem}[section]

\newtheorem{claim}[theorem]{Claim}
\newtheorem{lemma}[theorem]{Lemma}

\newtheorem{corollary}[theorem]{Corollary}

\newtheorem{comment}[theorem]{Comment}

\newtheorem{observation}[theorem]{Observation}

\def\boldhead#1:{\par\vskip 7pt\noindent{\bf #1:}\hskip 10pt}
\def\ithead#1:{\par\vskip 7pt\noindent{\it #1:}\hskip 10pt}

\def\inline#1:{\par\vskip 7pt\noindent{\bf #1:}\hskip 10pt}
\def\midinline#1:{\par\noindent{\bf #1:}\hskip 10pt}
\def\dnsinline#1:{\par\vskip -7pt\noindent{\bf #1:}\hskip 10pt}
\def\ddnsinline#1:{\newline{\bf #1:}\hskip 10pt}
\def\largeinline#1:{\par\vskip 7pt\noindent{\large\bf #1:}\hskip 10pt}
\long\def\comment #1\commentend{}
\long\def\commhide #1\commhideend{}
\long\def\commfull #1\commend{#1}
\long\def\commabs #1\commenda{}
\long\def\commtim #1\commendt{#1}
\long\def\commb #1\commbend{}
\long\def\commedit #1\commeditend{} 

\long\def\commB #1\commBend{}       

\long\def\commex #1\commexend{}     

\long\def\commsiena #1\commsienaend{}  
                                         
\long\def\commBI #1\commBIend{}  
                                         
\long\def\CProof #1\CQED{}

\def\blackslug{\hbox{\hskip 1pt \vrule width 4pt height 8pt
    depth 1.5pt \hskip 1pt}}
\def\QED{\quad\blackslug\lower 8.5pt\null\par}
\def\inQED{\quad\quad\blackslug}

\long\def\PPP#1{\noindent{\bf Proof:}{ #1}{\quad\blackslug\lower 8.5pt\null}}

\long\def\denspar #1\densend
{#1}

\newif\ifnotesw\noteswtrue
   {\ifnotesw\marginpar[\hfill\(\top\)]{\(\top\)}\fi}%
   {\ifnotesw\marginpar[\hfill\(\bot\)]{\(\bot\)}\fi}

\newcommand{\mnote}[1]%
    {\ifnotesw\marginpar%
        [{\scriptsize\it\begin{minipage}[t]{\marginparwidth}
        \raggedleft#1%
                        \end{minipage}}]%
        {\scriptsize\it\begin{minipage}[t]{\marginparwidth}
        \raggedright#1%
                        \end{minipage}}%
    \fi}

\def\cI{{\cal I}}

\def\cM{{\cal M}}
\def\cN{{\cal N}}
\def\cO{{\cal O}}

\def\cQ{{\cal Q}}

\def\MathF{\hbox{\rm I\kern-2pt F}}
\def\MathP{\hbox{\rm I\kern-2pt P}}
\def\MathR{\hbox{\rm I\kern-2pt R}}
\def\MathZ{\hbox{\sf Z\kern-4pt Z}}
\def\MathN{\hbox{\rm I\kern-2pt I\kern-3.1pt N}}
\def\MathC{\hbox{\rm \kern0.7pt\raise0.8pt\hbox{\footnotesize I}
\kern-4.2pt C}}
\def\MathQ{\hbox{\rm I\kern-6pt Q}}

\newsavebox{\ttop}\newsavebox{\bbot}

\def\eps{\epsilon}

\def\polylog{\mbox{polylog}}

\def\nin{{~\not \in~}}

\newcommand{\Prob}{\MathP}

\denseformat

\long\def\commabs #1\commabsend{}
\newcommand{\orient}[2]{#1 \rightarrow #2}
\begin{document}

\title{Fully Dynamic Almost-Maximal Matching: \\Breaking the Polynomial Barrier for Worst-Case Time Bounds}
\author{Moses Charikar\thanks{Stanford University. E-mail: {\tt moses@cs.stanford.edu}.}
\and Shay Solomon \thanks{IBM Research. E-mail: {\tt solo.shay@gmail.com}. Part of this work was done while the author was affiliated with Tel Aviv University and Stanford University.}}

\date{\empty}

\begin{titlepage}
\def\thepage{}
\maketitle

\begin{abstract}
Despite significant research effort, the state-of-the-art algorithm for maintaining an approximate matching in fully dynamic graphs has a polynomial \emph{worst-case} update time, even for very poor approximation guarantees. In a recent breakthrough, Bhattacharya, Henzinger and Nanongkai showed how to maintain a constant approximation to the minimum vertex cover, and thus also a constant-factor estimate of the maximum matching size, with polylogarithmic worst-case update time.  Later (in SODA'17 Proc.) they improved the approximation factor all the way to $2+\eps$. Nevertheless, the longstanding fundamental problem of \emph{maintaining} an approximate matching with sub-polynomial worst-case time bounds remained open.

We present a randomized algorithm for maintaining an \emph{almost-maximal} matching in fully dynamic graphs with polylogarithmic worst-case update time.
Such a matching provides $(2+\eps)$-approximations for both   maximum matching and   minimum vertex cover, for any   $\eps > 0$.
Our result was done independently of the $(2+\eps)$-approximation result of Bhattacharya et al., so:
\begin{itemize}
\item Together with Bhattacharya et al.'s result, it provides the first algorithm for maintaining $(2+\eps)$-minimum vertex cover with polylogarithmic worst-case update time.
\item More importantly, it provides the first algorithm for maintaining a $(2+\eps)$-maximum (integral) matching
with polylogarithmic worst-case update time.
\end{itemize}

The worst-case update time of our algorithm, $O(\poly(\log n,\eps^{-1}))$, holds deterministically, while the almost-maximality guarantee holds with high probability.
This result settles the aforementioned problem on dynamic matchings, and also provides essentially the best possible approximation guarantee for dynamic vertex cover (assuming the unique games conjecture).

To prove this result, we exploit a connection between the standard oblivious adversarial model, which can be viewed as inherently ``online'', and an ``offline'' model where some (limited) information on the future can be revealed efficiently upon demand. Our randomized algorithm is derived from a deterministic algorithm in this offline model.
We show that the deterministic guarantees obtained in the offline model translate 
into similar probabilistic guarantees in the standard model. This approach gives an elegant way to analyze randomized dynamic algorithms, and is of independent interest.  
\end{abstract}

\clearpage

\renewcommand{\baselinestretch}{0.95}\normalsize
\tableofcontents
\renewcommand{\baselinestretch}{1.0}\normalsize

\end{titlepage}


\section{Introduction}
Consider a \emph{fully dynamic} setting where we start from an initially empty graph on $n$ fixed vertices $G_0$,
and at each time step $i$ a single edge $(u,v)$ is either inserted in the graph $G_{i-1}$ or deleted from it, resulting in graph $G_i$.
The problem of maintaining a large matching or a small vertex cover in such graphs has attracted a lot of research attention in recent years.
In general, one would like to devise an algorithm for maintaining a ``good'' matching and/or vertex cover with $\polylog(n)$ \emph{update time}
(via a data structure that answers queries of whether an edge is matched or not in constant time),
where being ``good" means to provide a good approximation to the maximum matching and/or the minimum vertex cover,
and the update time is the time required by the algorithm to update the matching/vertex cover at each step.

One may try to optimize the \emph{amortized} (i.e., average) update time of the algorithm or its \emph{worst-case} (i.e., maximum) update time, over a worst-case sequence of graphs. There is a strong separation between the state-of-the-art amortized bounds and the worst-case bounds.
A similar separation exists for various other dynamic graph problems, such as spanning tree, minimum spanning tree and two-edge connectivity.
Next, we provide a brief literature survey on dynamic matchings. (See \cite{OR10,BGS11,PS16, Sol16} for a detailed survey.)

In FOCS'11, Baswana et al.\ \cite{BGS11} devised an algorithm for maintaining a \emph{maximal matching} with an expected amortized update time of $O(\log n)$
under the oblivious adversarial model.\footnote{The \emph{oblivious adversarial model} is a standard model, which has been   used for analyzing   randomized data-structures such as universal hashing \cite{CW77b}
and dynamic connectivity \cite{KKM13}.
The model allows the adversary to know all the edges in the graph and their arrival order, as well as
the algorithm to be used. However, the adversary is not aware of the random bits used by the algorithm,
and so cannot choose updates adaptively in response to the randomly guided choices of the algorithm.}
Building on the framework of \cite{BGS11}, Solomon \cite{Sol16} devised a different randomized algorithm whose expected amortized update time is $O(1)$.
Note that a maximal matching provides a 2-approximation for both the maximum matching and the minimum vertex cover.
Moreover, under the unique games conjecture (UGC), the minimum vertex cover cannot be efficiently approximated within any factor better than 2 \cite{KR08}.
In SODA'15, Bhattacharya et al.\ \cite{BHI15} devised a deterministic algorithm for maintaining $(2+\eps)$-approximate vertex cover with
amortized update time $O(\log n/\eps^2)$.
In STOC'16, Bhattacharya et al.\ \cite{BHN16} devised a different deterministic algorithm for maintaining $(2+\eps)$-approximate matching
with amortized update time  $O(\poly(\log n,\eps^{-1}))$.

All the known algorithms for maintaining a better-than-2 approximate matching or vertex cover require polynomial update time.\footnote{This statement is true for general graphs.
For low arboricity graphs, significantly better results are known; see \cite{KKPS14,HTZ14,PS16}.}
In FOCS'13, Gupta and Peng \cite{GP13} devised a deterministic algorithm for maintaining $(1+\eps)$-approximate matching with
a worst-case update time $O(\sqrt{m}/\eps^2)$.
Bernstein and Stein \cite{BS16} maintained $(3/2+\eps)$-approximate matching with an amortized update time $O(m^{1/4}/\eps^{2.5})$,
generalizing their earlier work  \cite{BS15} for bipartite graphs (in which they provide a worst-case update time guarantee).

There are two main open questions in this area.
The first is whether one can maintain a \emph{better-than-2} approximate matching in \emph{amortized} polylogarithmic update time.
The second is the following:
\begin{question}
Can one maintain a ``good'' (close to 2) approximate matching and/or vertex cover with \emph{worst-case} polylogarithmic update time?
\end{question}
In a   recent breakthrough, Bhattacharya, Henzinger and Nanongkai
devised a deterministic algorithm that maintains a \emph{constant} approximation to the minimum vertex cover,
and thus also a constant-factor estimate of the maximum matching size, with polylogarithmic worst-case update time.
While this result makes significant progress towards Question 1, this fundamental question remained open.\footnote{Later (in SODA'17 Proc.\  \cite{BHN17}) Bhattacharya et al.\ significantly improved the approximation factor all the way to $2+\eps$.
However, our result  was done independently to \cite{BHN17}.
Moreover, even if one considers the improved result of \cite{BHN17},  it solves Question 1 in the affirmative only for vertex cover, leaving the question on matching open;
see App.\ A for a more detailed discussion, which also covers a recent paper by Arar et al.\ \cite{ACCSW17}.}
In particular, no algorithm for maintaining a matching with sub-polynomial worst-case update time was known,
even if a polylogarithmic approximation guarantee on the matching size is allowed! 
%

In this paper we devise a randomized algorithm that maintains an \emph{almost-maximal matching (AMM)} with a polylogarithmic update time.
We say that a matching for $G$ is \emph{almost-maximal} w.r.t.\ some slack parameter $\eps$, or \emph{$(1-\eps)$-maximal} in short,
if it is maximal w.r.t.\ any graph obtained from $G$ after removing $\eps \cdot |\cM^*|$ arbitrary vertices, where $\cM^*$ is a maximum matching for $G$.
Just as a maximal matching provides a 2-approximation for the maximum matching and minimum vertex cover, an AMM provides a $(2+\eps)$-approximation.
We show that for any $\eps > 0$,
one can maintain an AMM with worst-case update time $O(\poly(\log n,\eps^{-1}))$, where the $(1-\eps)$-maximality guarantee holds with high probability.
Specifically, our update time is  $O(\max\{\log^7 n /\eps,\log^5 n / \eps^4\})$; although reducing this upper bound towards constant is an important goal,
this goal lies outside the scope of the current paper (see Section \ref{discuss} for some details).

\ignore{
We also strengthen the algorithm's worst-case guarantee to bound the number of changes (replacements) to the matching.
To this end we strengthen a scheme by  \cite{GP13} for maintaining approximate matchings, so as to bound the number of changes to the matching
in the worst-case, in addition to the update time. The problem of minimizing the number of matching replacements has received growing attention recently \cite{GKKV95,CDKL09,BLSZ15,BLZS17,BHR17}, as in some applications such as job scheduling, web hosting, streaming content delivery, data storage and hashing, a replacement of a matched edge by another one may be   costly. Moreover, bounding the number of matching replacements is also important when the maintained matching algorithm is used as a blackbox subroutine inside a larger data structure (cf.\ \cite{BS16,ADKKP16}).
We then also strengthen several results in this area \cite{GP13,BS15,PS16} in a similar way as immediate corollaries.
See the second remark following Theorem \ref{maint}, as well as Sections \ref{b11}, \ref{b21} and \ref{changes}, for further details.}
The algorithm's worst-case guarantee can be strengthened, using \cite{Sol18}, to bound the number of changes (replacements) to the matching.
Optimizing this measure is important in applications such as job scheduling, web hosting and hashing, where a replacement of a matched edge by another one is  costly. This measure is important also in cases where the matching algorithm is a blackbox subroutine inside a larger data structure (cf.\ \cite{BS16,ADKKP16}).

Our result resolves Question 1 in the affirmative, up to the $\eps$ dependency.
In particular, under the unique games conjecture, it is essentially the best result possible for the dynamic vertex cover problem.\footnote{Since our result (which started to circulate in  Nov.\ 2016) was done independently of the $(2+\eps)$-approximate vertex cover result of \cite{BHN17}, it provides the first $(2+\eps)$-approximation
 for both vertex cover (together with \cite{BHN17}) and (integral) matching.}

On the way to this result, we devise a \emph{deterministic} algorithm that maintains an almost-maximal matching with a polylogarithmic update time
in a natural \emph{offline model} that is described next.
This deterministic algorithm may be of independent interest, as discussed in Section \ref{technical}. 
Our randomized algorithm for the   oblivious adversarial model is derived from this deterministic algorithm in the offline model,
and this approach is likely to be useful in other dynamic graph problems.


\subsection{A Technical Overview} \label{technical}
Our algorithm and its analysis are elaborate and quite intricate.
While this may be viewed as a drawback, the strength of our paper is not just in the results we achieve, but also in the techniques we develop in establishing them.
We believe that our techniques are of independent interest and will find applications to other dynamic graph problems, also outside the context of matchings
\vspace{5pt}
\\
\noindent
{\bf The offline model.~} We start with describing an offline model that is a useful starting point for designing algorithms for matching in fully dynamic graphs. Suppose that the entire update sequence is known in advance, and is stored in some data structure.
Suppose further that for any $i$, accessing the $i$th edge update via the data structure is very efficient, taking $\polylog(n)$ or even $O(1)$ time.
A natural question to ask is whether one can exploit this knowledge of the future to obtain better algorithms for maintaining a good matching and/or vertex cover.
Consider in particular the \emph{maximal matching} problem.
Handling edge insertions can be done trivially in constant time. Handling edge deletions is the problematic part.
Consider a deletion of a matched edge $(u,v)$ from the graph.
If $u$ has a free neighbor, we need to match them, and similarly for $v$.
The algorithm may naively scan the neighbors of $u$ and $v$, which may require $O(n)$ time.
Surprisingly, this naive $O(n)$ bound is the state-of-the-art for general (dense) graphs, unless one allows both randomization and amortization as in \cite{BGS11,Sol16}.
Can one do better in the offline setting?

A natural strategy is to match a vertex $v$ along its incident edge $(v,w)$ that will be deleted last.
Indeed, by the time edge $(v,w)$ gets deleted from the graph, all other edges incident on $v$ must be ``new'',
i.e., having been inserted to the graph since the last time $v$ was matched.
So when $v$ becomes free, we should be able to afford (in the amortized sense) scanning all neighbors of $v$,
to find a free neighbor.  

This approach, however, only works when all neighbors of $v$ are free, which holds only initially. If some of them are matched, and we stick to the strategy of picking
the incident edge that will be deleted last, the update algorithm itself may be forced to delete matched edges from the matching.
Alas, due to deletions of matched edges by the algorithm, whenever a vertex becomes free, its neighbors are not necessarily new.
Instead, one may want to pick the incident edge that will be deleted last among those leading to free neighbors --
but determining the free neighbors of a vertex is indeed the crux of this problem!

Despite this hurdle, we argue that a dynamic maximal matching can be maintained in the offline setting \emph{deterministically}
with constant \emph{amortized} update time.
To this end, we make the following surprisingly simple observation:
The machinery of \cite{BGS11,Sol16} extends seamlessly to the offline setting above.
More specifically, instead of choosing the matched edge of $v$ uniformly at random among a certain subset of adjacent edges $E_v$
(which is computed carefully by the algorithms \cite{BGS11,Sol16}, details will be provided next),
in the offline setting we choose the matched edge to be the one that will be deleted last among $E_v$.
It is not difficult to verify that the analysis of \cite{BGS11,Sol16} carries over to the offline setting directly.

Notice that the resulting deterministic algorithm for the offline setting is inherently \emph{amortized}, whereas our focus in this work is on \emph{worst-case} bounds.
To obtain good worst-case bounds, we build on the machinery of \cite{BGS11,Sol16}.  
The price of translating the amortized bounds of \cite{BGS11,Sol16} into a similar worst-case bound is that the maintained matching is no longer maximal, but rather almost-maximal.\footnote{The amortized update time analysis of the algorithm from \cite{BGS11} (both the FOCS'11 and subsequent journal SICOMP'15 versions) was erroneous,
but was corrected in a subsequent erratum by the same authors. (The amortized update time analysis of the algorithm from \cite{Sol16} is different than the one used in \cite{BGS11}, and does not have that mistake.)  Although our algorithm builds on the machinery of \cite{BGS11,Sol16},
the mistake in \cite{BGS11} does not affect the current paper, as we provide an independent analysis for a different algorithm, which bounds the worst-case update time of our algorithm rather than the amortized update time.}
This translation is highly non-trivial, and is carried out in two stages.
First, we consider the offline setting, and devise a deterministic algorithm there.
Coping with the offline setting is easier than with the standard setting, as it allows us to ignore intricate probabilistic considerations, and to handle them separately.
The second stage is to convert the results for the offline setting to the standard setting.
The algorithm itself remains essentially the same. (Instead of choosing the edge that will be deleted last, choose a random edge.)
On the other hand, showing that the maintained matching remains almost-maximal requires more work.
This two-stage approach
thus provides an elegant way to analyze randomized dynamic algorithms, and we believe it would be useful in other dynamic graph problems as well.

Furthermore, the offline setting seems important in its own right. Indeed, in some real-life situations (e.g., road networks),
we might get some estimated schedule regarding future deletions of edges.
Moreover, there are some applications where the users of the network themselves may determine (to some extent) the lifespan of an edge (see, e.g., \cite{PP15}).
Note also that our algorithm does not need a complete knowledge of the future,
just an oracle access to an edge that will be deleted after a constant fraction of the other edges (from $E_v$) have been deleted.
In fact, the oracle does not have to be correct all the time, just on average.
Therefore, it seems that the offline setting may capture various practical scenarios.
\vspace{7pt}
\\
\noindent
{\bf The framework of \cite{BGS11,Sol16}.~}
We next provide a   rough description of the \emph{amortized} framework of \cite{BGS11,Sol16}.\footnote{Note that \cite{Sol16} builds on the framework of \cite{BGS11} and extends it; for clarity, we will not distinguish between \cite{BGS11} and \cite{Sol16}.}

Matched edges will be chosen randomly. If an edge $e = (u,v)$ is chosen to the matching uniformly at random among $k$ adjacent edges of either $u$ or $v$, w.l.o.g.\ $u$,
we say that its \emph{potential} is $k$.
Under the oblivious adversarial model, the expected number of edges incident on $u$ that are deleted from the graph before deleting edge $(u,v)$
is $k/2$.
Thus, following a deletion of a matched edge $(u,v)$ with potential $k$ from the graph, we have time $\tilde O(k)$ to handle $u$ and $v$ in the amortized sense.

Each vertex $v$ in the graph maintains a dynamically changing \emph{level} $\ell_v$;
roughly speaking, $v$'s level will be logarithmic in the potential value of the only matched edge adjacent to $v$.
Free vertices will be at level $-1$, and matched vertices will be at levels between 0 and $O(\log n)$.
Based on the levels of vertices, a dynamic edge orientation is maintained, where each edge is oriented towards the lower level endpoint.

When a vertex $u$ becomes free, the algorithm (usually) chooses a mate for it randomly.
If this mate $w$ is already matched, say to $w'$, the algorithm will have to delete edge $(w,w')$ from the matching in order to match $u$ with $w$.
However, note that we can compensate for the loss in potential value (caused by deleting edge $(w,w')$ from the matching)
if this potential loss is significantly smaller than the potential of the newly created matched edge on $u$.
Recalling that vertices' levels are logarithmic in their potential,
all neighbors of $u$ with lower level should have potential value at most half the potential value of the edge that was just deleted on $u$.
In other words, for each of these neighbors, we can afford to break their old matched edge.
Consequently, the mate $w$ of $u$ will be chosen uniformly at random among $u$'s  neighbors with lower level.

A central obstacle is to distinguish between neighbors of $u$ with level $\ell_u$ and those with lower level.
Indeed, it is possible that most of $u$'s neighbors have level $\ell_u$, and none of them can be chosen as mates for $u$.
Roughly speaking, the execution of the algorithm splits into two cases.
If the current out-degree of $u$ is not (much) larger than its out-degree at the time its old matched edge got created,
then we should be able to afford to scan all of them, due to sufficiently many adversarial edge deletions that are expected to occur.
Notice that in this case the charging argument is based on \emph{past} edge deletions.

The second case is when the out-degree of $u$ is (much) larger than what it was when the old matched edge got created.
The time needed for distinguishing $u$'s neighbors at level $\ell_u$ from those at lower levels could be significantly larger than
the ``money'' that we got from past edge deletions. In this case the algorithm \emph{raises} $u$ to a possibly much higher level $\ell^*$,
where there are not too many neighbors for $u$ at that level as compared to the number of its neighbors at lower levels.
Having raised $u$ to that level, we can perform the random sampling of its mate among all its neighbors of level lower than $\ell^*$.
Notice that in this case the charging argument is not based on past edge deletions, but rather on \emph{future} edge deletions.
(Future edge deletions may not occur, but we may charge the edges remaining in the graph to their last insertion.)
\vspace{7pt}
\\
\noindent
{\bf Our approach.~}
Notice that the framework of \cite{BGS11,Sol16} is inherently \emph{amortized}:
Every once in a while there are very ``expensive'' operations, which are charged to ``cheap'' operations that occurred in the past or will occur in the future.
To obtain a low worst-case update time, we should be \emph{cheap in any time interval}, meaning that we can rely neither on the past nor the future.
Consider a matched edge $(u,v)$ that is deleted by the adversary. We expect the adversary to make many edge deletions on at least one of these endpoints
before deleting this edge. However, it is possible that all these edge deletions occurred a long time ago.
The worst-case algorithm will not be able to exploit these edge deletions at this stage.

Consider the offline setting, and let $e_1,e_2,\ldots,e_\eta$ be $\eta$ arbitrary matched edges with the same potential value $k$.
For each such edge $e_i$, let $S(e_i)$ be its \emph{sample}, i.e., the set of all edges from which $e_i$ was chosen to the matching.
In the offline setting, we are guaranteed that $e_i$ will be deleted only after all $k-1$ other edges from its sample have been deleted.
However, it is possible that the adversary first deletes the first $k-1$ edges from the samples of each and every one of the  matched edges,
and only then turn to deleting the matched edges.
(Recall that the worst-case update time should hold with respect to a worst-case update sequence.)
Assuming $k$ is large, it takes a long time for the adversary to delete the first $\eta(k-1)$ edges from all   $\eta$ samples.
During all this time, the amortized algorithms of \cite{BGS11,Sol16} remain idle.
On the other hand, an algorithm with a low worst-case update time must be active throughout this time interval,
because immediately afterwards the adversary can remove the $\eta$ matched edges from the graph rather quickly,
much faster than the algorithm can add edges to the matching in their place, leading to a poor approximation guarantee.
Consequently, at any point in time, the algorithm needs to be proactive and protect itself from such a situation happening in the future.

Generally, while in an amortized algorithm invariants may be violated from time to time, and then restored via expensive operations,
an algorithm with a low worst-case update time should persistently ``clean'' the graph, making sure that it is never close to violating any invariant.
Naturally, we will need to maintain additional invariants to those maintained by the amortized algorithms of \cite{BGS11,Sol16}.
To this end we employ four different data structures that we call \emph{schedulers}, each for a different purpose.
Each of these schedulers consists of a logarithmic number of sub-schedulers, a single sub-scheduler per level.
Next we fix some level $\ell \approx \log k$, where $k$ is the potential of the matched edges on that level, and focus on it.

The scheduler $\unmatch$ will periodically remove   edges from the matching, one after another,
by always picking a matched edge whose \emph{remaining sample} (i.e., the set of edges from the sample that have not been deleted yet from the graph) is smallest;
indeed, such an edge is closest to getting deleted by the adversary.
As strange as it might seem, this strategy enables us to guarantee that only few matched edges will ever be deleted by the adversary.
Note that removing a matched edge from the matching is not a cheap operation, because we need to find new mates for the two endpoints of the edge.
Therefore, the execution of the scheduler must be \emph{simulated} over sufficiently many adversarial update operations.
During this time interval, the adversary may perform additional edge deletions.
Nevertheless, we control the rate at which the scheduler is working, and we can make sure that it works sufficiently faster than the adversary.
Thus, in this game between the scheduler and the adversary, the scheduler will always win.


The role of $\unmatch$ is to make sure that all the samples are pretty full.
Intuitively, this provides the counter-measure of relying on past adversarial edge deletions, as done in the amortized argument.
The next scheduler $\rise$ provides the counter-measure of relying on future adversarial edge deletions.
Recall that future edge deletions are   used in the amortized argument only in the case that a vertex had to rise to a higher level.
A vertex is rising only if its out-degree became too large with respect to its current level.
For this reason, we need to make sure that the out-degrees of vertices are always commensurate with their level, and this is where $\rise$ comes into play.
This scheduler will periodically raises vertices to the level $\ell$ of which it is in charge, one after another,
by always choosing to raise a vertex with the largest number of neighbors at level lower than $\ell$.
Intuitively, such a vertex is closest to getting chosen to rise to level $\ell$ or higher.
Although  the two schedulers are based on the same principle, the game that we play here is not between the scheduler and the adversary,
because here the algorithm itself may change the level of vertices and their out-degree, so $\rise$ has to compete against both the adversary
and the algorithm. In contrast to the other scheduler, speeding up the rate at which $\rise$ works will not help winning the game.
Instead, we manage to bound the speed of the scheduler with respect to that of the (adversary + algorithm),
which enables us to show that the out-degree of vertices is always in check. 

For the offline model, these two schedulers suffice. However, in the standard oblivious adversarial model,
the adversary will manage to destroy some matched edges from time to time.
The scheduler $\temp$ periodically handles all the vertices that become free due to the adversary, one after another.
Using the property that all samples are always pretty full, we manage to prove that only an $\eps$-fraction of the matched edges get destroyed by the adversary at any time interval. Note that this bound is probabilistic --
to make sure that it indeed occurs with high probability, we also use another scheduler $\shuffle$, which periodically
removes  a random edge from the matching. For technical reasons, it is vital that $\shuffle$ would work sufficiently faster than some of the other schedulers.

Finally, we point out another difficulty in getting a worst-case update time out of the algorithms of \cite{BGS11,Sol16}.
Following a single update operation,  these amortized algorithms may remove \emph{themselves} many matched edges from the matching, one after another.
Before this process finishes, the algorithms make sure to add new matched edges instead of the ones that got removed.
However, these algorithms may require a lot of time before starting to ``repair'' the matching.
In particular, if we simulate their execution, performing just a few computational steps per adversarial update,
we might get a very poor matching at some points in time, even without the ``help'' of the adversary!
Our new algorithm employs the aforementioned schedulers to guarantee that it never removes many matched edges before adding others in their place.
\vspace{7pt}
\\
\noindent
{\bf Technical Highlights.~}
Describing the algorithm and analysis will require some lengthy preliminaries.
Before diving into details,
we wish to highlight some novel technical aspects of the paper.

The schedulers $\unmatch$ and $\rise$ are implemented and analyzed by analogy to a balls and bins game between two players from \cite{LO88} (Section \ref{schedulersAnal}).
While this game has been studied before even in the context of dynamic graph algorithms that guarantee worst-case update times \cite{Thorup05,ACK17,Wulff16},
the worst-case update time achieved in these papers (for the dynamic all-pairs shortest paths problem in \cite{Thorup05,ACK17} and for the dynamic MSF in \cite{Wulff16})
is polynomial. In particular, the approach that we take 
is inherently different from that in previous works \cite{Thorup05,ACK17,Wulff16},
and we  believe that it will inspire more usage of this game in dynamic graph algorithms with polylogarithmic worst-case update time in the future.

In the balls and bins formulation, $\unmatch$ competes with an adversary that is removing balls from bins and trying to get the number of balls in some bin below a certain threshold (Section \ref{sec32}).
On the other hand, $\unmatch$ can remove bins to prevent them from becoming under full.
Since we make sure that $\unmatch$ works sufficiently faster than the adversary, we can show that $\unmatch$ wins the game,
which ensures that all the samples are close to full.
In Section \ref{sec33} we analyze a similar in spirit yet far more intricate game, concerning $\rise$. There balls are being added to bins, and we don't want the bins become over full.
The main twist is that the algorithm itself is competing against $\rise$,
so changing the speed of $\rise$ relative to the adversary does not help.
The key question is: can $\rise$ maintain all bins below a certain overload threshold?
This ensures that outdegrees of vertices are commensurate with their level. The   difficulty is not in analyzing the abstract balls and bins game,
but rather in translating the algorithm's operation to this game, which is where the various schedulers come into play.
One of the challenges in this translation is to cope with the interdependencies between games played at multiple levels.

In Section \ref{sec:almost} (Lemma \ref{whp}) we show that (with high probability) there cannot be too many temporarily free vertices due to the adversary at any point in time,
thus ensuring that the matching maintained by the algorithm is almost-maximal.
To prove this lemma, we build on the invariant that all samples are close to full (due to $\unmatch$).
Alas, this invariant by itself is not enough for guaranteeing a high probability bound, and thus we resort to the random shuffling provided by $\shuffle$.
We believe that this shuffling ``fix'' can be applied in various dynamic graph problems, also outside the scope of matchings or even of worst-case bounds, and is one example of the generality of our techniques. In general, we believe that our techniques are of broader applicability and interest than to the area of dynamic matchings.

\vspace{-6pt}
\subsection{Organization}
\paragraph{Main text.~}
The invariants, data structures and basic principles that govern the operation of the update algorithm are presented in Section \ref{basics}, whereas
the procedures that the update algorithm employs, as well as the analyses of those procedures, are given in Section \ref{procedures}.
The various schedulers used and orchestrated by our algorithm are the focus of Section \ref{sec22}, but are also described in other parts of Section \ref{basics} and in Section \ref{procedures}. Section \ref{schedulersAnal} is devoted to the analysis of those schedulers.
Our mechanism for resolving potential conflicts between the   schedulers is described and analyzed in Section \ref{exceptions}.

We prove
in Section \ref{sec:almost}  that the matching maintained by our   algorithm is almost-maximal. Lemma \ref{whp} is central to the argument of the almost-maximality guarantee,
 and its proof is spread over Sections \ref{highlight}--\ref{cont}. In Section \ref{conclude} we derive the main results of this paper as corollaries of Lemma \ref{whp}.

Two simplifying assumptions used by our algorithm are formally justified in Section \ref{just}.

Finally, a brief discussion with some open problems is given in Section \ref{discuss}.

\vspace{-6pt}
\paragraph{Appendix.~}
The analysis of the schedulers from Section \ref{schedulersAnal} relies on some  balls and bins game by \cite{LO88};
a proof of a sufficient condition for winning that game is given for completeness in App.\ \ref{ballsbins}.
In App.\ \ref{3+eps} we sketch an argument showing that a $(3+\eps)$-approximate matching can be obtained as a corollary of \cite{BHN17}. 
\section{The Update Algorithm, Part I: Basics and Infrastructure} \label{basics}

\subsection{\emph{Levels}: data structures and invariants} \label{sec21}
Our   algorithm builds on the amortized algorithms by \cite{BGS11,Sol16},
which maintain for each vertex $v$ a \emph{level} $\ell_v$, with $-1 \le \ell_v \le \log_\gamma (n-1)$, where $\gamma = \Theta(\log n)$.\footnote{Note that we use logarithms in base $\gamma = \Theta(\log n)$,
whereas \cite{BGS11} and \cite{Sol16} use logarithms in base 2 and 5, respectively. For the offline setting, we can use base 2, and this change leads to shaving a logarithmic factor from the update time.}
Based on the levels of vertices, a dynamic edge orientation is maintained; the out-degree of a vertex under this orientation, which is the number of its outgoing edges, serves as an important parameter.
The amortized algorithms of \cite{BGS11, Sol16} maintain the following   invariants (Invariants \ref{first}(a)-\ref{first}(d)) at all times.
This means that these invariants hold at the \emph{end of the execution} of the corresponding update algorithms,
i.e., before the next update operation takes place.
These invariants may become violated \emph{throughout the execution} of the update algorithms.
Furthermore, the runtime of the update algorithms of \cite{BGS11,Sol16} may be $\Omega(n)$ in the worst case,
thus it may take them a lot of time to restore the validity of these invariants, once violated.
We added a comment to the right of each of these invariants, where the comment is   /* maintained */ or /* partially maintained */,
to  indicate whether the respective invariant is maintained fully or only partially   by our new algorithm.
\begin{invariant} \label{first}
(a) Any matched vertex has level at least 0. /* maintained */
\\(b) The endpoints of any matched edge are of the same level, and this level remains unchanged until the edge is deleted from the matching.
(We   henceforth define the level of a matched edge, which is at least 0 by item (a), as the level of its endpoints.)  /* maintained */
\\(c) Any free vertex has level -1 and out-degree 0. (The matching is maximal.) /* partially maintained */
\\(d) An edge $(u,v)$ with $\ell_u > \ell_v$ is \emph{oriented} by the algorithm as $\orient u v$. (If $\ell_u = \ell_v$, the orientation of   $(u,v)$ will be determined
suitably by the algorithm.) /* partially maintained */
\end{invariant}

Our algorithm will maintain Invariants \ref{first}(a) and \ref{first}(b) at all times, as in the amortized algorithms \cite{BGS11,Sol16}.
On the other hand, we maintain Invariants \ref{first}(c) and \ref{first}(d) only partially.
Next, we make this statement   precise.

Once a matched vertex becomes free, its level will exceed -1 until the update algorithm handles it.
We say that such a vertex is \emph{temporarily free}, meaning that it is not matched to any vertex yet, but its level and out-degree remain temporarily as before.
From now on, we distinguish between \emph{free} vertices and \emph{temporarily free} vertices:
Free vertices are unmatched and their level is -1, while temporarily free vertices are unmatched and their level exceeds -1.
By making this distinction, Invariant \ref{first}(c) holds true as stated.
Moreover, combining it with Invariant \ref{first}(a), we obtain:
\begin{invariant} [Invariant \ref{first}'(c)] \label{1c}
Any vertex of level -1 is unmatched and has out-degree 0.
\end{invariant}
Note that Invariants \ref{first}(c) and \ref{1c} do not apply to temporarily free vertices.
In particular, there may be edges between temporarily free vertices, which are unmatched by definition, meaning that the maintained matching is not necessarily maximal.
The challenge is to guarantee that the number of temporarily free vertices is small with respect to the number of matched vertices, yielding an almost-maximal matching.

Temporarily free vertices are handled via data structures that we call \emph{schedulers}.
We distinguish between vertices that become temporarily free  due to the adversary and those due to the update algorithm itself.
For each level $\ell$, we   maintain a queue $Q_\ell$ of level-$\ell$ vertices that become temporarily free   due to the adversary,
and the vertices in $Q_\ell$ will be handled, one after another, via appropriate schedulers.
As vertices in the queues are being handled and are thus removed from the queues, other vertices may become temporarily free due to the adversary and thus join the queues.
We will need to make sure that the total number of vertices over the queues of all levels is in check at any point in time.
In contrast, there is no reason to add vertices that become temporarily free due to the update algorithm itself to the queues,
as the update algorithm controls the rate in which such vertices become temporarily free, hence it will create temporarily free vertices periodically
at a rate that matches the time needed to handle them, again via appropriate schedulers.
The various schedulers need to work together, without conflicting each other;
the exact way in which they work constitutes the heart of our algorithm, and is described first in Section \ref{sec22}, and then in more detail in subsequent sections.
The analysis of the schedulers is provided in Section \ref{schedulersAnal}, and our mechanism for resolving conflicts between them is presented in Section \ref{exceptions}.

A temporarily free vertex that is being handled by some scheduler is called \emph{active}, and the process of handling it,
which involves updating various data structures,  may be simulated over multiple update operations.
Therefore, there might be \emph{inconsistencies} in the data structures throughout this process concerning the active vertices.
In particular, an active vertex $v$ that rises or falls from level $\ell$ to level $\ell'$ is stored as a level-$\ell$ vertex in the data structures
of some of its neighbors  and as a level-$\ell'$ vertex in the data structures of its remaining neighbors.
Moreover, among the edges whose orientation needs to be ``flipped'' as a result of this rise or fall of $v$,
so as to satisfy Invariant \ref{first}(d), some have   performed the flip and the rest have not done so yet.
To account for these inconsistencies, we   hold a list of active vertices, denoted $Active$, and we will make sure that this list is of size $O(\log_\gamma n) = O(\log n)$ at any point in time.\footnote{Note that $\log_\gamma n = O(\log n / \log\log n)$, but we make no attempt here to optimize factors that are polynomial in $\log\log n$.}
(Although the number of temporarily free vertices should be small with respect to the number of matched vertices, it may be significantly larger than the number $O(\log n)$ of active vertices.)
By bounding the number of active vertices, we can \emph{authenticate} the updated information concerning active vertices efficiently;
this \emph{authentication process} is described in Section \ref{sec23}.
Our algorithm will maintain Invariant \ref{first}(d) for any edge $(u,v)$ with both $u$ and $v$ not in the $Active$ list, or in other words:
\begin{invariant} [Invariant \ref{first}'(d)] \label{1d}
Any edge $(u,v)$,  with $\ell_u > \ell_v$ and $u,v \nin Active$, is \emph{oriented}  as $\orient u v$.
\end{invariant}

Following \cite{Sol16}, for each vertex $v$, we maintain linked lists $\cN_v$ and $\cO_v$ of its neighbors and outgoing neighbors, respectively.  
The information about $v$'s incoming neighbors will be maintained via a more detailed data structure $\cI_v$:
A   hash table,  where each element corresponds to a distinct level $\ell \in \{-1,0,\ldots,\log_\gamma (n-1)\}$.
Specifically, an element $\cI_v[\ell]$ of $\cI_v$ corresponding to level $\ell$ holds a \emph{pointer} to the head of a non-empty linked list that
contains all incoming neighbors of $v$ with level $\ell$.
If that list is empty, then the corresponding pointer is not stored in the hash table.
Hence the total space over all hash tables is linear in the dynamic number of edges in the graph.

By Invariant \ref{1d},   any edge $(u,v)$, with $\ell_u > \ell_v$ and $u,v \nin Active$ is oriented as $\orient u v$.
The consequence is that for any such edge, $v \in \cO_u$ and $u \in \cI_v[\ell_u]$. In particular,
the data structure $\cI_v$ provides information on the levels of $v$'s incoming neighbors that do not belong to the $Active$ list.
It will not be in $v$'s \emph{responsibility} to maintain the
data structure $\cI_v$, but rather within the responsibility of $v$'s incoming neighbors.
On the other hand, no information whatsoever on the levels of $v$'s outgoing neighbors is provided by the data structure $\cO_v$.
In particular, to determine if $v$ has an outgoing neighbor at a certain level (most importantly at level -1, i.e., a free neighbor),
we need to scan the entire list $\cO_v$. On the other hand, $v$ has an incoming neighbor at a certain level $\ell$ iff
the corresponding list $\cI_v[\ell]$ is non-empty.
We keep mutual pointers between the elements in the various data structures: For any vertex $u$ and any outgoing neighbor $v$ of $u$,
we have mutual pointers between all elements $v \in \cO_u, u \in \cI_v[\ell_u],u \in N_v,v \in N_u$.
(We do not provide a description of the trivial  maintenance of these pointers for the sake of brevity.)
\subsection{\emph{Schedulers}: overview and invariants} \label{sec22}
Our algorithm will employ four different schedulers, each for a different purpose.
Each of these schedulers   consists of $O(\log_\gamma n) = O(\log n)$ sub-schedulers, a single sub-scheduler per level $\ell = 0,1,\ldots,\log_\gamma (n-1)$.
It is instructive to think of each sub-scheduler as running threads of execution,
and of its scheduler as synchronizing $O(\log n)$ threads, a single thread per level.
Each thread  executed by a level-$\ell$ sub-scheduler, hereafter \emph{level-$\ell$ thread}, will run in exactly the same amount of time
$T_\ell = \gamma^\ell \cdot \Theta(\log^4 n)$, by ``sleeping'' if finishing the execution prematurely.
We assume that the constant hiding in this $\Theta$-notation is sufficiently large, thus rendering $T_\ell$ sufficiently larger than the overall runtime of
any procedure described in the sequel that has a total runtime of $\gamma^\ell \cdot O(\log^4 n)$; as will be shown in the sequel, this   enables us to guarantee
that each level-$\ell$ thread finishes its run within the ``time slot'' of $T_\ell$ computation steps allocated to it.
\vspace{7pt}
\\
\noindent
{\bf Synchronization.~}
As mentioned, any level-$\ell$ thread runs in exactly the same amount of time $T_\ell$, for each level $\ell$.
However, to achieve a low worst-case update time, the execution of this thread is not carried out at once, but is rather carried out (hereafter, \emph{simulated}) over multiple update operations, carrying out (or simulating) a fixed number of computation steps per update operation.
We refer to that number of computation steps as a \emph{simulation parameter},
and for technical reasons we use
two simulation parameters, $\Delta := \Theta(\log^5 n /\eps)$ and $\Delta' = \Delta \cdot \gamma = \Delta \cdot \Theta(\log n)$.
Each of the two simulation parameters does not change with the level, and is not associated with a level-$\ell$ thread or with the sub-scheduler running it,
but rather with the corresponding scheduler.  
Each of the schedulers will use exactly one of these two simulation parameters and will stick to it throughout;
since $\Delta' / \Delta = \gamma = \Theta(\log n)$,
in this way we make sure that some schedulers will consistently work faster than others by a logarithmic factor, a property that will be useful for our analysis.
Note that the simulation parameters,   $\Delta$ or $\Delta'$, determine the number of update operations required to finish the execution of the thread,
$T_\ell / \Delta$ or $T_\ell / \Delta'$, respectively. We refer to this number as the \emph{(level $\ell$) simulation time}; unlike the simulation parameters,
which do not change with the level, the corresponding simulation times grow with each level by a factor of $\gamma$.
Therefore the simulation time of a thread depends not only on the respective scheduler but also on the respective level $\ell$,
and is associated with both the respective thread and the sub-scheduler running it.
A sub-scheduler or a thread with a lower (respectively, higher) simulation time than another is said to be \emph{faster} (resp., \emph{slower}) than it;
we may compare the simulation times of sub-schedulers or threads even if they are at different levels.
For any levels $\ell$ and $\ell'$ such that $\ell' > \ell$, a level-$\ell$ sub-scheduler may be either faster or at the same speed as any level-$\ell'$ sub-scheduler.
Since $T_{\ell'} / T_\ell = \gamma^{\ell' - \ell}$ and $\Delta' / \Delta = \gamma$,
any two such sub-schedulers are at the same speed if only if $\ell' = \ell+1$ and the simulation time of the level-$\ell$ sub-scheduler is $T_\ell / \Delta$
while the simulation time of the level-$\ell'$ sub-scheduler is $T_{\ell+1} / \Delta'$. In the complementary case, i.e., when either $\ell' \ge \ell+2$
or when the schedulers corresponding to those sub-schedulers have the same simulation parameter, the level-$\ell'$ sub-scheduler is slower than
 the level-$\ell$ sub-scheduler by at least a factor of $\gamma = \Theta(\log n)$ (i.e., the simulation time of the level-$\ell'$ sub-scheduler is higher than that of the level-$\ell$ scheduler by that factor).

The execution threads are  run by the various schedulers in a precise   \emph{periodic} manner so as to achieve the following \emph{nesting property}.
Viewing the time axis as a 1-dimensional line and the simulation times of the execution threads as intervals of this line, the intervals of any two threads are either disjoint or one of them is nested in the other; in what follows we may identify threads with the corresponding 1-dimensional intervals, and may henceforth say that two threads are disjoint or one of them is nested in the other.
Multiple threads may be nested in a single thread at a higher level. Moreover, multiple level-$\ell$ threads may be nested in a single level-$\ell$ thread, for any $\ell$,
but this may happen for at most $\gamma$ such threads, and only if their simulation time is $T_{\ell} / \Delta'$
while that single thread's simulation time is $T_{\ell} / \Delta$.

The description of the    schedulers is provided below. The schedulers
$\temp, \rise$ and $\shuffle$  have a simulation parameter of $\Delta'$, whereas the scheduler $\unmatch$ has a simulation parameter of $\Delta$,
and is thus slower than the other schedulers by a factor of $\gamma = \Theta(\log n)$.
Consider  any scheduler among the four, and denote its simulation parameter by $\tilde \Delta$, where $\tilde \Delta$ is either $\Delta$ or $\Delta'$.
While it may be instructive to view the $\log_\gamma (n-1)+1$ execution threads (over all levels) that this scheduler runs following every update operation as operating in parallel,  these threads are handled \emph{sequentially}.
It is technically useful   to handle these threads by decreasing order of simulation times, and thus by decreasing order of levels,
i.e., the $\log_\gamma(n-1)$-level thread is handled first, then the $\log_\gamma(n-1)-1$-level thread, etc., until the $0$-level thread.  
Following each update operation, the $\log_\gamma(n-1)$-level thread simulates $\tilde \Delta$ computation steps of its execution,
the $\log_\gamma(n-1)-1$-level thread simulates $\tilde \Delta$ computation steps of its own execution, and so on.
When any of these threads finishes its execution (sleeping if the execution has finished prematurely),
the corresponding sub-scheduler starts executing a new thread at that level, which again simulates $\tilde \Delta$ computation steps following each update operation.
Hence the total time spent by this scheduler following a single update operation is $\tilde \Delta \cdot (\log_\gamma (n-1)+1)$,
which is either $O(\log^6 n /\eps)$ or $O(\log^7 n /\eps)$, depending on whether $\tilde \Delta$ is $\Delta$ or $\Delta'$, respectively. Observe that this scheme gives rise to a worst-case update time of $O(\log^7 n /\eps)$, and this bound holds deterministically.

Observe that when any level-$\ell$ thread starts its execution, all lower level threads run by the same scheduler are also about to start their execution,
and they will finish their execution before the level-$\ell$ thread does. As for threads run by other schedulers, things are slightly more involved, as a lower level thread may have the same simulation time as that level-$\ell$ thread (if their simulation parameters are $\Delta'$ and $\Delta$, respectively).
We can generalize the above observation by noting that all lower level threads run by the four schedulers start their execution at the same update operation as the level-$\ell$ thread does,
and they will finish their execution before or at the same update operation as the level-$\ell$ thread does.  
This observation, which follows from the fact that we handle the threads by decreasing order of simulation times (and levels) and from the aforementioned nesting property of threads,
will play a central role  
in our mechanism for resolving potential conflicts between the various schedulers, described in Section \ref{exceptions}.


Each of the level-$\ell$ sub-schedulers will run a single level-$\ell$ thread  at any point in time.
Each level-$\ell$ thread will handle vertices of level at most $\ell$, but some of these threads will handle a super constant number of such vertices.
We will make sure to address this issue, and also show that any level-$\ell$ thread (including those that handle a super constant number of vertices) requires an overall time of $T_\ell$ to complete its execution.
\vspace{7pt}
\\
\noindent
{\bf 1st scheduler.~} The first scheduler $\temp$ handles all vertices that become temporarily free due to the adversary.  
More specifically, for each level $\ell = 0,1,\ldots,\log_\gamma (n-1)$, the corresponding sub-scheduler $\temp_\ell$
handles all vertices of $Q_{\ell}$, one after another.
The exact procedure for handling a temporarily free vertex $v$, $\evalf(v)$, is described in Section \ref{sec25}.  
Procedure $\evalf(v)$ will be executed by a single level-$\ell$ thread corresponding to $v$ that runs in an overall time of $T_\ell$,
simulating $\Delta'$ steps of this procedure following each update operation.   
The $\log_\gamma (n-1)+1$ execution threads (over all levels) executed by $\temp$ are handled sequentially.  
Note that these threads execute different calls of Procedure $\evalf$, which handle vertices at different levels.
As mentioned, 
the $\log_\gamma(n-1)$-level thread is handled first,
then the $\log_\gamma(n-1)-1$-level thread, etc., until the $0$-level thread.  
Following each update operation, the $\log_\gamma(n-1)$-level thread simulates $\Delta'$ steps of its own call of Procedure $\evalf$,
the $\log_\gamma(n-1)-1$-level thread simulates $\Delta'$ steps of its own call, and so on, hence the total time spent by   $\temp$ following a single update operation is
$\Delta' \cdot (\log_\gamma (n-1)+1) = O(\log^7 n /\eps)$.

By the same principle,  the total time spent by     $\rise$ and $\shuffle$ following a single update operation
will be bounded by $\Delta' \cdot (\log_\gamma (n-1)+1) = O(\log^7 n /\eps)$. On the other hand, $\unmatch$ has a simulation parameter of $\Delta$ rather than $\Delta'$,
hence the total time spent by this scheduler following a single update operation will be bounded by $\Delta \cdot (\log_\gamma (n-1)+1) = O(\log^6 n /\eps)$.
\vspace{7pt}
\\
\noindent
{\bf 2nd scheduler.~}
The second scheduler $\unmatch$ removes matched edges from the matching in a specific order.
As strange as it might seem, this strategy enables us to guarantee that the remaining matched edges are unlikely to be destroyed by the adversary.
More specifically, for each level $\ell$, the corresponding sub-scheduler $\unmatch_\ell$ removes level-$\ell$ edges from the matching,
one after another, in the following way.
Similarly to the amortized algorithms, each level-$\ell$ matched edge $e = (u,v)$ is sampled uniformly at random from $\Theta(\gamma^\ell)$ edges,
but there is a difference: While in the amortized algorithms the matched edge is sampled from precisely $\gamma^\ell$ edges,  here, for technical reasons,
we sample the matched edge from between $(1-\eps) \cdot \gamma^\ell$ and $\gamma^\ell$ edges.
(In the offline setting, we choose the edge that will be deleted last among those.)
We denote this edge set by $S(e)$, and refer to it as the \emph{sample space} (shortly, \emph{sample}) of edge $e$. As time progresses, some  edges of $S(e)$ may be deleted from the graph;
we denote by $S^*(e)$ the \emph{original} sample of $e$, with $(1-\eps) \cdot \gamma^\ell \le |S^*(e)| \le \gamma^\ell$, and by $S_t(e) = S(e)$ its sample \emph{remaining} at time $t$, omitting the subscript $t$ when it is clear from the context.
The goal of    $\unmatch_\ell$ is to guarantee that the samples of all level-$\ell$ matched edges will never reach $(1-2\eps) \cdot \gamma^\ell$:
\begin{invariant} \label{samp:in}
For any level-$\ell$ matched edges $e$ with $T_\ell / \Delta \ge 1$  and any $t$,  $|S_t(e)| > (1-2\eps) \cdot \gamma^\ell$.
\end{invariant}
To maintain this invariant, $\unmatch_\ell$  will always remove a matched edge of smallest remaining sample.
Observe that the samples are changed only due to edge removals,
and each edge removal reduces the size of at most two samples by one unit each.
Hence, it is easy to maintain the samples of all level-$\ell$ edges via a data structure that supports all the required operations, including the removal of a matched edge of smallest sample, in constant time; we do not describe this  data structure for the sake of brevity.
For each level-$\ell$ matched edge $e = (u,v)$ that is   removed by $\unmatch_\ell$, its two endpoints $u$ and $v$ become temporarily free,
and they are handled by appropriate calls to Procedure $\evalf$. 
More specifically, we execute Procedure $\evalf(u)$ and then $\evalf(v)$ by running a level-$\ell$ thread,  
which runs in an overall time of $T_\ell$, simulating $\Delta$ steps of execution following each update operation.
The intuition as to why $\unmatch_\ell$ is able to maintain Invariant \ref{samp:in} is the following.
(See Section \ref{sec32} for the formal argument.)
Since $T_\ell = \gamma^\ell \cdot \Theta(\log^4 n)$ and $\Delta = \Theta(\log^5 n/\eps)$,
the simulation time $T_\ell / \Delta$ of a thread run by $\unmatch_\ell$ (which designates the number of update operations needed for simulating its entire execution) is $\Theta(\eps(\gamma^\ell /\log n))$.
In other words, $\unmatch_\ell$ can remove a level-$\ell$ matched edge within $T_\ell /\Delta = \Theta(\eps(\gamma^\ell / \log n))$ adversarial update operations.
On the other hand, the expected number of adversarial edge deletions needed to turn a ``full'' level-$\ell$ matched edge $e$ (with sample $|S^*(e)| \ge (1-\eps) \cdot \gamma^\ell$)
into an ``under full'' edge (with sample  $\le (1-2\eps) \cdot \gamma^\ell$) is $\Omega(\eps \cdot \gamma^\ell)$.
Thus $\unmatch_\ell$ is faster than the adversary by at least a logarithmic factor, assuming $T_\ell / \Delta \ge 1$ (which holds when $\gamma^\ell = \Omega(\log n / \eps)$),
a property that suffices for showing that no edge is ever under full, or in other words,  the samples of all level-$\ell$ matched edges will always be in check.
This is the basic idea behind maintaining the validity of Invariant \ref{samp:in} in any level $\ell$ for which the simulation time satisfies $T_\ell / \Delta \ge 1$.
This invariant, in turn, guarantees that the adversary is unlikely to delete any particular edge from the matching,
using which we show (in Section \ref{sec:almost}) that the maintained matching is always almost-maximal with high probability.
The complementary regime of levels, namely, levels $\ell$ with simulation time satisfying $T_\ell / \Delta < 1$, is trivial and does not rely on Invariant \ref{samp:in},
as then the adversary does not make any edge deletion within the time required by a level-$\ell$ thread to complete its entire execution.
\ignore{
Sanity check -- suppose that at each update step we do $\log^4 n$ computational steps of the procedure.
But suppose that the runtime of the procedure is $x \cdot \log^2 n$, this means that we need $x/ \log^2 n$ update steps to carry out the procedure.
So we are actually faster than the adversary by a factor of $\log^2 n$.
What happens with the rising scheduler? Let's say that the rising thing works in the same way, so we do $\log^4 n$ computational steps of the procedure.
What matters is that we need $x/\log^2 n$ update steps to carry out the procedure, and within this time we give a penalty of $x$.
It's not true that the penalty is $x$, however. It is $x \cdot \log^2 n$. This is because the same level may get penalty due to every one of the
sub-schedulers. Moreover, a single sub-scheduler might effect all levels (because remember that we need to update the $\phi$ values in the appropriate level ranges).
But actually, this is too pessimistic. I think we can just look at any specific level and say that this level might be effected by any one of the schedulers.
The second aggregation over levels is not true because we also replace a logarithmic number of schedulers, and we ask now -- at any level, what is the damage?
}
\vspace{7pt}
\\
\noindent
{\bf 3rd scheduler.~}
Let $N_v(\ell)$ denote the set of neighbors of $v$ with level strictly lower than $\ell$, and write $\phi_v(\ell) = |N_v(\ell)|$.
For each vertex $v$, we will maintain the $\phi_v(\ell)$ values for all levels $\ell$ greater than the current level $\ell_v$ of $v$.
For any level $\ell \le \ell_v$, the corresponding value $\phi_v(\ell)$ will not be maintained, and the algorithm will have to compute it on the fly, if needed.
The algorithm of \cite{BGS11} maintains the invariant that $\phi_v(\ell) < \gamma^{\ell}$, for any $v$ and $\ell > \ell_v$.
(Recall that $\gamma$ is taken to be constant  in \cite{BGS11}, whereas here we take $\gamma$ to be $\Theta(\log n)$.)
The scheduler $\rise$ maintains the following relaxation of the invariant from \cite{BGS11},  
and it does so by raising vertices to higher levels in a specific order, as described next.
\begin{invariant} \label{risei}
For any vertex $v$ and any level $\ell > \ell_v$, $\phi_v(\ell) \le \gamma^{\ell} \cdot O(\log^2 n)$.
\end{invariant}
For each level $\ell$, the corresponding sub-scheduler $\rise_\ell$ is responsible for maintaining the invariant with respect to that level.
Whenever a new level-$\ell$ thread is initiated by $\rise_\ell$, it starts by \emph{authenticating} the $\phi_v(\ell)$ values over all vertices $v$ using the $Active$ list.
(The authentication process takes time $O(\log^2 n)$ to guarantee that all $\phi_v(\ell)$ values are up to date, and is described in Section \ref{sec23}.)
Then the thread picks a vertex $v$ whose $\phi_v(\ell)$ value is highest among all vertices with level lower than $\ell$.
Observe that each change of a $\phi$ value is either an increment or a decrement of one unit.
Hence it is easy to maintain   the level-$\ell$ $\phi$ values of all relevant vertices
via a data structure that supports all the required operations, including the extraction of a vertex with highest level-$\ell$ $\phi$ value,
 in constant time; we do not describe this  data structure for the sake of brevity.
These two steps (authenticating the $\phi_v(\ell)$ values and picking a vertex of maximum value, thus extracting it from the data structure) can be implemented within time $O(\log^2 n)$, and are therefore carried out by the thread ``instantly'', i.e., without simulating their execution over subsequent update operations.
The same execution thread continues to removing $v$'s old matched edge $(v,w)$ (if exists) from the matching, and raises $v$ to level $\ell$ by
executing Procedure $\setlvl(v,\ell)$, whose description is provided in Section \ref{sec24}.
The execution of this procedure, however, cannot be carried out instantly, so it is simulated over multiple update operations; we make sure to simulate
$\Delta'$ execution steps of this procedure following each update operation.
Then the same execution thread handles the two temporarily free vertices $v$ and $w$ using Procedure $\evalf$, which is again simulated over multiple update operations.
That is, the same thread continues to executing the call to $\evalf(v)$ and then the call to $\evalf(w)$, simulating $\Delta'$ execution steps following each update operation.
\vspace{7pt}
\\
\noindent
{\bf 4th scheduler.~}
The fourth scheduler $\shuffle$ removes matched edges from the matching uniformly at random. By working sufficiently faster than some of the other schedulers,
it forms a dominant part of the algorithm, using which we   basically show (in Section \ref{shuffling}) that it provides a near-uniform random shuffling of the matched edges.
This random shuffling facilitates  the proof of the assertion that the adversary is unlikely to delete any particular edge from the matching.  
More accurately, to prove this assertion, it suffices that $\shuffle$ would be sufficiently faster than $\unmatch$, for technical reasons that will become clear later on.
For each level $\ell$, the corresponding sub-scheduler $\shuffle_\ell$
will always pick a matched edge uniformly at random among all remaining level-$\ell$ edges, and will remove it from the matching.
As with $\unmatch_\ell$, for each level-$\ell$ matched edge $e = (u,v)$ that is   removed by $\shuffle_\ell$, its two endpoints $u$ and $v$ become temporarily free,
and they are handled by appropriate calls to Procedure $\evalf$. 
Moreover, as before, we execute these calls (to $\evalf(u)$ and then to $\evalf(v)$)
by running a level-$\ell$ thread,  
which runs in an overall time of $T_\ell$. The difference is that now we simulate $\Delta'$ (rather than $\Delta$) execution steps following each update operation,
which ensures that $\shuffle$ is faster than $\unmatch$ by a logarithmic factor.
We only need to apply the shuffling in levels $\ell$ for which the simulation time satisfies $T_\ell / \Delta \ge 1$, as
in the complementary regime ($T_\ell / \Delta < 1$) the adversary does not make
any edge deletion within the time required by a level-$\ell$ thread to complete its entire execution, and then a random shuffling is redundant.
\subsection{The authenticating process} \label{sec23}
By invariant \ref{1d}, any edge between two non-active vertices $v$ and $w$ is oriented towards the lower level endpoint,
i.e.,  the corresponding data structures of $v$ and $w$ are updated with the right levels of $v$ and $w$,
thus if $\ell_v > \ell_w$, then $w \in \cO_v$ and $v \in \cI_w[\ell_v]$.
On the other hand, if $v$ or $w$ are active, then these data structures may not be updated with the right levels yet.
The fact that the data structures are outdated should not be viewed as an error or an exception of the algorithm,
but rather as an inherent consequence to the way our algorithm works, by simulating the execution of the various procedures over multiple update operations.
We should henceforth handle this issue of having outdated data structures  in a systematic manner.
In particular, for every vertex $w$ that changes its level, hereafter, in the process of either \emph{falling} to a lower level or \emph{rising} to a higher level (see Section \ref{sec24} for more details), the data structures of $w$ and some of its neighbors might not have been updated regarding $w$'s new level.
Noting that any vertex that changes its level (i.e., falls or rises) must be active to do so, we can efficiently authenticate the updated levels of vertices  whenever needed
by monitoring the   short $Active$ list.

Consider any procedure used by our update algorithm that is carried out by a thread whose execution is simulated over multiple update operations.
Throughout the execution of that procedure, many vertices may change their level, possibly multiple times.
For this reason, just before the end of the execution, we perform an authentication process that takes $O(\log^2 n)$ time.
More specifically, if the procedure handles vertex $v$, then we scan the entire $Active$ list, looking for neighbors of $v$.
For any such neighbor $w$ that we find, we update the relevant data structures accordingly within $O(\log n)$ time.
(The bottleneck is to update a possibly logarithmic number of $\phi_v(\cdot)$ values, for each neighbor $w$ of $v$ that has changed its level.)
This may not be enough, however, as some neighbors of $v$ in the $Active$ list may leave it prior to this scan of $v$.
Consequently, when any vertex $z$ leaves the $Active$ list, we update the relevant data structures of all its neighbors $w$ currently in the list regarding $z$.
Since the $Active$ list is of size $O(\log n)$ and as we spend  $O(\log n)$ time per vertex $w \in Active$ for updating the  data structures of $w$
regarding $z$, the time needed for this part of the authentication process, and thus also for the entire process, is $O(\log^2 n)$.

While any authentication process takes time $O(\log^2 n)$, the simulation parameters $\Delta = \Theta(\log^5 n/\eps)$ and $\Delta' = \Theta(\log^6 n/\eps)$,
which is the time reserved for any execution thread run by $\unmatch$ and the remaining schedulers following a single update operation, respectively, is much larger.
This allows us to run an authentication process following any update operation and by any of the threads (when and where needed) instantly, i.e., without simulating its execution over subsequent update operations, while increasing the
worst-case update time of the algorithm by a negligible factor.
Recall that these threads run sequentially, i.e., when any execution thread runs,   the other ones are idle.
During this time, the running thread is the only one that may access and modify the $Active$ list.
Hence, the $Active$ list remains intact during any authentication process,
which is crucial for the validity of this process.


As mentioned, the outdated values stored in the data structures lead to  \emph{inconsistencies}, particularly regarding the levels of vertices.
These inconsistencies may lead to conflicts between the various schedulers, e.g., some level-$\ell$ vertex $v$ may choose
$w$ as its mate uniformly at random among all its neighbors of level lower than $\ell$, while $w$ is in the process of rising to level $\ge \ell$ and should thus not be chosen as a mate for $v$.
Since all {inconsistencies} in the data structures concern active vertices and as there are just few of those,
we can detect those inconsistencies effectively as part of the authentication process. There is a difference, however, between \emph{detecting} an inconsistency and \emph{resolving} it,
as the latter may require a long process, which needs to be simulated over multiple update operations. During this time interval of resolving an inconsistency,
a conflict between schedulers may arise as a result of that inconsistency.
The authentication process is not aimed at resolving \emph{all} potential conflicts between the various schedulers,
but should rather be viewed as a tool for detecting them and minimizing their number and variety.
In Section \ref{exceptions} we describe and analyze our mechanism  for resolving all potential conflicts between the various schedulers,
which heavily relies on the authentication process. Although this mechanism is technically elaborate,
it does not overcome a major conceptual challenge, but rather a minor technicality;
indeed, all potential conflicts concern only $O(\log n)$ vertices, namely, the active ones, and it is not too difficult to resolve conflicts
that concern $O(\log n)$ vertices with a polylogarithmic update time.

\section{The Update Algorithm, Part II: The Underlying Procedures} \label{procedures}
\subsection{Procedure $\setlvl(v,\ell)$} \label{sec24}
Whenever the update algorithm examines a vertex $v$, it may need to re-evaluate its level.
After the new level $\ell$ is determined, the algorithm calls
Procedure $\setlvl(v,\ell)$. 
(The exact way in which the new level $\ell$ of $v$ is determined is not part of this procedure; it is either determined by $\rise$
as described in Section \ref{sec22} or by the procedures described in Sections \ref{sec25} and \ref{sec26}.)
Note that setting the level of $v$ to $\ell$ can be done instantly.
The task of Procedure $\setlvl(v,\ell)$ is to update the relevant data structures as a result of this level change.
This process involves updating the sets of outgoing and incoming neighbors of $v$ and some of its neighbors (which can be viewed
as flipping the respective edges) so as to maintain Invariant \ref{1d} (or Invariant \ref{first}(d)), and also updating the $\phi$ values of $v$ and its relevant neighbors.
We refer to this process as the \emph{falling} or \emph{rising} of $v$, depending on whether $\ell < \ell_v$ or $\ell > \ell_v$, respectively.
(If $\ell = \ell_v$, Procedure $\setlvl(v,\ell)$ does not do anything.)
We remark that the rising/falling of vertex $v$ to level $\ell$ is a (possibly long) process that does not end until the corresponding call to $\setlvl(v,\ell)$ finishes its execution,
and the level of $v$ is viewed as its destination level $\ell$ starting from the beginning of the rising/falling process.

Procedure $\setlvl(v,\ell)$, which carries out the falling/rising process of $v$ from level $\ell_v$ to level $\ell$, is invoked by our update algorithm either
by $\rise_\ell$, in which case we have $\ell > \ell_v$, or by Procedure $\evalf$ that is described in Section \ref{sec26}.
We argue that any call to $\setlvl(v,\ell)$ is   executed by a level-$\hat \ell$ thread, where $\hat \ell \ge \tilde \ell := \max\{\ell_v,\ell\}$.
This assertion is immediate if the call to $\setlvl$ is due to $\rise_\ell$, since $\rise_\ell$ may only run level-$\ell$ threads.
If the call to $\setlvl$ is due to Procedure $\evalf$, then the assertion follows from Corollary \ref{threadeval}; see Section \ref{behavior} for details.

Note that the thread that executes Procedure $\setlvl(v,\ell)$ simulates multiple execution steps of this procedure
following each update operation; the exact number of execution steps is either $\Delta$ or $\Delta'$, depending on the scheduler that runs this thread,
but in any case the execution of this thread is simulated over multiple update operations.
We next describe Procedure $\setlvl(v,\ell)$, disregarding the fact that it is being simulated over multiple update operations.
In this description we also disregard the fact that some neighbors of $v$ may be active at the beginning of the procedure's execution or become active throughout the execution.
Then we address the technicalities that arise from these facts.  

\paragraph{A high-level description.~}
Procedure $\setlvl(v,\ell)$ starts by storing the old level $\ell_v$ of $v$ in some temporary variable $\ell^{old}_v$ and setting
the new level of $v$ to $\ell$, i.e., $\ell_v = \ell$. (Thus the level of $v$ is set as its destination level from the beginning of the rising/faling process.)
Then the procedure updates the outgoing neighbors of $v$ about $v$'s new level.
Specifically, we scan the entire list $\cO_v$, and for each vertex $w \in \cO_v$, we move $v$ from $\cI_w[\ell^{old}_v]$ to $\cI_w[\ell]$.

Suppose that $\ell < \ell^{old}_v$. In this case the level of $v$ is decreased by at least one, i.e., $v$ is falling.
As a result, we need to update the values $\phi_v(j)$, for all $\ell + 1 \le j \le \ell^{old}_v$;
this can be carried out by scanning the list $\cO_v$, as all non-active neighbors of $v$ with level at most $\ell^{old}_v-1$ are in $\cO_v$.
Moreover, we need to update the $\phi$ values of the relevant neighbors of $v$.  
Specifically, we scan the entire list $\cO_v$, and for each vertex $w \in \cO_v$ with $\ell_w < \ell^{old}_v$,
we increment $\phi_w(j)$ by 1, for all $\max\{\ell,\ell_w\} + 1\le j \le \ell^{old}_v$.
We also need to flip the outgoing edges of $v$ towards vertices of level between $\ell+1$ and $\ell^{old}_v$ to be incoming to $v$.
Specifically, we scan the   list $\cO_v$,
and for each vertex $w \in \cO_v$ such that $\ell+1\le \ell_w \le \ell^{old}_v$, we
perform the following operations:
Delete $w$ from $\cO_v$, add $w$ to  $\cI_v[\ell_w]$, delete $v$ from $\cI_w[\ell]$, and add $v$ to $\cO_w$.

If $\ell > \ell^{old}_v$, the level of $v$ is increased  by at least one, i.e., $v$ is rising.
As a result,
we need to flip $v$'s incoming edges from vertices of level between $\ell^{old}_v$ and $\ell-1$ to be outgoing of $v$.
Specifically, for each non-empty list $\cI_v[i]$, with $\ell^{old}_v \le i \le \ell-1$,
and for each vertex $w \in \cI_v[i]$, we perform the following operations: Delete $w$ from $\cI_v[i]$, add $w$ to  $\cO_v$, delete $v$ from $\cO_w$, and add $v$ to $\cI_w[\ell]$.
Note that we do not know for which levels $i$ the corresponding list is non-empty; the time overhead needed to verify this information is $O(\ell)$.
We also update the $\phi$ values of the relevant neighbors of $v$.  
Specifically, we scan the updated list $\cO_v$, and for each vertex $w \in \cO_v$ with $\ell_w < \ell$,
we decrement $\phi_w(j)$ by 1, for all $\max\{\ell^{old}_v,\ell_w\} + 1 \le j \le \ell$.


The following observation is implied by Invariant \ref{1d} and the high-level description of this procedure.
\begin{observation} \label{basicscan}
Any non-active neighbor of $v$ scanned by Procedure $\setlvl(v,\ell)$ has level at most $\tilde \ell$.   
Moreover, at the beginning of the procedure's execution, all non-active neighbors of $v$ with level less than $\ell_v = \ell^{old}_v$ are outgoing of $v$
and all non-active outgoing neighbors of $v$ at that time  have level at most $\ell_v$.
\end{observation}

\paragraph{Zooming in.~}
Since the execution of this procedure is simulated over multiple update operations,
it is possible that some neighbors of $v$ are falling and/or rising throughout this time interval, possibly multiple times.
Consequently, we need to apply the authentication process described in Section \ref{sec23}.
Although this process was described in Section \ref{sec23}, we find it instructive to repeat the details of this process that are relevant to Procedure $\setlvl$, for concreteness.
We authenticate the values of $v$'s neighbors that the procedure  scans using the $Active$ list,
as well as scan the entire $Active$ list at the end of this procedure, to make sure that no relevant neighbor is missed.
The level of an active vertex (as stored in the $Active$ list) is viewed as its \emph{destination} level (i.e., the level to which it falls/rises).
In addition, we update the vertices that belong to the $Active$ list about the new level of $v$, when $v$ is removed from that list (after the procedure's execution terminates).
Recall that we do that not just for $v$, but rather for every vertex that leaves the $Active$ list, which guarantees that the data structures of any active vertex
$x$ are always updated regarding any neighbor that changes its level throughout the time interval during which $x$ is active.
In particular,  Procedure $\setlvl(v,\ell)$ may not be able to handle properly neighbors of $v$ that change their level throughout the procedure's execution. 
This is why when any such neighbor $w$ leaves the $Active$ list, we update the data structures of (the currently active) $v$ regarding $w$'s new level.

Throughout the execution of Procedure $\setlvl(v,\ell)$, $v$ may acquire new neighbors at levels at most $\tilde \ell$,
either due to neighbors falling to such levels or due to adversarial edge insertions.
Similarly, $v$ may lose neighbors at such levels, either due to neighbors rising to levels higher than $\tilde \ell$ or due to adversarial edge deletions.
As each adversarial edge update occurs, we make sure to update the data structures of the two endpoints in $O(\log n)$ time.
If an edge $(v,w)$ is added/removed to/from the graph throughout the execution of this procedure, we update the relevant data structures according to the \emph{destination} level $\ell$ of $v$ rather than the old one;
if $w$ is also in the process of rising/falling, we update the data structures according to the destination level of $w$ rather than the old one.
Focusing on the data structures of $v$, we will store the new neighbors of $v$ in temporary data structures $\cO'_v$ and $\cI'_v$ throughout the execution
of Procedure $\setlvl(v,\ell)$, and merge them with the old data structures at the end of the execution.
In this way we can avoid scanning the new neighbors of $v$ throughout the procedure's execution, which is useful
for bounding the total runtime of this procedure.
We deal with new/old falling/rising neighbors of $v$ just like we deal with new/old neighbors due to adversarial edge insertions/deletions.
In particular, a neighbor $w$ of $v$ that falls to level at most $\tilde \ell$ will be stored (when needed) in the aforementioned temporary data structures $\cO'_v$ and $\cI'_v$;
note that the relevant data structures of $v$ are not updated regarding $w$ as part of Procedure $\setlvl(v,\ell)$,
but rather at the end of the execution of Procedure $\setlvl(w,\cdot)$, as $w$ leaves the $Active$ list.
In addition, at the end of the execution of Procedure $\setlvl(v,\ell)$, we need to update the data structures of $v$ regarding $v$'s new neighbors that appear in the $Active$ list,
but this update is done as part of the authentication process, which does not distinguish between new and old neighbors of $v$.
As mentioned in Section \ref{sec23}, the runtime of the authentication process is bounded by $O(\log^2 n)$, and this bound holds independently of the number of neighbors that $v$ acquires (or loses) during the execution of  Procedure $\setlvl(v,\ell)$.

The correctness of Procedure $\setlvl(v,\ell)$ follows from the description above and Observation \ref{basicscan}.

Recall that $\tilde \ell = \max\{\ell_v,\ell\}$, where $\ell_v = \ell^{old}_v$ denotes the level of $v$ just before the execution of Procedure $\setlvl(v,\ell)$ starts.
We next analyze the total runtime of Procedure $\setlvl(v,\ell)$, denoted $L_{\tilde \ell}(v)$, for a vertex $v$ that falls or rises from level $\ell_v$ to level $\ell$.
\begin{lemma}  \label{setlevelTime}
Denote by $\phi_v(\tilde \ell+1)$ the number of $v$'s neighbors of level lower than $\tilde \ell +1$ at the beginning of the execution of this procedure.
Then $L_{\tilde \ell}(v) = O((\phi_v(\tilde \ell+1) + \log n) \cdot \log n)$.
\end{lemma}

\begin{proof}
First, recall that the authentication process takes $O(\log^2 n)$ time.  
Also, the time needed for updating the values $\phi_v(j)$ in the case that $v$ falls to level $\ell$, for all $\ell + 1 \le j \le \ell_v$, is $O(\log n)$.

By Observation \ref{basicscan}, any non-active neighbor of $v$ scanned by the procedure  has level at most $\tilde \ell$.
We may restrict our attention to $v$'s neighbors of level at most $\tilde \ell$, since any time spent by the procedure
for handling active neighbors of $v$ of higher level is encapsulated within the authentication process.
Note that only $O(\log n)$ neighbors of $v$ at the beginning of the procedure's execution may be active.
For each neighbor $w$ of $v$ of level at most $\tilde \ell$ (active or not),
the procedure spends at most $O(1)$ time for updating the appropriate data structures $\cO_v,\cI_v,\cO_w,\cI_w$,
and at most $O(\log n)$ time for updating the relevant $\phi_w$ values.  
This, however, does not imply that the procedure's runtime is  $O((\phi_v(\tilde \ell + 1) + \log n)\cdot \log n)$,
as the set of neighbors of $v$ is not static,
 but rather changes dynamically throughout the procedure's execution.
In particular, the corresponding set $N_v(\tilde \ell+1)$ (of $v$'s neighbors of level lower than $\tilde \ell+1$) is not static.


A vertex $w$ joins $N_v(\tilde \ell+1)$ in one of two ways, the first is due to adversarial edge insertions.
Whenever an edge $(v,w)$ adjacent to $v$ is added to the graph, the data structures are updated appropriately.
This update of the data structures is not part of Procedure $\setlvl(v,\ell)$, but rather part of the procedure that handles the insertion of edge $(v,w)$,
namely Procedure $\ins$,  described in Section \ref{sec25}.


The second way for a vertex $w$ to join $N_v(\tilde \ell+1)$ is by
falling from level higher than $\tilde \ell$ to level at most $\tilde \ell$.
It is possible that $w$ changes its level throughout the execution of this procedure multiple times.
For every such change in level except for maybe the last, the data structures of $v$ are updated as part of the respective calls to Procedure $\setlvl(w,\cdot)$,
and more concretely, each time $w$ is removed from the $Active$ list throughout the execution of Procedure $\setlvl(v,\ell)$, the data structures of $v$ are updated accordingly.
Those data structures may be outdated only if $w$ belongs to the $Active$ list at the end of the execution of Procedure $\setlvl(v,\ell)$.
For this reason we scan the entire $Active$ list at that stage, spending $O(\log n)$ time per active neighbor of $v$ for updating the relevant data structures.
This update of the data structures is part of the authentication process, whose cost was already taken into account.

When a vertex $w$ leaves $N_v(\tilde \ell+1)$, the data structures of $v$ (either the old or the temporary ones) are updated accordingly.
As with vertices that join $N_v(\tilde \ell+1)$, this update of the data structures is not part of Procedure  $\setlvl(v,\ell)$.
If $w$ leaves $N_v(\tilde \ell+1)$ due to an edge deletion, this update is part of the procedure that handles the deletion of edge $(v,w)$,
namely Procedure $\del$,  described in Section \ref{sec25}.
The second way for a vertex $w$ to leave $N_v(\tilde \ell+1)$ is by
rising from level at most $\tilde \ell$ to level higher than $\tilde \ell$,
in which case the data structures of $v$ are updated either as part of the respective calls to Procedure $\setlvl(w,\cdot)$
or as part of the authentication process.

When a vertex $w$ joins or leaves $N_v(\tilde \ell+1)$, the temporary data structures of $v$, namely $\cO'_v$ or  $\cI'_v$, are updated accordingly;
if $w$ joins $N_v(\tilde \ell+1)$,   these data structures are updated according to the new level of $w$.
Also, the data structures of each such neighbor are updated according to the destination level $\ell$ of $v$.
Hence there is no need for Procedure $\setlvl(v,\ell)$ to update the temporary data structures of $v$ nor to update the respective data structures of vertices
belonging to the temporary data structures of $v$, and so the only cost due to vertices that join or leave $N_v(\tilde \ell+1)$ incurred by this procedure is that of merging the old data structures of $v$ with the temporary ones.
Merging $\cO_v$ with $\cO'_v$ takes constant time, whereas merging $\cI_v$ with $\cI'_v$ takes $O(\log n)$ time,
since we merge here two hash tables (or we may create a new one instead), and moreover, we need to merge the respective lists $\cI_v[j]$ and $\cI'_v[j]$ one by one, where $j$ may range from 1 to $\tilde \ell$.
Thus the extra cost due to vertices that join or leave $N_v(\tilde \ell+1)$ is $O(\log n)$.

\ignore{
 which requires at least the same amount
of time as that required by the level-$\tilde \ell$ thread that executes Procedure $\setlvl(v,\ell)$.
The number of different threads running at any point in time is bounded by $O(\log n)$,
at most a single thread for each sub-scheduler at each of the $\log_\gamma (n-1)$ levels. (Recall that there are three different sub-schedulers at each level.)
Since each thread
Consequently,
}

Summarizing, we have shown that $L_{\tilde \ell}(v) = O((\phi_v(\tilde \ell+1) + \log n) \cdot \log n)$.
\QED
\end{proof}

The following observation is implied by the description of this procedure.
\begin{observation} \label{basicscan2}
At the end of the execution of Procedure $\setlvl(v,\ell)$, Invariant \ref{1d} holds with respect to all edges adjacent to $v$ that lead to non-active neighbors of $v$.
\end{observation}

\subsection{Procedures $\ins(u,v)$ and $\del(u,v)$} \label{sec25}
Following an edge insertion $(u,v)$, we apply Procedure $\ins(u,v)$.
Besides updating the relevant data structures in the obvious way within $O(\log n)$ time,
this procedure matches between $u$ and $v$ if they are both at level -1, otherwise it leaves them unchanged.
By Invariant \ref{1c}, any vertex at level -1 is unmatched.
Note that unmatched vertices whose level exceed -1, namely temporarily free vertices, are not matched by this procedure.
Matching $u$ and $v$ involves setting their level to 0 by making the calls to $\setlvl(u,0)$ and $\setlvl(v,0)$; this guarantees that Invariants \ref{first}(a) and \ref{first}(b) will continue to hold. Note also that no matter how the new edge $(u,v)$ is oriented, Invariant \ref{1d} will continue to hold.
The other invariants also continue to hold.
Invariant \ref{risei} implies that both $\phi_u(1)$ and $\phi_v(1)$ are bounded by $\gamma \cdot O(\log^2 n) = O(\log^3 n)$,
hence the runtime of these calls is at most $O(\log^4 n)$ by Lemma \ref{setlevelTime}.
As the runtime of this procedure is $O(\log^4 n)$, we can complete its entire execution well within the time reserved for a single update operation, namely, the worst-case update time of the algorithm, $O(\log^7 n / \eps)$.

Following an edge deletion $(u,v)$, we apply Procedure $\del(u,v)$.  
If edge $(u,v)$ does not belong to the matching,   we only need to update the relevant data structures.
In this case the runtime of this procedure will be $O(\log n)$, and we can complete its execution well before the next update step starts.
If  edge $(u,v)$ is matched, both $u$ and $v$ become temporarily free, and they are inserted to the
appropriate queue $Q_{\ell_u}$; recall that $\ell_u = \ell_v$ by Invariant \ref{first}(b).
Note that all invariants continue to hold.
The sub-scheduler $\temp_\ell$ makes sure to handle $u$ and $v$ (by making the calls to $\evalf(u)$ and $\evalf(v)$, as described in Section \ref{sec22}), one after another, after handling all the preceding vertices in the queue.
Note that until each of them is handled, its level will exceed $-1$ and its out-degree may be positive.
%
\subsection{Procedure $\evalf(v)$} \label{sec26}
This procedure handles a temporarily free vertex $v$, and is first invoked by the various schedulers as described in Section \ref{sec22}, but then also recursively.
We first provide a high-level overview of this procedure, and later zoom in on the parts that require further attention.

\paragraph{A high-level description.~}
The procedure starts by computing the highest level $\ell, 0 \le \ell \le \ell_v$, where $\phi_v(\ell) \ge \gamma^{\ell}$,
as well as the corresponding vertex set $N_{\ell}(v)$ of $v$  in order to randomly sample a neighbor $w$ of level lower than $\ell$ as the new mate of $v$.
The sampling is not done from the entire set $N_{\ell}(v)$, but rather from its subset $N'_\ell(v)$ of all non-active vertices in $N_{\ell}(v)$.
Specifically, if $|N'_\ell(v)| > \gamma^{\ell}$,   we sample from $\gamma^{\ell}$ arbitrary vertices of $N'_\ell(v)$; otherwise, we   sample from the entire set $N'_\ell(v)$.
This sampling is done only for levels $\ell$ satisfying $T_\ell / \Delta \ge 1$, which are the levels for which Invariant \ref{samp:in} is  maintained.
For such levels we have $\gamma^\ell = \Omega(\log n / \eps)$, and since $|N_{\ell}(v)| \ge \gamma^{\ell}$ and the number of active vertices is $O(\log n)$, it follows that only an $O(\eps)$-fraction of the vertices of $N_{\ell}(v)$
may be active. Appropriate scaling thus yields $|N'_\ell(v)| \ge (1-\eps) \cdot \gamma^{\ell}$, so we sample from between $(1-\eps) \cdot \gamma^{\ell}$
and $\gamma^{\ell}$ vertices, as required.

The complementary regime of levels, i.e., levels $\ell$ satisfying $T_\ell / \Delta < 1$, is trivial,
as the simulation time is either $T_\ell / \Delta$ or $T_\ell / \Delta'$ (depending on the thread executing this procedure),
each of which is smaller than 1. In this case we do not rely on Invariant \ref{samp:in}, and a naive deterministic treatment suffices.
(See Section \ref{lowlevels} for details.)

In order to match $v$ with $w$, we first delete the old matched edge $(w,w')$ on $w$ (if exists), thus rendering $w'$ temporarily free.
Second, we let $v$ and $w$ fall and rise to the same level $\ell$, respectively, by calling to $\setlvl(v,\ell)$ and $\setlvl(w,\ell)$.
(Note that the random sampling of $w$ was intentionally done prior to making these calls to $\setlvl$.)
We then match $v$ with $w$, thus creating a new level-$\ell$ matched edge, which satisfies the validity of Invariant \ref{first}(b).
Note that Invariants \ref{first}(a) and \ref{1d} also continue to hold.
Finally, assuming $w$ was previously matched to $w'$, we handle $w'$ recursively by calling to $\evalf(w')$.

In the degenerate case that no level $\ell$ as above exists, we have $\phi_v(0) = 0$, i.e.,  $v$ does not have any neighbor at level -1.
In this case $v$ remains free, and we make the call to $\setlvl(v,-1)$.
Note that Invariants \ref{first}(c) and \ref{1c} continue to hold.

In general, it is easy to verify that this procedure does not invalidate any invariant, disregarding Invariants \ref{samp:in} and \ref{risei},
whose validity is proved in Section \ref{schedulersAnal} as part of the analysis of the schedulers. 

By the description of the update algorithm, this procedure is executed by a level-$\ell_v$ thread, where $\ell_v$ is $v$'s level at the beginning of the procedure's execution.
The same thread is used also for the recursive call $\evalf(w')$, and for all subsequent recursive calls.
While the level of the thread executing this procedure matches the level of $v$, the vertex at the top recursion level,
we show in Section \ref{behavior} that it exceeds the levels of vertices handled by   subsequent recursive calls.  
This thread runs in an overall time of $T_{\ell_v}$, simulating $\Delta$ or $\Delta'$ execution steps following each update operation (depending on the sub-scheduler running it).

\paragraph{Zooming in.~}
Recall that none of the values $\phi_v(\ell)$ and vertex sets $N_v(\ell)$, for $0 \le \ell \le \ell_v$, is maintained by the update algorithm,
and computing them is a process that is simulated over multiple update operations.
By the time we finish this process, some of the scanned neighbors of $v$ may have different levels than those they had at the time they were scanned.
Also, $v$ may acquire new neighbors of level at most $\ell_v$ and lose others.
(Recall that we had similar problems with Procedure $\setlvl$.)
By Invariant \ref{1d}, when the execution of this procedure starts (i.e., just after $v$ becomes temporarily free),
every non-active neighbor of $v$ at level lower than $\ell_v$ is an outgoing neighbor of $v$.
We thus restrict our attention to the outgoing neighbors of $v$ in order to compute ``estimations'' for the values $\phi_v(j)$ and vertex sets $N_v(j)$,
for all $j \le \ell_v$. Once we finish computing these estimations, which might be simulated over multiple update operations,
we make these estimations accurate by running the authentication process described in Section \ref{sec23}.
Recall that the authentication process also considers vertices that belong to the $Active$ list at that stage.
It enables us to update the values $\phi_v(j)$ and vertex sets $N_v(j)$ in $O(\log^2 n)$ time, for $j \le \ell_v$,
and so this update can be carried out within the time reserved for a single update operation, namely, either $\Delta$ or $\Delta'$ depending on the
thread executing the procedure, during which the $Active$ list remains intact. We then compute in time $O(\log n)$ the
 highest level $\ell, 0 \le \ell \le \ell_v$, where $\phi_v(\ell) \ge \gamma^{\ell}$.
 Moreover, the same amount of time $O(\log n)$ suffices for computing the vertex set $N'_v(\ell)$.
Indeed, this set can be obtained by pruning the active vertices from the previously computed set $N_v(\ell)$,
which can be carried out within time linear in the size of the $Active$ list if we keep mutual pointers between the elements of $Active$ and $N_v(\ell)$.
\ignore{
A special attention should be given to vertices that belong to the $Active$ list during the authentication process.
For every such vertex that is falling or rising between levels $\ell'$ and $\ell''$, when updating the $\phi_v(\cdot)$ values and $N_v(\cdot)$ vertex sets,
its level is viewed as $\max\{\ell',\ell''\}$, i.e., the higher level among the two.
This stands in contrast to what we did in Procedure $\setlvl$,
where we viewed the level of a vertex that is falling or rising as its destination level. So the two procedures behave the same for rising vertices,
but behave differently for falling vertices. 
The reason we view a falling vertex as being in its old level is because we do not want vertices at intermediate levels to choose it
as their random mate. Such vertices are being handled by execution threads at lower levels than the one that handles the falling vertex,
and they may have to wait too much time for the vertex to complete its falling.
Note, however, that this is just a temporary inconsistency, which ends whenever the vertex finishes falling.
By being aware of this type of temporary inconsistencies, we can make sure to detect and resolve them on the fly.
}
Following similar lines to those in the proof of Lemma \ref{setlevelTime},
the runtime of this  process
is  $O(\phi_v(\ell_v+1) + O(\log n)) \cdot O(\log n)$,
where $\phi_v(\ell_v+1)$ denotes the number of $v$'s neighbors of level lower than $\ell_v + 1$ at the beginning  of the execution of this procedure.
By Invariant \ref{risei}, at any   time $\phi_v(\ell_v+1) = \gamma^{\ell_v+1} \cdot O(\log^2 n) = \gamma^{\ell_v} \cdot O(\log^3 n) $,
hence the runtime of this process is  $O(\gamma^{\ell_v} \cdot \log^4 n)$.

For every vertex $w$ that joins or leaves $N_v(j)$ throughout the procedure's execution, for $j \le \ell_v + 1$, we make sure to update the
values $\phi_v(j)$ and vertex sets $N_v(j)$ accordingly.
However, all such updates are not done as part of Procedure $\evalf(v)$, but rather
as part of the respective calls to Procedures $\ins(v,w)$ or $\del(v,w)$ in the case of adversarial edge insertions and deletions, respectively,
or as part of the respective calls to Procedure $\setlvl(w,\cdot)$.

Note that setting the level of $v$ to $\ell$ by calling to $\setlvl(v,\ell)$ is another long process, which involves updating the relevant data structures,
and is thus simulated over multiple update operations.
During this process, the vertex set $N'_v(\ell)$ that we have just computed may change.
%
%
For this reason,
we randomly sample a vertex $w$ from $N'_\ell(v)$ as the new mate of $v$ \emph{before} initiating this process.
By doing this, we sample from at least $(1-\eps) \cdot \gamma^{\ell}$ vertices,
whereas if we were to make the random sampling after this process is finished,  
the sample space could a-priori be much smaller than  $(1-\eps) \cdot \gamma^{\ell}$, which would, in turn, invalidate our almost-maximality guarantee.
(Although it is not difficult to prove that the sample space does not reduce significantly during the call to $\setlvl(v,\ell)$, there is no need for it.)
Since $N'_\ell(v)$ does not contain any active vertices,
$w$ is not active when it is chosen as a random mate for $v$,
and in particular, it is not in the process of falling or rising; we then add $w$  to the $Active$ list.
It is crucial that the new mate $w$ of $v$ will be at level at most $\ell-1$.  
We argue that the random sampling of $w$ can be carried out within the time $\Delta$ or $\Delta'$ reserved for a single update operation.
(Although it is not difficult to prove that the sample space does not reduce significantly during the time it takes 
to randomly sample $w$ \emph{naively}, i.e., via a procedure whose runtime is linear in the sample space, there is no need for it.)
To this end, the set $N'_\ell(v)$  that we have just computed should be stored via a data structure that allows for fast retrieval of a random element,
such as a balanced (of logarithmic depth) tree in which every node is uniquely associated with a vertex of $N'_\ell(v)$ via mutual pointers,
and it also holds a counter for the number of nodes in its subtree.
Sampling a random integer $k$ from $\{1,2,\ldots,\min\{N'_v(\ell),\gamma^{\ell}\}\}$ takes at most $O(\log (\gamma^{\ell})) = O(\log n)$ time.
Since the depth of this  tree is $O(\log n)$, it is straightforward to retrieve the $k$th element in the tree in another $O(\log n)$ time.
In this way we sample a vertex uniformly at random from $N'_\ell(v)$ in time $O(\log n) \ll \Delta < \Delta'$,
 making sure not to sample from more than $\gamma^{\ell}$ options, as required.
The caveat is that this tree has to be maintained in the time interval required for constructing it, during which some elements may leave the tree while others may join it.
In particular, any vertex leaving/joining $N'_\ell(v)$ throughout this time interval triggers a node insertion/removal to/from the tree, which requires $O(\log n)$ time.
Note, however, that the time needed by our algorithm to update the various $\phi_v$ counters
following 
a fall/rise of a single neighbor of $v$ is also logarithmic. Since each update of the tree is triggered by a change to $N'_\ell(v)$, which may result from
a fall/rise of a neighbor of $v$ that incurs a logarithmic cost anyway,
it follows that the overhead due to maintaining this tree is negligible.


Next, the same thread continues to setting the levels of $v$ and $w$ to $\ell$,
by executing Procedure $\setlvl(v,\ell)$ and then executing Procedure $\setlvl(w,\ell)$. Note that $\ell_v \ge \ell$ and $\ell_w \le \ell-1$.
By Lemma \ref{setlevelTime}, the runtime of the former (respectively, latter) call is bounded by $L_{\ell_v}(v) =  O((\phi_v(\ell_v+1) + \log n) \cdot \log n)$
(resp., $L_{\ell}(w) =  O((\phi_w(\ell+1) + \log n) \cdot \log n)$), where $\phi_v(\ell_v+1)$ (resp., $\phi_w(\ell+1)$)
stands for the number of neighbors of $v$ (resp., $w$) of level lower than $\ell_v+1$ (resp., $\ell+1$) just before the call to $\setlvl(v,\ell)$
(resp., $\setlvl(w,\ell)$).
By Invariant \ref{risei}, $\phi_v(\ell_v+1) = \gamma^{\ell_v+1} \cdot O(\log^2 n) =  \gamma^{\ell_v} \cdot O(\log^3 n)$ and $\phi_w(\ell+1) = \gamma^{\ell+1} \cdot O(\log^2 n)
\le \gamma^{\ell_v} \cdot O(\log^3 n)$,
hence the runtime of these calls is $O(\gamma^{\ell_v} \cdot \log^4 n)$.
We then match $v$ with $w$, thus creating a new level-$\ell$ matched edge; this guarantees that Invariant \ref{first}(b) will continue to hold.

Observe that during the call to  $\setlvl(v,\ell)$, the values $\phi_v(j)$ that we computed may change, for $\ell+1 \le j \le \ell_v$.
Recall that the level of $v$ during this falling of $v$ is viewed as its destination level $\ell$. At the beginning of this falling
(i.e., when the execution of the call to $\setlvl(v,\ell)$ starts) we have $\phi_v(j) < \gamma^{j}$, for all $\ell + 1 \le j \le \ell_v$.
These $\phi_v(j)$ values may grow during this falling, possibly beyond the threshold $\gamma^j \cdot O(\log^2 n)$ required by Invariant \ref{risei}.
In this case some sub-schedulers of $\rise$ may need to raise $v$ while $v$ is falling, which triggers a conflict between those sub-schedulers and the sub-scheduler executing the call to $\setlvl(v,\ell)$. We address this issue in Section \ref{mainc}.

Finally, assuming $w$ was previously matched to $w'$, the vertex $w'$ becomes temporarily free.
By Invariant \ref{first}(b), the level of $w'$ is the same as the level of $w$ before executing Procedure $\setlvl(w,\ell)$,
hence it is at most $\ell-1$. We use the same thread to handle $w'$ recursively by calling $\evalf(w')$.

(The degenerate case that no level $\ell$ where $\phi_v(\ell) \ge \gamma^\ell$ exists, for $0 \le \ell \le \ell_v$, was  fully addressed in the high-level description of this procedure,
and no further details are needed in this case.)

Any vertex $v$ handled by a call to $\evalf(v)$ is temporarily free, and so its level exceeds -1 by definition.
Hence the recursion must bottom at vertices of level at least 0.  
Consider a vertex $v$ whose level $\ell_v$ at the beginning of the procedure's execution is 0, and note that
$\phi_v(0) \le \phi_v(1) = O(\log^3 n)$ by Invariant \ref{risei}.
Computing $\phi_v(0)$ and $N_v(0)$, including the authentication process, takes $O(\phi_v(1) + O(\log n)) \cdot O(\log n) = O(\log^4 n)$ time.
If $\phi_v(0) = 0$, then $v$ remains free, and the runtime of the required call to $\setlvl(v,-1)$ is
$O(\phi_v(0) + O(\log n)) \cdot O(\log n) = O(\log^2 n)$ by Lemma \ref{setlevelTime}.
Otherwise $\phi_v(0) \ge 1$, and $v$ will be matched deterministically to some neighbor $w$ of level -1, which is unmatched by Invariant \ref{1c}.
By Invariant \ref{1d}, all neighbors of $v$ of level -1 are outgoing of $v$, and all outgoing neighbors of $v$ have level at most 0.
We can find a neighbor $w$ of $v$ of level -1 by simply scanning the at most $\phi_v(1)$ outgoing neighbors of $v$, which requires time $O(\phi_v(1)) = O(\log^3 n)$.
The new matched edge is created at level 0, which triggers  calls to $\setlvl(v,0)$ and $\setlvl(w,0)$.
The runtime of the call to $\setlvl(v,0)$  is $O(\phi_v(1) + O(\log n)) \cdot O(\log n) = O(\log^4 n)$ by Lemma \ref{setlevelTime}.
Similarly, we have $\phi_w(1) = O(\log^3 n)$,  
and the runtime of the call to  $\setlvl(w,0)$  is $O(\phi_w(1) + O(\log n)) \cdot O(\log n) = O(\log^4 n)$;  
since $w$ was unmatched prior to getting matched with $v$,  the procedure does not proceed recursively.
We conclude that the time required by Procedure $\evalf$ to handle any level-0 vertex is $O(\log^4 n)$.

We next analyze the runtime of the procedure, for levels $\ell$ satisfying $T_\ell / \Delta \ge 1$.
The complementary regime of levels is addressed in Section \ref{lowlevels}.
\begin{lemma}    \label{evalftime}
Denote by $H_\ell$ the maximum (overall) time needed for executing Procedure $\evalf(v)$, where $v$ ranges over all level-$\ell$ vertices.
Then $H_\ell = O(\gamma^{\ell} \cdot \log^4 n)$. 
\end{lemma}
{\bf Remark.~} Lemma \ref{evalftime} implies that $H_\ell$ is bounded by $T_\ell$. (In fact, $H_\ell$ is smaller than $T_\ell$ by as large enough constant as we want.)
Moreover, the lemma also bounds the runtime of the recursive calls of this procedure.
Specifically, this runtime is bounded by $H_i = O(\gamma^{i} \cdot \log^4 n)$, where $i$ is the level of the vertex  handled by the recursive call.
(Again, $H_i$ is smaller than $T_i$ by a large constant.)
\vspace{4pt}
\begin{proof}
Consider an arbitrary level-$\ell$ vertex $v$ handled by Procedure $\evalf$.
By the description of this procedure,  an (overall) time of $O(\gamma^{\ell} \cdot \log^4 n)$ suffices for
 computing  the values $\phi_v(j)$ and vertex sets $N_v(j)$ (including the authentication process), for all   $j \le \ell_v$,
then computing the index $\ell$ and the vertex set $N'_v(\ell)$, next randomly sampling a random mate $w$ for $v$ and matching them (after deleting the old matched edge $(w,w')$),
and finally executing the calls $\setlvl(v,\ell)$ and $\setlvl(w,\ell)$.
Since the level of the old mate $w'$ of $w$ (if exists) is at most $\ell-1$,
the time needed by the recursive call $\evalf(w')$ is bounded by $H_{\ell-1}$.
We thus obtain the recurrence $H_\ell \le O(\gamma^{\ell} \cdot \log^4 n) + H_{\ell-1}$.
The recursion stops if the randomly chosen mate $w$ for $v$ is unmatched. 
By Invariant \ref{1c}, any vertex of level -1 is unmatched. Since the level of $w$ is at most $\ell-1$, it follows that the recursion must stop at a vertex of level 0 or higher, and the maximum time needed for executing this procedure is achieved when the level of that vertex is 0.  
Consequently, the basis of this recurrence is $H_0 = O(\log^4 n)$,  thus resolving to $H_\ell = O(\gamma^{\ell} \cdot \log^4 n)$. \QED
\end{proof}

\subsubsection{Levels with simulation time smaller than 1} \label{lowlevels}
In levels $\ell$ for which  $T_\ell / \Delta < 1$, the simulation time is smaller than 1, and then the adversary does not make any update operation within the time required by a level-$\ell$ thread to complete its entire execution.
For such levels $\ell$ the queue $\cQ_\ell$ of temporarily free vertices is empty, and every level-$\ell$ vertex that becomes temporarily free is handled instantly  by this procedure, i.e., without simulating the execution of this procedure over multiple update operations.
Note that the active vertices for such levels $\ell$ may constitute a significant part of $N_{\ell}(v)$, hence $N'_\ell(v)$ may be much smaller than 
$N_\ell(v)$, and in particular, the size of $N'_\ell(v)$ may be much smaller than $(1-\eps) \cdot \gamma^\ell$,
thus the sample space may be too small to apply a probabilistic argument.
However, for such levels there is no need for probabilistic arguments, as we may naively scan all vertices of $N'_\ell(v)$, match $v$ with an arbitrary vertex $w$ of level -1 there (such a vertex is unmatched by Invariant \ref{1c}),
or leave $v$ free if none is found. In the former case we let $v$ and $w$ fall and rise to the same level $\ell$, respectively, by calling to $\setlvl(v,\ell)$ and $\setlvl(w,\ell)$, and then match $v$ with $w$, thus creating a new level-$\ell$ matched edge.
(We may create this new matched edge  at any level between 0 and $\ell$. We choose to create it at level $\ell$ to be consistent with the way we handle levels
$\ell$ for which $T_\ell / \Delta \ge 1$.)
In the latter case $v$ does not have any free neighbor by Invariant \ref{first}(c), thus we let $v$ become free by making the call to $\setlvl(v,-1)$.
The cost of this deterministic treatment is
dominated by other parts of this procedure; disregarding the calls to $\setlvl(v,\ell)$ and $\setlvl(w,\ell)$, this cost
is linear in the size $\phi_v(\ell)$ of $N_{\ell}(v)$, which does not exceed  $\gamma^{\ell} \cdot O(\log^2 n)$ by Invariant \ref{risei}.




\subsubsection{The Behavior of Falling and Rising Vertices} \label{behavior}
In this section we establish some properties regarding the behavior of vertices that are falling and rising throughout the execution of Procedure $\evalf(v)$;
these properties are used in Sections \ref{schedulersAnal}-\ref{sec:almost}.

Consider an arbitrary vertex $v$ handled by Procedure $\evalf$, at the first recursion level.
By the description of the procedure, $v$ falls from level $\ell^{(0)} := \ell_v$ to level $\tilde \ell^{(0)} := \ell$.
It is possible that $\ell^{(0)} = \tilde \ell^{(0)}$, in which case the level of $v$ remains unchanged.  
On the other hand, $v$'s randomly chosen mate, $w$, rises from level $\ell^{(1)} := \ell_w$ at most $\tilde \ell^{(0)}-1$ to level $\tilde \ell^{(0)}$.
Assuming $w$ was previously matched to $w'$, the matched edge $(w,w')$ gets deleted from the matching
and the edge $(v,w)$ is inserted to the matching instead, and the procedure proceeds recursively to handling $w'$.
By Invariant \ref{first}(b), the level of $w'$ is $\ell^{(1)}$.  
Therefore, in the first recursion level, a single vertex fell from level $\ell{(0)}$ to level $\tilde \ell^{(0)}$ (possibly $\ell^{(0)} = \tilde \ell^{(0)}$, but we
may view this too as a fall)
and a single vertex rose from level $\ell^{(1)} \le \tilde \ell^{(0)} -1$ to level $\tilde \ell^{(0)}$.
The next lemma follows by a simple induction.
\begin{lemma} \label{parenthesis}
For each recursion level $i = 0,1,\ldots$, let $\ell^{(i)}$ denote the level of the vertex that the recursive call handles,
let $\tilde \ell^{(i)}$ denote the level to which that vertex falls in order to match,
and let $\ell^{(i+1)}$ be the (old) level of its new mate, which rises to level $\tilde \ell^{(i)}$ for the match.
Then for any $i$ and $j$ with $i < j$, we have
$$\ell^{(i)} ~\ge~ \tilde \ell^{(i)} ~>~   \ell^{(i+1)} ~\ge~ \ell^{(j)} ~\ge~ \tilde \ell^{(j)} ~>~ \ell^{(j+1)}.$$
In particular, as part of the execution of this procedure, at most $\log_\gamma(n-1)+1$ vertices are falling and at most $\log_\gamma(n-1) +1$ vertices are rising.
\end{lemma}

The following three corollaries are implied by Lemma \ref{parenthesis} and the descriptions of the algorithm and Procedure $\evalf$.
(The description of the algorithm is used only for deriving Corollary \ref{threadeval}.) 
\begin{corollary} \label{boundactive}
Among all   vertices that become temporarily free (and active) as part of the execution of Procedure $\evalf(v)$, at most three are temporarily free at any point in time during this execution.
\end{corollary}
\begin{corollary} \label{falling}
As part of the execution of Procedure $\evalf(v)$, vertices of level higher than $v$'s original level $\ell_v$ remain intact.
For any   $\ell \le \ell_v$:
~(1) At most one vertex falls from level $\ge \ell$ to level $< \ell$.
~(2) At most one matched edge gets deleted at  level $\ell$; at most one matched edge gets created at level $\ell$. 
\end{corollary}

\begin{corollary} \label{threadeval}
The call to $\evalf(v)$ is executed by a thread of level at least $\ell_v$.
Thus for any $\ell$, a thread must be of level at least $\ell$ in order to carry out, as part of the execution of Procedure $\evalf$:
~(1) The falling of a vertex from level $\ge \ell$ to level $< \ell$.  
~(2) The rising of a vertex to level $\ell$.  
~(3) The deletion of a level-$\ell$ matched edge.  
\end{corollary}


\section{Analysis of the Schedulers} \label{schedulersAnal}
The schedulers $\unmatch$ and $\rise$ are responsible for maintaining Invariants \ref{samp:in} and \ref{risei}, respectively.
While these invariants seem unrelated,
the underlying principle governs the operation of the respective schedulers is similar.
We find it instructive to describe this principle via a simple \emph{deterministic} balls and bins game from \cite{LO88}; see Section \ref{sec31}.

Understanding this principle is just the first step in the analysis of the schedulers.
To prove that the schedulers $\unmatch$ and $\rise$ indeed maintain the respective invariants, we carefully build on this principle in several steps;
see Sections \ref{sec32} and \ref{sec33}.
The analysis of $\rise$ is significantly more intricate than that of $\unmatch$, since
our update  algorithm affects both players (Player I and Player II) in the balls and bins game underlying the operation of $\rise$.
\subsection{A balls and bins game} \label{sec31}
The game is played between Player I and Player II. Initially there are $N$ empty bins.  
In each  round, Player I removes a bin of largest size.
Subsequently, Player II may add balls to bins, up to $b$ balls in total, where $b \ge 1$ is any positive integer. 
Denote by $|B|$ the size (i.e., number of balls) in bin $B = B_t$ in round $t$.
The game ends when no bin is left or when the size of any bin reaches some parameter $k$.
Player I wins (respectively, loses) in the former (resp., latter) case.
As shown in \cite{LO88}, Player I wins if $b < \frac{k}{(\ln N+1)}$. (A simple proof is given in App.\ \ref{ballsbins} for completeness.)
\subsection{$\unmatch$} \label{sec32}
Consider the variant of the game from Section \ref{sec31}, where each of the $N$ bins contains initially at least $k$ balls,
Player I removes a bin of smallest (rather than largest) size, and Player II removes (rather than adds) balls from bins, up to $b$ balls in total.
Player I wins when no bin is left, as before, but Player II wins when the size of any bin reduces from $\ge k$ to some   parameter $k'$.   
It is easily verified that Player I   wins if $b < \frac{k-k'}{(\ln N+1)}$,
using a symmetric argument to the one used in  App.\ \ref{ballsbins}.
\ignore{
Specifically, supposing for contradiction that bin $B_1$ contains $k'$ balls at time $t_1$,
the following is the symmetric statement of Claim \ref{basicgame}:
\begin{claim}
For all $i \ge 1$, at time $t_i$, we have $|B_i| \le k' + \sum_{\ell = 1}^{i-1} \frac{b}{\ell}$
and $\sum_{\ell=1}^{i} |B_\ell| \le i \cdot (k' + \sum_{\ell = 1}^{i-1} \frac{b}{\ell})$.
\end{claim}
Now $B_r$ contains exactly $k$ balls at time $t_r = 0$, and we obtain a  contradiction:
\\$|B_r| ~\le~ k' + \sum_{\ell = 1}^{r-1} \frac{b}{\ell} ~\approx~ k' + b\ln (r-1) ~<~ k' + b\ln N ~<~ k$.
}

The number of update operations needed for simulating the entire execution of a level-$\ell$ thread  run by $\unmatch_\ell$ is $T_\ell / \Delta = \Theta(\eps(\gamma^\ell / \log n))$.
Fix any level $\ell$ such that $T_\ell / \Delta \ge 1$ (i.e.,   consider the regime of levels where $\gamma^\ell = \Omega(\log n / \eps)$).
We show that Invariant \ref{samp:in} is maintained by  translating the above variant of the game appropriately,
specifically, we argue that $\unmatch_\ell$ and the adversary can be viewed as Player I and Player II in this game, respectively,
where the level-$\ell$ matched edges are the $N$ bins, the edges in the samples of these bins are the balls, $k = (1-\eps) \cdot \gamma^{\ell}, k' = (1-2\eps) \cdot \gamma^{\ell}$,
and $b \le 2T_\ell / \Delta$.

The sub-scheduler $\unmatch_\ell$ always removes a level-$\ell$ matched edge of smallest remaining sample from the matching, which is the analog of removing a bin of smallest size.
The adversary may delete edges of its choice from the graph, and as a result from the corresponding samples, which is the analog of removing balls from bins.
Observe that the same edge $(u,v)$ may belong to the samples of at most two matched edges, one adjacent to $u$ and another to $v$,
thus a single edge deletion by the adversary translates into at most two ball removals in the balls and bins game.
We will need to bound the maximum number of edges deleted from the graph by the adversary, denoted $b'$,
throughout the time needed for $\unmatch_\ell$ to remove a single level-$\ell$ matched edge.
Specifically, to guarantee that Player I wins, we will show that $b' < \frac{k-k'}{2(\ln N+1)}$, which translates into $b < \frac{k-k'}{(\ln N+1)}$ in the balls and bins game.

Note that other ingredients of the update algorithm (besides $\unmatch_\ell$) may remove level-$\ell$ matched edges from the matching.
Moreover, the adversary too may trigger the removal of such edges from the matching, by deleting them from the graph.
All these extra edge removals from the matching actually {save} $\unmatch_\ell$ the time of removing the corresponding bins on its own.
Since our goal is to show that Player I wins, we may henceforth disregard this extra power added to it.

Level-$\ell$ matched edges (for levels $T_\ell / \Delta \ge 1$) are created only by Procedure $\evalf$. By the description of this procedure (see Section \ref{sec26}),
when any level-$\ell$ matched edge is created, its sample is at least $k = (1-\eps) \cdot \gamma^\ell$.
We may thus assume that the corresponding bin was there from the outset of the game, containing at least $k$ balls then.
Although the same edge may be deleted and inserted from the matching and/or   graph multiple times,
we view any newly created level-$\ell$ matched edge  as a different bin.  
A naive upper bound on the total number of level-$\ell$ matched edges is the overall runtime of the algorithm.
Since the worst-case update time of the algorithm is bounded (deterministically) by $O(\max\{\log^7 n /\eps,\log^5 n / \eps^4\})$
and as we may assume that $\eps = \Omega(1/n)$, the overall runtime   does not exceed the total number of update operations by more than a factor of
$O(n^4 \log^5 n)$.
While we consider arbitrarily long update sequences, it is easy to verify (refer to Section \ref{short:app}) that the assumption that
the length of the update sequence is bounded by $O(n^2)$ does not lose generality,
hence $N \le  O(n^2) \cdot O(n^4 \log^5 n) = O(n^6 \log^5 n)$.


Within a time interval of length $T_\ell / \Delta = O(\eps(\gamma^\ell / \log n)) \ge 1$, the sub-scheduler $\unmatch_\ell$ removes a matched edge of smallest remaining sample.
During this time, the adversary may delete at most $T_\ell / \Delta$ edges from the graph,
thus we  have  $b' \le T_\ell / \Delta = O(\eps(\gamma^\ell / \log n))$.
Note that $\ln N = \Theta(\log n)$ and recall that $\Delta = \Theta(\log^5 n/\eps)$.  By setting the constant hiding in the $\Theta$-notation of the definition of $\Delta$ to be sufficiently large, we obtain $b' < \frac{k-k'}{2(\ln N+1)}$,  which translates into $b < \frac{k-k'}{(\ln N+1)}$ in the balls and bins game, as required.
It follows that Player I ($\unmatch_\ell$)  wins the game, for any level $\ell$ such that $T_\ell / \Delta \ge 1$,
or equivalently, the size of the samples of all level-$\ell$ matched edges always exceeds $(1-2\eps) \cdot \gamma^\ell$.

Summarizing, we showed that Invariant \ref{samp:in} is maintained.
\subsection{$\rise$} \label{sec33}
Consider the variant of the game from Section \ref{sec31},
where the bins are not empty initially, but rather contain at most $k' \ll k$ balls each, and everything else remains as before.
Using the same argument as the one used in App.\ \ref{ballsbins}, Player 1 wins if $b < \frac{k-k'}{(\ln N+1)}$.

Next, 
we show that Invariant \ref{risei} is maintained by translating this variant of the game appropriately.

Fix an arbitrary level $\ell \ge 0$.
Invariant \ref{risei} requires that the $\phi_v(\ell)$ values are always  bounded by $\gamma^{\ell} \cdot O(\log^2 n)$, for all vertices $v$ with $\ell_v < \ell$.
In the corresponding balls and bins game, the bins will represent the respective vertex sets $N_v(\ell)$ (of $v$'s neighbors of  level at most $\ell-1$), for $\ell_v < \ell$;
we stress that bins for vertices $v$ with $\ell_v \ge \ell$ are not considered.
(Note that our algorithm does not   maintain these sets, only the corresponding $\phi$ values.)
The sub-scheduler $\rise_\ell$ will be Player I in the game.
Note that it always picks a vertex $v$ whose $\phi_v(\ell)$ value is highest, and lets it rise to level $\ell$.
Following this rise, $\ell_v = \ell$, hence the invariant for $v$ and level $\ell$ holds vacuously, as the respective bin is no longer considered.
Thus, the analog of removing a bin by Player I is to let a vertex rise to level $\ell$.

Following the rise of $v$ to level $\ell$, the invariant for $v$ holds vacuously also for all levels lower than $\ell$,
which   translates into removing the   bins of $v$ also from all lower levels (where those bins are considered).
Observe that a separate balls and bins game is   played at each level, with its own Player I and Player II,
and so bin removals at level $\ell$ may be triggered by   higher level games.
Note, however, that this may only strengthen the power of   Player I at level $\ell$,  as the extra bin removals due to higher level games actually save $\rise_\ell$ the time of removing those bins on its own.
Since our goal is to show that Player I wins, we may henceforth disregard this extra power added to it due to higher levels.

At the beginning of the update sequence  the graph is empty, so all vertices are free and their level is -1.
Moreover,  at this stage all vertex sets $N_v(\ell)$ are empty.  
Thus initially we have an empty bin for every vertex.  
As time progresses some of these bins are being removed due to vertex rising, either by the sub-scheduler $\rise_\ell$ or by Procedure $\evalf$.
Recall that when a vertex rises to level $\ell$, all its bins up to level $\ell$ are removed instantly.
On the other hand, bins are also being created due to vertices falling, by Procedure $\evalf$.
When a vertex $v$ starts falling from level $\ell_v$ to level $\ell$, it is as if the corresponding vertex sets $N_v(j)$ in all levels $j \in \{\ell + 1,\ldots,\ell_v\}$
are created instantly; indeed, recall that the level of $v$ is viewed as its destination level $\ell$ from the moment its falling to level $\ell$ starts.
Although the same vertex set $N_v(j)$ may be removed and created multiple times,
we view any such newly created set as a different bin that was there from the outset of the game, to be in accordance
with the game description,  which states that all $N$ bins exist from the game outset.
To comply with the condition that every bin contains initially at most $k'$ balls,
we set the corresponding threshold $k' = k'_\ell$ to be $\gamma^{\ell}$ (which grows geometrically with the level $\ell$),
and prove the following lemma.  
\begin{lemma} \label{lemfinale}
Any newly created level-$j$ bin contains at most $k'_j = \gamma^{j}$ balls, for any $j = 0,1,\ldots,\log_\gamma(n-1)$.
\end{lemma}

\begin{proof}
Consider an arbitrary fall of some vertex $v$, and the corresponding call to Procedure $\evalf(v)$.
Note that $v$ is falling from level $\ell_v$ to the highest level $\ell, 0 \le \ell \le \ell_v$,
where $\phi_v(\ell) \ge \gamma^{\ell}$. It is possible that $v$ did not fall to a lower level, i.e., $\ell = \ell_v$,
but if it did, then for every level $j$ with $\ell + 1 \le j \le \ell_v$, we have $\phi_v(j) < \gamma^{j}$.
Consequently, any level-$j$ bin that got created as a result of a fall of $v$ from level $\ell_v$ to a lower level $\ell$, for $\ell + 1 \le j \le \ell_v$, contains less than $k'_j = \gamma^{j}$ balls.
\QED
\end{proof}

\noindent
{\bf Remark.~} The level of $v$ is viewed as its destination level $\ell$ from the moment its falling to level $\ell$ starts.
The bound on the number of balls in the corresponding bins of $v$ provided by Lemma \ref{lemfinale} holds for that moment,
but it may grow significantly throughout the falling of $v$.
For any level $j$, it is the responsibility of   $\rise_j$ to guarantee that no level-$j$ bin ever contains too many balls,
also during the falling process that triggered the creation of that bin.  
Consequently, some sub-schedulers of $\rise$ may have to raise $v$ while $v$ is falling,
which may trigger a conflict between those sub-schedulers and the sub-scheduler responsible for the falling of $v$.
While we must upper bound $k' = k'_\ell$, the bound $\gamma^\ell$ provided by Lemma \ref{lemfinale} is too strict,
and any bound smaller than $k = O(\gamma^\ell \cdot \log^2 n)$ by a constant factor will do.
We use this observation in Section \ref{mainc} to resolve the potential conflict between sub-schedulers that may try to raise $v$ while it is falling.

The level-$\ell$ vertex sets $N_v(\ell)$ and the corresponding values $\phi_v(\ell)$ of vertices may grow either due to edge insertions (by the adversary) or due to falling vertices
(by the update algorithm).  
Hence, Player II in the game will consist of both the adversary and the update algorithm.
Coping with adversarial edge insertions can be done similarly to Section \ref{sec32}, i.e., by letting $\rise_\ell$ work sufficiently faster than the adversary.
Coping with falling vertices is the tricky part.  
In particular, letting $\rise_\ell$ work faster than the other sub-schedulers is doomed to failure (if arguing about it naively):
While this would lead to more vertices rising, which helps Player I win,
each vertex rising may trigger the fall of another vertex, which has the reverse effect.
Hence, we will not try to increase the speed in which Player I (i.e., $\rise_\ell$) works, but rather show that
the number of balls $b = b_\ell$ added to the bins by Player II (adversary +  update algorithm) while Player I removes a bin is not too large.
\begin{lemma} \label{lem2finale}
$b = O(\gamma^{\ell} \cdot \log n)$.
\end{lemma}

\begin{proof}
The level-$\ell$ vertex sets $N_v(\ell)$ and the corresponding values $\phi_v(\ell)$ of vertices may grow, either by adversarial edge insertions
or by vertices falling, which is the analog of adding balls to bins.
Symmetrically, these vertex sets and corresponding values may shrink, either by adversarial edge deletions or by vertices rising,
which translates into removing balls from bins, and only helps Player I ($\rise_\ell$).
Moreover, bins are removed not only by Player I, but also due to vertices that are rising by Procedure $\evalf$,
which saves Player I the time of removing the corresponding bins on its own.
Since our goal is to show that Player I wins, we may henceforth
disregard this extra power added to it, specifically, due to removals of balls and due to removals of bins by Procedure $\evalf$.  

Within a time interval of length $T_\ell / \Delta' = O(\eps(\gamma^\ell / \log^2 n))$, the sub-scheduler $\rise_\ell$
lets a vertex $v$ of highest $\phi_\ell(v)$ value rise to level $\ell$.
During such a time interval, the adversary may add at most $O(\eps(\gamma^\ell / \log^2 n))$ edges to the graph.
Since each edge insertion $(v,w)$ may increase the values of $\phi_\ell(v)$ and $\phi_\ell(w)$ by at most one each,
we may view it as if $O(\eps(\gamma^\ell / \log^2 n))$ balls have been added by Player II to the level-$\ell$ bins.
Next, we  bound the number of balls added to those bins due to vertices falling.

We say that a fall of a vertex $v$ \emph{intersects} level $\ell$ if it is from level at least $\ell$ to level lower than $\ell$.
Note that each vertex fall that does not intersect level $\ell$ has no effect on the $\phi_\ell(w)$ values of vertices $w$.
Next, we analyze a vertex fall of $v$ that intersects level $\ell$.
This fall increases by one each of the $\phi_\ell(w)$ values of $v$'s neighbors $w$ with level at most $\ell-1$.  
The key observation is that there are less than $\gamma^{\ell}$ such neighbors $w$, otherwise $v$ would not fall below level $\ell$.
We may henceforth view it as if less than $\gamma^{\ell}$ balls are added to the level-$\ell$ bins corresponding to those neighbors $w$ due to this vertex fall.
To bound the number of such vertex falls,
note that each vertex fall intersecting level $\ell$ is made within a call to Procedure $\evalf(u)$, for some vertex $u$ whose level is at least $\ell$ by Corollary
\ref{falling}.
By Corollary \ref{threadeval}(1), such a call to Procedure $\evalf$ is executed by a thread of level at least $\ell$.
By Corollary \ref{falling}(1), as part of the execution of this thread, at most one vertex fall may intersect level $\ell$.
Recall that the time $T_j$ of the execution threads grows geometrically with the level $j$, by a factor of $\gamma = \Theta(\log n)$.
We argue that during a time interval of length $T_\ell / \Delta'$ (which suffices for simulating the entire execution of any level-$\ell$ thread that is run by $\rise_\ell$),
at most $O(\log n)$ execution threads of levels at least $\ell$ may run.
Indeed, since the simulation parameter of $\rise_\ell$ is $\Delta'$, it is no slower than any of the other level-$\ell$ sub-schedulers,
so we may have $O(1)$ execution threads due to level $\ell$.
We may have at most $O(1)$ execution threads due to any higher level,
since the execution threads get slower with each level, which  sums to $O(\log n)$ execution thread over all levels.
Consequently, at most $O(\log n)$ vertex falls may intersect level $\ell$ during any time interval
$T_\ell / \Delta'$.
Recall that each such vertex fall triggers less than $\gamma^{\ell}$ balls to be added to the level-$\ell$ bins,
thus at most $O(\gamma^{\ell} \cdot \log n)$ balls are added to those bins during this time interval due to vertices falling.
It follows that $b = O(\gamma^{\ell} \cdot \log n)$.
\QED
\end{proof}


To bound the number $N$ of level-$\ell$ bins, note that besides the $n$ empty bins at the beginning of the update sequence,
any newly created bin is due to a vertex fall that intersects level $\ell$. Since each vertex fall requires the algorithm to spend $\Omega(1)$ time,
 the total number of level-$\ell$ bins is naively upper bounded by $n$ plus the overall runtime of the algorithm.
Following similar lines to those in Section \ref{sec32}, we thus have $N \le O(n^6 \log^5 n)$, and so $\ln N =   \Theta(\log n)$.
Taking $k = O(\gamma^\ell \cdot \log^2 n)$ completes the translation of the balls and bins game.
By setting the constant hiding in the $O$-notation of the definition of $k$ to be sufficiently large, we obtain $b < \frac{k-k'}{\ln N}$.
It follows that Player I ($\rise_\ell$) wins the game for any level $\ell$, or equivalently,
the $\phi_v(\ell)$ values of all vertices $v$ with $\ell_v < \ell$ are always smaller than $k = \gamma^{\ell} \cdot O(\log^2 n)$.

Summarizing, we showed that Invariant \ref{samp:in} is maintained.



\section{A Mechanism for Resolving Conflicts between Schedulers} \label{exceptions}
What happens if multiple schedulers want to manipulate on the same vertex at the same time?
While this question may seem problematic at first, recall that only $O(\log n)$ vertices may be active  at any  time,
and so it is not difficult to resolve such conflicts (in various ways) by paying a polylogarithmic overhead in the update time.
We next describe the various conflicts, and our mechanism for resolving them so that the worst-case update time does not increase by more than a constant factor.  

\subsection{Minor issues} \label{minor}
The first conflict may arise if a level-$\ell$ vertex $v$  chooses as its random mate a lower level neighbor $w$ while $w$ is rising/falling, and more generally while $w$ is active.
For this reason we made sure to prune from $N_{\ell}(v)$ all active vertices to obtain $N'_\ell(v)$,
and perform the sampling from the pruned set $N'_\ell(v)$. This tweak guarantees that no vertex will be chosen as a random mate while it is active.

Consider   the schedulers $\unmatch$ and $\shuffle$. A conflict may arise if the next edge chosen for deletion from the matching by one of these schedulers is no longer matched.
This is not really a conflict, though, since the execution of threads is not done in parallel, but rather sequentially.
In general, once any thread deletes an edge from the matching, it makes sure to update the various schedulers about it;
this update takes $O(\log n)$ time.
As a result, whenever any scheduler tries to delete an unmatched edge from the matching, that edge is guaranteed to be matched.

Consider next the scheduler $\temp$ and an arbitrary level $\ell$.
A conflict may arise if the next vertex chosen by $\temp_\ell$ from the queue $Q_\ell$ (of temporarily free level-$\ell$ vertices due to the adversary) is active. Again, this is not really a conflict, since we can make sure to remove any level-$\ell$ vertex from the queue $Q_\ell$ whenever it becomes active; this update takes constant time.

\subsection{The main conflict} \label{mainc}
The scheduler $\rise$ is more problematic than the others. Consider an arbitrary level $\ell$.

A conflict arises when $\rise_\ell$ raises a vertex $v$ to level $\ell$ while $v$ is active.  
In contrast to the other conflicts, here we may have to let $\rise_\ell$ raise $v$ to level $\ell$ while $v$ is active,
despite the fact that $v$ may already be ``in motion'' (i.e., falling/rising).
Specifically, if $v$'s level is $\ge \ell$, then $\rise_\ell$ should not raise $v$ to level $\ell$, as in this case it will either lower $v$'s level or do nothing.
(If $v$ is in motion, we consider its destination level.)
However, if the (destination) level of $v$ is $< \ell$ (regardless of whether $v$ is falling or rising), it is sometimes essential that
$\rise_\ell$ raises $v$ to level $\ell$ for our algorithm to perform correctly.
For this reason, as mentioned in Section \ref{sec22}, we give precedence to the sub-schedulers
by decreasing order of simulation times (and thus by decreasing order of levels),
i.e., the $\log_\gamma(n-1)$-level thread is handled first, then the $\log_\gamma(n-1)-1$-level thread, etc., until the $0$-level thread.
(This precedence is kept among all  threads, not just those run by $\rise_\ell$.)

Recall that the execution of threads is being simulated over multiple update operations, and
consider the segments in the execution of all threads (over all levels) run by the scheduler $\rise$, which are being simulated following an arbitrary update operation.
The execution segment corresponding to $\rise_\ell$ will be run after all higher level execution segments have finished running (following that update operation).
Due to the nesting property of threads discussed in Section \ref{sec22}, whenever a new execution thread is starting to run by $\rise_j$, for any level $j > \ell$,
a new execution thread will also start running by $rise_\ell$.
Although these runs occur following the same update operation, they do not occur in parallel
but rather sequentially (one after another), thus the first execution segment of $\rise_\ell$ will be run after that of $\rise_j$, for any $j > \ell$.
When $\rise_\ell$ chooses which vertex to raise at the start of its execution,
it will make sure to choose a different vertex than the one chosen by $\rise_j$, for any $j > \ell$,
and this is the meaning of $\rise_j$ having ``precedence'' over $\rise_\ell$.
Specifically, $\rise_\ell$ will
raise the vertex $v$ (where $\ell_v < \ell$)
whose $\phi_v(\ell)$ value is highest among all vertices that have not been chosen by any higher level $\rise_j$, for $j > \ell$.
In this way we guarantee that no vertex will be chosen to get raised by more than one sub-scheduler of $\rise$, but we are not done,
as  $\rise_\ell$ may choose to raise a vertex that was not chosen by another sub-scheduler of $\rise$, but is nonetheless active.

\ignore{
A vertex may become active either due to getting chosen as a random mate of some vertex,
which may happen only by Procedure $\evalf$,
or due to getting chosen by one of the schedulers, in which case that vertex $v$ is handled
via a call to either $\evalf(v)$ or $\setlvl(v,\cdot)$.
More specifically, any rising vertex $v$ always rises to a higher level, either due to getting chosen as a random mate
as part of a call to $\evalf(w)$, for   $w \ne v$, or due to getting chosen by $\rise$; in both cases the rising of $v$ is carried out as part of the subsequent call to $\setlvl(v,\cdot)$.  
Any other active vertex is handled by choosing a random mate, which happens only as part of Procedure $\evalf$; to do so it may either fall to a lower level,
in which case the falling is carried out as part of the subsequent call to $\setlvl$,
or remain at the same level, which for technical reasons will also be viewed as a fall.
Recall that the rising/falling of a vertex $v$ does not end before the respective call to $\setlvl(v,\cdot)$ finishes its execution.
}
It is instructive to view the falling/rising of a vertex $v$ as extending beyond the execution of the respective call to $\setlvl$, to include all operations done prior to this call 
and are needed for its execution (e.g., computing the level to which $v$ should fall/rise) as well as all operations done immediately afterwards  (e.g., matching $v$ to a random mate).   
In this way we may associate each active vertex at any point in time with a call to $\setlvl$, which allows us to view it as either falling or rising.
We may therefore restrict our attention to conflicts that arise when $\rise_\ell$ chooses to raise a vertex $v$ that is already falling or rising.
Furthermore, we only consider falling or rising of $v$ due to Procedure $\evalf$. Indeed, the only way for a vertex to fall/rise outside this procedure is due to $\rise$,
but we have already addressed conflicts between multiple sub-schedulers of $\rise$.

Consider first a vertex $v$ that \emph{rises} due to Procedure $\evalf$.
There is no conflict if $v$ rises to level  $\ge \ell$, as then $\rise_\ell$ need not raise $v$ to level $\ell$;
we   henceforth assume that $v$ rises to   level $\ell' < \ell$.
By the description of Procedure $\evalf$,
this may  occur only if $v$ is chosen as a random mate of some level-$\ell'$ vertex $w$, as part of the call to $\evalf(w)$.
Note that $\rise_\ell$ may choose to raise $v$ to level $\ell$ while $v$ is  rising to level $\ell'$ by the sub-scheduler
executing the call to $\evalf(w)$,
which triggers a conflict between that sub-scheduler and $\rise_\ell$; we refer to such a conflict as a \emph{rising conflict}.
(The process of $v$ rising to level $\ell'$ ends after the call to $\setlvl(v,\ell')$ finishes its execution, when $v$ becomes matched and is no longer considered active.)
By Corollary \ref{threadeval}(2), such a call to Procedure $\evalf$ is executed by a thread of level at least $\ell'$.
Although the   level of that thread may be significantly larger than $\ell'$,
the remark following Lemma \ref{evalftime} implies that the overall runtime of the call to $\evalf(w)$ is bounded by $T_{\ell'}$,
regardless of that thread's level.
To cope with rising conflicts,
we augment the $Active$ list by adding to it, for every level $\ell$, the subsequent vertex that is going to be raised by $\rise_\ell$.
That vertex is computed in advance, i.e., just before $\rise_\ell$ starts raising some vertex, it   computes the subsequent vertex that it is going to raise,
hereafter the \emph{next-in-line} vertex; from now on we add all next-in-line vertices (over all levels) to the $Active$ list, viewing them as active.
Any next-in-line vertex is computed by giving precedence to higher level threads, which implies that both the current and next-in-line level-$\ell$ vertices chosen by $\rise_\ell$ are different than the current and the next-in-line level-$j$ vertices chosen by $\rise_j$ for any $j > \ell$, and thus also for any $j \ne \ell$.
The next lemma shows that this tweak prevents   rising conflicts.
\begin{lemma}
$\rise_\ell$ never raises a vertex $v$ while $v$ is rising as part of a call to $\evalf$.
\end{lemma}
\begin{proof}
Consider a call to $\evalf(w)$, where $w \ne v$ is a level-$\ell'$ vertex with $\ell' < \ell$, during which $v$ is chosen as the random mate of $w$.
Since the random sampling of a mate excludes active vertices, $v$ is not active when $w$ chooses it as its random mate, thus at that time $v$ is
neither the vertex currently raised by $\rise_\ell$ nor the next-in-line.
The execution thread of $\rise_\ell$ has an overall runtime of $T_\ell$ and a simulation parameter of $\Delta'$, i.e.,  its simulation time is $T_\ell / \Delta'$.
By Lemma \ref{evalftime}, the overall runtime of the call to $\evalf(w)$ is $\le T_{\ell'} \le T_{\ell} / \gamma$, and the simulation parameter of the thread running this call is either
$\Delta'$ or $\Delta = \Delta' / \gamma$.   
Consequently, the number of update operations needed for the thread simulating the call to $\evalf(w)$ to carry out the entire simulation is $\le T_{\ell'} / \Delta$,
which is no greater than the simulation time  $T_\ell / \Delta'$ of  $\rise_\ell$.
It follows that the execution of the call to $\evalf(w)$, and in particular the rising of $v$ that it triggers,  must terminate before
the thread raising the next-in-line vertex of $\rise_\ell$ terminates its execution.  
\QED
\end{proof}

Consider next a vertex $v$ that \emph{falls} due to Procedure $\evalf$, and denote by $\ell_v$ its level just before the fall.
There is no conflict if $v$ falls to level  $\ge \ell$, as then $\rise_\ell$ need not raise $v$ to level $\ell$;
we   henceforth assume that $v$ falls to   level $\ell' < \ell$.
By the description of Procedure $\evalf$,
this may  occur only 
as part of a call to $\evalf(v)$, i.e., where $v$ is the temporarily free vertex handled by this call.
Moreover, this may occur only if $\phi_v(j) < \gamma^{j}$, for any $\ell' + 1 \le j \le \ell_v$.
(It is possible that $\ell' = \ell_v$, but we view this too as a fall.)
Note that $\rise_\ell$ may choose to raise $v$ to level $\ell$ while $v$ is   falling to level $\ell'$ by the sub-scheduler
executing the call to $\evalf(v)$,
which triggers a conflict between that sub-scheduler and $\rise_\ell$; we refer to such a conflict as a \emph{falling conflict}.  
(The process of $v$ falling to level $\ell'$ ends after the call to $\setlvl(v,\ell')$ finishes its execution, when $v$ becomes either matched or free, and is no longer considered temporarily free and active.)
By the remark following Lemma \ref{evalftime}, the overall runtime of the call to $\evalf(v)$ is bounded by $T_{\ell_v}$, but
$\ell_v$ may be much larger than $\ell$. Thus the tweak used for coping with rising conflicts does not work for falling conflicts.
To cope with falling conflicts, we aim at preventing $\rise_j$ from handling $v$  while $v$ is falling, for any $j$.
(We consider an arbitrary level $j$ rather than $\ell$, as we cannot bound $\ell_v$ in terms of $\ell$.)
Concretely, if $v$ is the next-in-line vertex of any  $\rise_j$ just before the execution of $\evalf(v)$, for $j > \ell_v$, then we ignore the call to $\evalf(v)$;
this is fine, since the sub-scheduler among those of highest level will handle $v$ in the sequel.
We henceforth assume that $v$ is not the next-in-line vertex of any
$\rise_j$, for $j > \ell_v$. Since the overall runtime of the call to $\evalf(v)$ is $\le T_{\ell_v}$ and as
the simulation parameter of the thread running this call is either
$\Delta'$ or $\Delta = \Delta' / \gamma$,
 the number of update operations needed for this thread to carry out the entire simulation   is $\le T_{\ell_v} / \Delta$,
which is no greater than the simulation time  $T_j / \Delta'$ of  $\rise_j$, for any $j > \ell_v$.
It follows that no sub-scheduler $\rise_j$ raises $v$ while $v$ is falling as part of the call to $\evalf(v)$, for   $j > \ell_v$.

For levels $j \le \ell_v$, however,
some sub-schedulers $\rise_j$ may try to raise $v$ while $v$ is falling (as part of the call to $\evalf(v)$).
We thus prevent any  $\rise_j$ from raising $v$ to level $j$ while $v$ is falling, for any $j \le \ell_v$.
While this tweak guarantees that there are no falling conflicts, it tampers with the regular operation of  $\rise$,
which may lead to a violation of Invariant \ref{risei}.
Indeed, since the falling of $v$ 
is simulated over multiple update operations, while initially we have $\phi_v(j) < \gamma^{j}$, for any $\ell' + 1 \le j \le \ell_v$, some of these $\phi_v(j)$ values may grow during the fall, possibly beyond the threshold $\gamma^j \cdot O(\log^2 n)$ required by Invariant \ref{risei}.
To address the potential violation to the invariant at levels up to $\ell_v$, we slightly modify Procedure $\evalf(v)$
starting after the execution of Procedure $\setlvl(v,\ell')$ ends, so as to raise $v$ to the highest violating level $\ell^*$.
(In what follows we assume that there is a violation, otherwise we are done.)
In particular, the sample space that consists of all neighbors of $v$ with level $< \ell^*$ is large, and so we can choose a random mate for $v$
from $\gamma^{\ell^*}$ options, and continue similarly to before.
It is instructive to view the concatenation of the original fall of $v$ to level $\ell'$ and the subsequent rise to level $\ell^*$
as a single fall of $v$ from level $\ell_v$ to level $\ell^*$; by doing so, all statements of Section \ref{behavior} remain valid unchanged.
Procedure $\evalf(v)$ should no longer create a matched edge at level $\ell'$,
but rather at level $\ell^*$, hence the rest of the procedure's execution should be modified accordingly.
The details of this modification follow similar lines as in the original procedure, and are thus omitted.
Observe that the modified part of the procedure (starting after the fall of $v$ to level $\ell'$) does not handle any vertex of level $> \ell^*$.
This observation is immediate for the current recursion level, whereas for subsequent   levels
it follows from Lemma \ref{falling}, as we proceed recursively with a vertex of level $< \ell^*$.
Hence, using similar lines as in the proof of Lemma \ref{evalftime}, the runtime of the modified part of the procedure is  $O(\gamma^{\ell^*} \cdot \log^4 n)$,
which we may assume is $\le T_{\ell^*}$.
By Lemma \ref{evalftime}, the runtime of the original call to $\evalf(v)$ is $O(\gamma^{\ell_v} \cdot \log^4 n)$.
Since the runtime of the modified procedure exceeds the original   by at most an additive factor of $O(\gamma^{\ell^*} \cdot \log^4 n)$, where  $\ell^* \le \ell_v$,
it is also $O(\gamma^{\ell_v} \cdot \log^4 n)$, which we may assume is $\le T_{\ell_v}$.

To complete the argument, we argue that no sub-scheduler $\rise_j$ \emph{needs} to raise $v$, for any level $j$, while $v$ is rising from level $\ell'$ to level $\ell^*$.
For any level $j \le \ell^*$, there is  no need for $\rise_j$ to raise $v$ to level $j$.
For any level  $j > \ell_v$, we proved that no sub-scheduler $\rise_j$ tries to raise $v$ during the entire execution of the call to $\evalf(v)$.
Although this assertion was proved for the original procedure, the same argument carries over for the modified procedure,
since the overall runtime of the call to $\evalf(v)$ remains bounded by $T_{\ell_v}$.
It is left to consider levels $j \in \{\ell^* + 1, \ldots,\ell_v\}$. Just before raising $v$ from level $\ell'$ to level $\ell^*$,
we have $\phi_v(j) < \gamma^{j}$, for any $\ell^* + 1 \le j \le \ell_v$.
These $\phi_v$ values may grow during the execution of the modified part of Procedure $\evalf$ (starting after the fall of $v$ to level $\ell'$).
We showed that the runtime of the modified part of the procedure is $\le T_{\ell^*}$.
Since the simulation parameter of the thread running this call to $\evalf(v)$ is either $\Delta'$ or $\Delta = \Delta' / \gamma$,
the number of update operations needed for the thread simulating this call to carry out the entire simulation of the modified part
is $\le T_{\ell^*} / \Delta = T_{\ell^*+1} / \Delta'$, which is no greater than the simulation time $T_j / \Delta'$ of any sub-scheduler $\rise_j$, for $j \ge \ell^*+1$.
Following   the same lines as in the proof of Lemma \ref{lem2finale},   each $\phi_v(j)$ value may grow by at most
$O(\gamma^j \cdot \log n)$ units within $T_j / \Delta'$ update operations, for any $j \ge \ell^*+1$. Summarizing, all these $\phi_v(j)$ values may not exceed $\gamma^j + O(\gamma^j \cdot \log n) = O(\gamma^j \cdot \log n)$
while $v$ is rising from level $\ell'$ to level $\ell^*$; thus although these values grow, they are still reasonably small.
Switching to the balls and bins game terminology of Section \ref{sec33},
the upper bound $k' = k'_j$ on the initial number of balls in a bin will no longer be bounded by $\gamma^j$, but rather by $O(\gamma^j \cdot \log n)$,
in the balls and bins game played by each $\rise_j$, for  $j \in \{\ell^* + 1, \ldots,\ell_v\}$.
This tweak does not change the outcome of the game by the remark following Lemma \ref{lemfinale},
which means that $\rise_j$  need not remove bins with less than $O(\gamma^j \cdot \log n)$ balls for winning the game.
In other words,  no sub-scheduler $\rise_j$ needs to raise $v$ while $v$ is rising from level $\ell'$ to level $\ell^*$ for Invariant \ref{risei} to hold, for any $j \in \{\ell^* + 1, \ldots,\ell_v\}$.
\section{Proof of (Almost-)Maximality} \label{sec:almost}
To prove that the matching is always almost-maximal,
we need to show that the number of temporarily free vertices at any point in time is just an $\eps$-fraction of the number of matched edges.

At any point in time, at most $O(\log n)$ threads are running.
Each of these threads may cause only $O(1)$ vertices to become temporarily free throughout its execution.
(This assertion holds trivially for all threads, except for those that run Procedure $\evalf$, but for those the assertion follows from Corollary \ref{boundactive}.)
Moreover, by the time any such thread terminates its execution, none of these vertices remains temporarily free, i.e., some of these vertices become matched
while the others become free.
It follows that the update algorithm causes only $O(\log n)$ vertices to be temporarily free at any point in time,
and all these vertices are being taken care of by appropriate threads of the algorithm.

As shown in Section \ref{app:small}, one can maintain a maximal matching easily and efficiently when the matching size is small,
and   we may henceforth assume that the matching is sufficiently large. (A formal case analysis is provided after Lemma \ref{whp}.)
''Losing'' $O(\log n)$ vertices due to the update algorithm is thus negligible, 
and our challenge is to bound the number of temporarily free vertices due to the adversary.

Consider first temporarily free vertices at levels $\ell$ with $T_\ell / \Delta < 1$.
Note that the simulation time of a level-$\ell$ thread is either $T_\ell / \Delta$ or $T_\ell / \Delta'$, and so
the number of update operations needed for simulating the execution of any level-$\ell$ thread is smaller than 1.
Thus any of the level-$\ell$ threads finishes its execution without simulating it over subsequent update operations, which means that
following an adversarial edge deletion that deletes a matched edge $(u,v)$, the relevant thread
(which executes the   calls $\evalf(u)$ and $\evalf(v)$) terminates its execution before the adversary  makes the next edge deletion.
Thus there are no temporarily free vertices due to the adversary at any such level.   

Next, we consider temporarily free vertices at the remaining levels, i.e., with $T_\ell / \Delta \ge 1$.
For any such level, Invariant \ref{samp:in} holds,
thus the samples of all level-$\ell$ edges always contain at least $(1-2\eps) \cdot \gamma^{\ell}$ edges.
\vspace{7pt}
\\
\noindent
{\bf The offline model.~}
For levels with $T_\ell / \Delta \ge 1$, all   samples are almost full, so the adversary will not delete any matched edge throughout the entire update sequence.
Thus all $\log_\gamma(n-1) + 1$ queues $Q_0,Q_1,\ldots,Q_{\log_{\gamma}(n-1)}$ of temporarily free vertices are always    empty,
and at any time all $O(\log n)$ vertices that are currently temporarily free are being taken care of by appropriate threads of the algorithm.
Assuming the matching is of size at least $\Omega(\log n / \eps)$, it must be  almost-maximal (deterministically).
\vspace{7pt}
\\
\noindent
{\bf The standard oblivious adversarial model.~}
In what follows we consider the oblivious adversarial model, and show that the adversary is not likely to cause too many vertices to become temporarily free.  
As explained in Section \ref{short:app}, we may assume that the length of the update sequence is $O(n^2)$.

Any matched edge is created by the algorithm by first determining its level, and only then performing the   sampling.
If the edge is matched at level $\ell$,  it is chosen uniformly at random from  between $(1-\eps) \cdot \gamma^{\ell}$ and $\gamma^{\ell}$ edges.
We   thus  fix some level $\ell$ with $T_\ell / \Delta \ge 1$, and focus on the matched edges at that level.

Consider any time step $t$, and let $V_t$ be the set of vertices at level $\ell$ at time $t$, excluding vertices that are temporarily free due to the update algorithm.
Let $A_t = A'_t \cap A''_t$,
where $A'_t$ is the event that $|V_t|  \ge c'(\log^4 n /\eps^3)$,
$A''_t$ is the event that at least a $(c''\cdot \eps)$-fraction of the vertices of $V_t$ are temporarily free due to the adversary at time $t$,
and $c'$ and $c''$ are sufficiently large constants.  
The following lemma is central in our proof of the almost-maxiamlity guarantee, and its proof is provided in Sections \ref{highlight}--\ref{cont}.
Then in Section \ref{conclude} we derive the almost-maximality guarantee as a corollary of this lemma.
\begin{lemma} \label{whp}
$\Prob(A_t) = O(n^{-c+2})$, for some constant $c$ that depends on $c'$ and $c''$.
\end{lemma}
\subsection{Proof of Lemma \ref{whp}, part I: Overview and challenges} \label{highlight}
Recall that any matched edge $e$ is sampled uniformly at random from between $(1-\eps) \cdot \gamma^{\ell}$ and $\gamma^{\ell}$ edges.
Consider the edges of the sample $S^*(e)$ of $e$ in the order they are deleted by the adversary,
\emph{even after the edge is removed from the matching}, either by the adversary or by the algorithm.
A matched edge is called \emph{bad} if it is one of the first (at most) $2\eps \cdot \gamma^{\ell}$ edges in this ordering; otherwise it is good.
Invariant \ref{samp:in} guarantees that the samples of all level-$\ell$ edges always contain $\ge (1-2\eps) \cdot \gamma^{\ell}$ edges,
so at most $2\eps \cdot \gamma^{\ell}$ edges are deleted from the sample of each matched edge (while it is matched).
It follows that a good edge cannot get deleted by the adversary while it is matched (hereafter, get \emph{hit});
a bad edge gets hit, unless it is removed from the matching by the algorithm prior to the adversarial deletion of that edge.

The probability of an edge to be bad is $\le \frac{2\eps \cdot \gamma^{\ell}} {(1-\eps) \cdot \gamma^{\ell}}$, which is at most $4\eps$ for all $\eps < 1/2$.
Our argument, alas, is not applied on all matched edges created since the algorithm's outset, but rather on a subset of edges that are matched at a certain time step $t'$,
and there are dependencies on previous coin flips of our algorithm, which are the result of edges being removed from the matching by the update algorithm itself (not the adversary).
Indeed, given that some edge $e$ is matched at time $t'$, the sample of $e$ may be significantly reduced,
which could increase the probability of $e$ being bad.
To overcome this hurdle, we use  $\shuffle$ to show that the fraction of bad edges at any point in time is $O(\eps)$.
To this end, we apply a game, hereafter the \emph{shuffling game},  
where in each step a single edge is either added or deleted (starting with no edges) by the following players:
(1) $\adder$: adds an edge, which is bad w.p.\ $\le 4\eps$, (2) $\shuf$: deletes an edge uniformly at random among the existing edges,
(3) $\mal$: deletes a good edge.
A newly created matched edge is bad w.p.\ $\le 4\eps$, thus $\adder$ assumes the role of creating matched edges by the algorithm,
and so an $O(\eps)$-fraction of the matched edges created by $\adder$ are bad w.h.p.
(In this overview we assume that the number of edges added and deleted is sufficiently large to obtain high probability bounds,
but this assumption is formally justified in Section  \ref{justify}.)
$\shuf$ assumes the role of $\shuffle$ in the algorithm, deleting matched edges uniformly at random.
If the fraction of bad edges during some time interval is $\Theta(\eps)$, the fraction of bad edges deleted by $\shuf$ in this interval is $\Theta(\eps)$ w.h.p.,
hence $\shuf$ does not change the fraction of bad edges by too much.
The role of $\mal$ is not to model the exact behavior of the other (non-shuffled) parts of the algorithm that delete matched edges, but rather to capture the worst-case scenario that might happen. We show that even if the other (non-shuffled) parts were to delete only good edges, the fraction of bad edges would be $O(\eps)$, just as if $\mal$ were not in the game.
Specifically, we prove that the fraction of bad edges at any time step $t'$ is w.h.p.\ $O(\eps)$.
This proof, provided in Section \ref{shuffling}, is nontrivial and makes critical use of the fact (see Lemma \ref{faster}) that $\shuf$ is faster than $\mal$ by a logarithmic factor.


Equipped with this bound on the fraction of bad edges at any time step,
we consider the last time   $t'$ prior to $t$ in which the queue $Q_\ell$ of temporarily free level-$\ell$ vertices is empty, i.e., $Q_\ell$ is non-empty in the entire time interval $[t'+1,t]$, which means that $\temp_\ell$ is never idle during that time.
We need to bound the fraction of bad edges not only among the ones matched at time $t'$, but also among those that get matched between times $t'$ and $t$.
The fraction of bad edges among those matched at time $t'$ is $O(\eps)$ w.h.p.\ by the shuffling game; as for
those that get created later on, there is no dependency on coin flips that the algorithm made prior to time $t'$, and so the probability of any of those edges to
be bad is $\le 4\eps$, independently of whether previously created matched edges are bad, and by Chernoff we get that the fraction of bad edges among them is also $O(\eps)$ w.h.p.
The formal proof for this bound on the fraction of bad edges among those is provided in Section \ref{cont}, and it implies that only an $O(\eps)$-fraction of   those edges may get hit w.h.p., and thus get into the queue.
This bound, however, does not suffice to argue that the number of  vertices in $Q_\ell$ at time $t$ is an $O(\eps)$-fraction of the matching size, due to edges that get deleted from the matching by the algorithm itself. Since $\temp_\ell$ is no slower than the other sub-schedulers (as its simulation parameter is $\Delta'$),
we show in Section \ref{cont} that it removes vertices from $Q_\ell$ in the interval $[t'+1,t]$ at least at the same (up to a constant) rate   as  matched edges get deleted by the algorithm.
By formalizing  these assertions and carefully combining them, we conclude with the required result.

\subsection{Proof of Lemma \ref{whp}, part II: The shuffling game} \label{shuffling}
In this section we describe and analyze a game, hereafter the \emph{shuffling game}, that is used in the proof of Lemma \ref{whp}.
We first study the game under simplifying assumptions (Section \ref{simplify}), and then (Section \ref{justify}) justify those assumptions.

Let $\varepsilon$ be a parameter that is larger than $\eps$ by a sufficiently large constant; specifically, we may set $\varepsilon = 8\epsilon$.
Although there is no need to use this additional parameter $\varepsilon$ and one may simply use $4\epsilon$ instead of it, the usage of this parameter helps
to separate the analysis of the shuffling game from the rest of
the proof of Lemma \ref{whp}, and more generally, from the rest of the almost-maximality argument of Section \ref{sec:almost}.

At the game's outset there are no edges.
Edges are   added by $\adder$, where each added edge is bad with probability $\le \varepsilon/2$.
Indeed, recall that any newly created matched edge is bad w.p.\ $\le 4\eps = \varepsilon/2$, thus $\adder$ assumes the role of creating matched edges by the algorithm.
Edges are   deleted by both $\shuf$ and $\mal$, where each edge deleted by $\shuf$ is chosen uniformly at random among all existing edges,
thus it assumes the role of $\shuffle$ in the algorithm,
and $\mal$ deletes only good edges.
The role of $\mal$ is not to model the exact behavior of the other (non-shuffled) parts of the algorithm that delete matched edges, which is indeed quite challenging, but rather to capture the worst-case scenario that might happen. Indeed, observe that the fraction of bad edges is maximized at the (extremely unlikely) scenario that all non-shuffled parts of the algorithm always delete good edges.

$\shuf$ deletes edges at a rate that is faster than $\mal$ by a logarithmic factor.
To be more accurate, we will show in Section \ref{justify} that $\shuf$ is faster than $\mal$ by (at least) a logarithmic factor, but only for sufficiently long time intervals,
and also show there how to cope with short time intervals.
We ignore this technicality in Section \ref{simplify}. 
If $x'$ denotes all deleted edges within a (sufficiently long) time interval and $x$ denotes the number of edges deleted by $\shuf$ therein, then we have $x' = x(1+ O(1/\log n))$.
As explained in Section \ref{justify} (see the remark following Lemma \ref{faster}), the constant hiding in this $O$-notation can be made as small as needed.

\subsubsection{Analysis under simplifying assumptions} \label{simplify}
In this section we study a deterministic version of this game, and in Section \ref{justify} adjust it using Chernoff bounds.
Specifically, we assume that the fraction of bad edges added by $\adder$ is $\varepsilon/2$,
and if the fraction of bad edges is lower bounded by $\varphi$ throughout any sequence of consecutive edge deletions of $\shuf$, for any parameter $\varphi \ge \varepsilon$,
we assume that the fraction of bad edges deleted by $\shuf$ among those deleted edges is at least $\varphi$.
Our goal is to show that at any point in time, the fraction of bad edges is at most $e \varepsilon$.

Fix any time step $t^*$. We next show that the fraction of bad edges at time $t^*$ is at most $e \varepsilon$.
We may assume that the fraction of bad edges at time $t^*$ exceeds $\varepsilon$, otherwise we are done.
Denoting by $t$ the last point in time before $t^*$ where the fraction is $\varepsilon$,
we thus have $t < t^*$,  and the fraction of bad edges is at least $\varepsilon$ in the entire time interval $[t,t^*]$.
For any time $\tau \in [t,t^*]$, denote the number of edge insertions (by $\adder$) during the time interval $[t,\tau]$ by $y_\tau$ and the number of edge deletions (by $\mal + \shuf$) by $x'_\tau$,
where $x_\tau$ of them are due to $\shuf$, and we have $x'_\tau = x_\tau(1+O(1/\log n))$;
for concreteness, we take the constant hiding in this $O$-notation to be 1, i.e., we have $x'_\tau = x_\tau(1+1/\log n)$

Denote the number of   edges at time $t$ by $k$, and note that the number of bad edges then is $\varepsilon \cdot k$.
The number of bad edges $k^{bad}_\tau$ at time $\tau$ is bounded from above by $\varepsilon k + \frac{\varepsilon}2 y_\tau - \varepsilon x_\tau$, and the total number of edges at that time is $k_\tau := k + y_\tau - x'_\tau$.
It is easy to verify that for any $\tau$ such that $y_\tau \ge x'_\tau$, we have
\begin{equation} \label{basicgoal}
\varepsilon k + \frac{\varepsilon}2 y_\tau - \varepsilon x_\tau ~\le~  e \varepsilon (k + y_\tau - x'_\tau).
\end{equation}
Hence if $y_{t^*} \ge x'_{t^*}$, the  fraction of bad edges at time $t^*$ is bounded from above by $e \varepsilon$, and we are done.

In what follows we assume that $y_{t^*} < x'_{t^*}$, hence $k_{t^*} < k$.
We argue  that for any $\tau$ such that $k_{\tau} \ge k/2$ and for any parameter $\varphi \ge \varepsilon$,
\begin{equation} \label{basis}
\varphi k + \frac{\varepsilon}2 y_\tau - \varphi x_\tau ~\le~ (1+1/\log n) \cdot \varphi (k + y_\tau - x'_\tau).
\end{equation}
As $\varphi \ge \varepsilon$, it suffices to prove that
$\varphi k + \frac{\varphi}2 y_\tau - \varphi x_\tau ~\le~ (1+1/\log n) \cdot \varphi (k + y_\tau - x'_\tau)$, or equivalently:
\begin{equation} \label{eqtau1}
0 ~\le~ y_\tau / 2 + \left(y_\tau/\log n + k/\log n \right) + \left(x_\tau - x'_\tau - x'_\tau/\log n\right).
\end{equation}
Noting that $y_\tau + k = x'_\tau + k_{\tau}$ and $x_\tau - x'_\tau - x'_\tau/\log n \ge -2x'_\tau/\log n$, it follows that the right-hand side of Equation (\ref{eqtau1})
is at least $y_\tau/2 + k_{\tau}/\log n - x'_\tau / \log n \ge (y_\tau + k_{\tau} - x'_\tau)/\log n$. Since $y_\tau + k = x'_\tau + k_{\tau}$ and $k_{\tau} \ge k/2$, we have
$(y_\tau + k_{\tau} - x'_\tau)  \ge  2k_{\tau} - k \ge 0$, thus proving that the right-hand side of  Equation (\ref{eqtau1})
is non-negative, which, in turn, proves the validity of Equation (\ref{basis}).

Let $i \ge 0$ be the index such that $k/k_{t^*} \in (2^i,2^{i+1}]$.
We next prove by induction on $i$ that, assuming the fraction of bad edges never goes below $\rho$ in the time interval $[t,t^*]$, for any parameter $\rho \ge  \varepsilon$,
the fraction of bad edges at time $t^*$ is at most $(1+1/\log n)^i \cdot \rho$.
Since $i \le \log k$ and $k \le n/2$ and as the fraction of bad edges never goes below $\varepsilon$,   this inductive statement for $\rho = \varepsilon$ yields an upper bound of $e \varepsilon$ on this fraction.

The basis $i=0$ of the induction follows from Equation (\ref{basis}) for $\tau = t^*$ and $\varphi = \varepsilon$. (To apply this equation, we rely on the fact that the fraction of bad edges is lower bounded by $\varepsilon$ in the   time interval $[t,t^*]$.)

For the induction step, we assume that the statement holds for $i-1$, $i \ge 1$, and prove it for $i$.
Since $i\ge 1$, we know that $k_{t^*} < k/2$. Let $t_1 > t$ be the last point in time before $t^*$ where the number of   edges is $k/2$.
Note that the number of   edges is upper bounded by $k/2$ in the time interval $[t_1,t^*]$.
Since there are $k/2$ edges at time $t_1$ and as the fraction of bad edges is lower bounded by $\varepsilon$ in the time interval $[t,t_1]$,
Equation (\ref{basis})  for $\tau = t_1$ and $\varphi = \varepsilon$ yields
$\varepsilon k + \frac{\varepsilon}2 y_{t_1} - \varepsilon x_{t_1} ~\le~ (1+1/\log n) \cdot \varepsilon (k + y_{t_1} - x'_{t_1})$.
In other words, the fraction $\varepsilon_1$ of bad edges at time $t_1$ satisfies $\varepsilon_1 \le (1+1/\log n) \cdot \varepsilon$.
Let $\tilde t_1$ be a point in time in the interval $[t_1,t^*]$ minimizing the fraction of bad edges; denote by $\tilde \varepsilon_1$ this fraction
and by $\tilde k_1$ the number of bad edges at time $\tilde t_1$, and notice that
$\tilde \varepsilon_1 \le \varepsilon_1 \le (1+1/\log n) \cdot \varepsilon$ and $\tilde k_1 \le k/2$.
Note that the fraction of bad edges never goes below $\tilde \varepsilon_1$ in the time interval $[\tilde t_1,t^*]$.
Also, we have $\tilde k_1 /k_{t^*}  \le (k/2) /k_{t^*} \in (2^{i-1},2^{i}]$,
thus denoting by $\tilde i$ the integer such that $\tilde k_1 / k_{t^*} \in (2^{\tilde i},2^{\tilde i +1}]$,
it follows that $\tilde i \le i-1$.
Since $\tilde i \le i-1$ and $\tilde \varepsilon_1 \ge \varepsilon$,
we can apply the induction hypothesis in the time interval $[\tilde t_1, t^*]$, with the index $\tilde i$  and $\rho = \tilde \varepsilon_1 \ge \varepsilon$.
By induction, the fraction of bad edges at time $t^*$ is at most $$(1+1/\log n)^{\tilde i} \cdot \tilde \varepsilon_1 ~\le~
(1+1/\log n)^{i-1} \cdot (1+1/\log n) \cdot \varepsilon
~=~ (1+1/\log n)^i \cdot \varepsilon.$$

\subsubsection{Justifying the assumptions} \label{justify}

\paragraph{Short time intervals may be ignored.~}
The shuffling game guarantees that at any point in time, the fraction of bad edges is at most $O(\varepsilon)$, where $\varepsilon$ is smaller than $\eps$ by a sufficiently large constant.
Recall that the good edges are not likely to be deleted (from the matching) by the adversary, thus if the fraction of bad edges is $O(\eps)$, it means that
the adversary is not likely to delete more than an $O(\eps)$-fraction of the existing matched edges.
Our probabilistic argument handles each level separately, and it reasons (using the shuffling game) only about levels in which the current number of matched vertices is large enough, namely, $\Omega(\log^4 n /\eps^2)$.
In particular, we only apply the shuffling game for levels $\ell$ in which the current number of matched vertices is $\Omega(\log^4 n /\eps^2)$.
Fix an arbitrary   level $\ell$, and note that we may restrict our attention to \emph{long} time intervals, namely, intervals during which $\Omega(\log^4 n /\eps)$ level-$\ell$ edges are deleted from the matching; indeed, if at the beginning of the interval the fraction of bad edges is $O(\eps)$, then under the assumption that the current number of matched vertices is $\Omega(\log^4 n /\eps^2)$, this fraction must remain $O(\eps)$ following the deletion of any $O(\log^4 n /\eps)$ level-$\ell$ edges.

\paragraph{$\shuf$ is faster than $\mal$.~}
We next argue that $\shuf$ is faster than $\mal$ by (at least) a logarithmic factor in any long time interval.  
\begin{lemma} \label{faster}
 $\shuf$ is faster than $\mal$ by (at least) a factor of $\Theta(\log n)$ in any time interval
 during which $\Omega(\log^2 n)$ edges are deleted from the matching, and in particular, in any long time interval.  
\end{lemma}
{\bf Remark.}  It readily follows from the proof of this lemma that the constant hiding in this $\Theta$-notation can be made as small as needed, by taking the constant hiding in the definition of $\gamma  = \Theta(\log n)$ to be sufficiently large.
\begin{proof}
Recall that we fixed an arbitrary   level $\ell$, and that $\shuf$ assumes the role of $\shuffle_\ell$ in the algorithm, deleting level-$\ell$ matched edges uniformly at random.
We associate all other edge deletions performed by our algorithm with $\mal$.
The simulation parameter of $\shuffle$ is $\Delta'$, which means that $\shuffle_\ell$ is no slower than any other level-$\ell$ sub-scheduler.
While this property may suffice for showing that $\shuf$ is no slower than $\mal$, it does not provide a logarithmic separation between them.
We argue that level-$\ell$ edges may get deleted from the matching only by threads of level $\ge \ell$.
This assertion is immediate except for threads that execute Procedure $\evalf$, but for those threads the assertion holds by Corollary \ref{threadeval}(3).
If the deletion is done by a level-$\ell$ thread, we refer to it as a \emph{low deletion}; otherwise it is done by a higher level thread, and we refer to it as a \emph{high deletion}.

We first consider low deletions. 
Recall that the simulation parameter of $\unmatch$ is $\Delta$, which implies that $\unmatch_\ell$  is slower than $\shuffle_\ell$ by a factor of $\gamma = \Theta(\log n)$.
Using the next claim, it follows that at most $\Theta(\log^2 n)/\gamma = O(\log n)$ low deletions are performed not by $\shuffle_\ell$ within any time interval that it takes  $\shuffle_\ell$ to delete $\Theta(\log^2 n)$ edges. 
\begin{claim} \label{either}
Any low deletion is due to either $\shuffle_\ell$ or $\unmatch_\ell$.
\end{claim}
\begin{proof}
We first argue that Procedure $\evalf$ does not trigger any low deletions.
Indeed, any deletion done as part of Procedure $\evalf$ is due to a vertex $v$ choosing a random mate from its lower level neighbors;
if the random mate $w$ of $v$ is at level $\ell$ prior to the sampling, then the deleted edge $(w,w')$ is of level $\ell$, whereas the
level of the thread executing this procedure is at least of the same level as $v$, which is greater than $\ell$ by definition.
Second, any deletion done outside of Procedure $\evalf$ is due to one of the schedulers.
Noting that $\rise$ makes only high deletions outside of Procedure $\evalf$ and
$\temp$ makes no deletions outside this procedure completes the proof of Claim \ref{either}.
\QED
\end{proof}

We continue the proof of Lemma \ref{faster} by considering high deletions
which are more versatile than low deletions.
Nonetheless, since the simulation time of schedulers grows geometrically by a factor of $\gamma = \Theta(\log n)$ with each level,
at most $O(\log n)$ high deletions may be performed within any time interval that it takes $\shuffle_\ell$ to delete $\Theta(\log^2 n)$ edges,
as there are $\Theta(\log^2 n)/\gamma = O(\log n)$ deletions due to threads at level $\ell+1$ and $\Theta(\log^2 n)/\Omega(\gamma^2) = O(1)$ more deletions per any higher level.

Summarizing, within any time interval that it takes  $\shuffle_\ell$ to delete $\Theta(\log^2 n)$ edges, the number of remaining edge deletions (both low and high) performed is bounded by $O(\log n)$. We conclude that $\shuf$ is faster than $\mal$ by a logarithmic factor in any time interval
during which $\Theta(\log^2 n)$ edges are deleted from the matching, and thus also in any longer time interval.

Lemma \ref{faster} follows.
\QED
\end{proof}

\paragraph{Adding randomization.~}
To add randomization, we make sure that the probabilistic argument handles levels in which the current number of matched vertices is $\Omega(\log^4 n /\eps^3)$.
Thus the total number of vertices over all levels that we ignore may be as high as
$\Omega(\log^5 n /\eps^3)$.
For this number to be an $O(\eps)$-fraction of the maximum matching size, we need the maximum matching size to be $\Omega(\log^5 n /\eps^4)$.
This is why the threshold $\delta$ for the maximum matching size is set as $\Theta(\log^5 n /\eps^4)$; refer to the second paragraph after Lemma \ref{whp} in Section \ref{sec:almost}.
As explained in Section \ref{app:small}, one can maintain a matching $\cM_\delta$ with a worst-case update time of $O(\log^5 n / \eps^4)$,
which is guaranteed to be   maximal whenever the maximum matching size is smaller than $\delta =\Theta(\log^5 n / \eps^4)$.
If we could make sure that $\cM_{rand}$ is almost-maximal in the complementary regime,
we would proceed as before: maintain these two matchings $\cM_{rand}$ and $\cM_{\delta}$  throughout the update sequence,
and at any point in time return $\cM_{\delta}$ as the output matching if $|\cM_{\delta}| < \delta$ and $\cM_{rand}$ otherwise.
Our goal is thus to make  sure that $\cM_{rand}$ is almost-maximal w.h.p.\ in the complementary regime.

The argument for proving that $\cM_{rand}$ is almost-maximal w.h.p. whenever the maximum matching size is at least $\delta =O(\log^5 n / \eps^4)$ employs the shuffling game.
In what follows we show how to adjust the deterministic assumptions used in the argument of Section \ref{simplify} using standard Chernoff bounds.
Observe that the argument of Section \ref{simplify} considers at most $\log k = O(\log n)$ time intervals.

Our first assumption was that the fraction of bad edges added by $\adder$ in an arbitrary time interval is $\varepsilon / 2$, which is the deterministic analog for every added edge being bad with probability $\varepsilon/2$.
By scaling, we shall assume   that each added edge is bad with probability $\varepsilon/4$ rather than $\varepsilon/2$, and so the expected fraction of bad edges added by $\adder$ in any time interval is $\varepsilon / 4$.
We may assume that the number of added edges within any time interval considered in the argument of Section \ref{simplify} is $\Omega(\log n/ \eps)$, since ignoring those edges (over at most $O(\log n)$ such intervals) may increase the fraction of bad edges only by a negligible factor. Since the expected number of bad edges added within any such interval is at least $\Omega(\log n)$,
the probability that the fraction of bad edges added within that interval exceeds $\varepsilon/2$ (thus deviating from the expectation by more than a factor of 2) is at most $n^{-c^*}$ by a Chernoff bound.
Here $c^*$ can be made as large as needed by choosing the other constants appropriately.
Applying a union bound over all time intervals considered in the above argument may increase this probability by at most a logarithmic factor.

Our second assumption was that if the fraction of bad edges is lower bounded by $\varphi$ throughout any sequence of consecutive edge deletions of $\shuf$, for any parameter $\varphi \ge \varepsilon$,
then  the fraction of bad edges deleted by $\shuf$ among those deleted edges is at least $\varphi$, which is the deterministic analog for $\shuf$ to delete edges uniformly at random among all existing edges.
Justifying this assumption is more involved than justifying the first one, since scaling cannot be applied here, and we do not want to lose a factor of 2 in the fraction of bad edges deleted by $\shuf$.
Consequently, we will allow to deviate from the expectation only by a factor of $1+O(1/\log n)$.  
We may assume that the number of edges deleted by $\shuf$ within any time interval considered in the argument of Section \ref{simplify} is $\Omega(\log^3 n/ \eps)$,
since ignoring those edges (over at most $O(\log n)$ such intervals) may only slightly increase the fraction of bad edges.  
Concretely, if the current number of matched vertices is $\Omega(\log^4 n /\eps^2)$, the total number of edges that we may ignore in this way
is only an $O(\eps)$-fraction of the matching size. 
Since the expected number of bad edges deleted by $\shuf$ within any such interval is at least $\Omega(\log^3 n)$,
the probability that the fraction of bad edges deleted by $\shuf$ within that interval is less than $\varphi (1-O(1/\log n))$ (thus deviating from the expectation by more than a factor of $1+O(1/\log n)$) is at most $n^{-c^*}$ by a Chernoff bound.
Again, $c^*$ can be made as large as needed by choosing the other constants appropriately.
Applying a union bound over at most $\log k$ time intervals considered in the above argument may increase this probability by at most a logarithmic factor.

Summarizing thus far, it follows that with high probability, the number of bad edges $k^{bad}_\tau$ at time $\tau$ is bounded from above by
$\varepsilon k + \frac{\varepsilon}2 y_\tau - \varepsilon (1-O(1/\log n)) x_\tau$, which requires us to modify Equation (\ref{basicgoal}).
Specifically, it is easy to verify that for any $\tau$ such that $y_\tau \ge x'_\tau$, the modified equation holds:
\begin{equation} \label{basicgoalb}
\varepsilon k + \frac{\varepsilon}2 y_\tau - \varepsilon(1-O(1/\log n))  x_\tau ~\le~  e \varepsilon (k + y_\tau - x'_\tau).
\end{equation}
Hence if $y_{t^*} \ge x'_{t^*}$, the  fraction of bad edges at time $t^*$ is bounded from above by $e \varepsilon$, as before.

Adjusting the argument for the complementary case $y_{t^*} < x'_{t^*}$, where we have $k_{t^*} < k$, requires more work.
Recall that $\shuf$ deletes edges at a rate that is faster than $\mal$ by (at least) a logarithmic factor. If $x'$ denotes all deleted edges within an arbitrary time interval and $x$ denotes the number of edges deleted by $\shuf$ therein,  then $x' = x(1+O(1/\log n))$.
Instead of taking the constant hiding within this $O$-notation to be 1 as in Section \ref{simplify}, we will now take it to be 1/2, i.e., $x' = x(1+ 1/2\log n)$.
Since the upper bound on the number of bad edges $k^{bad}_\tau$ at time $\tau$ has increased, we need to modify Equation (\ref{basis}) as follows:
\begin{equation} \label{basisb}
\varphi k + \frac{\varepsilon}2 y_\tau - \varphi (1-O(1/\log n)) x_\tau ~\le~ (1+1/\log n) \cdot \varphi (k + y_\tau - x'_\tau).
\end{equation}
We prove the validity of this modified equation for any $\tau$ such that $k_{\tau} \ge k/2$ and for any parameter $\varphi \ge \varepsilon$, just as before.
As $\varphi \ge \varepsilon$, it suffices to prove that
$\varphi k + \frac{\varphi}2 y_\tau - \varphi (1-O(1/\log n)) x_\tau ~\le~ (1+1/\log n) \cdot \varphi (k + y_\tau - x'_\tau)$, or equivalently:
\begin{equation} \label{eqtau1b}
0 ~\le~ y_\tau / 2 + \left(y_\tau/\log n + k/\log n \right) + \left(x_\tau - x_\tau/O(\log n) - x'_\tau - x'_\tau/\log n\right).
\end{equation}
We have $y_\tau + k = x'_\tau + k_{\tau}$ just as before. Also, it is easy to see that $x_\tau - x_\tau/O(\log n)- x'_\tau - x'_\tau/\log n \ge -2x'_\tau/\log n$. The rest of the argument of Section \ref{simplify} applies unchanged, except that we use the modified equation (\ref{basisb}) rather than (\ref{basis}).

\paragraph{The shuffling game result.~}
We   proved that for any time step $t^*$ where the number $k^*$ of edges is $\Omega(\log^4 n /\eps^2)$,
the fraction of bad edges is $O(\varepsilon) = O(\eps)$ w.h.p.
That is, there exist sufficiently large constants $c'$ and $c''$, such that
the probability that $k^*$ is $\ge (c'/2)(\log^4 n /\eps^2)$
and the fraction of bad edges is $\ge (c''/4) \cdot \eps$ is  $\le n^{-c}/2$,
where $c$ is a constant that can be made as large as needed.



\subsection{Proof of Lemma \ref{whp}, part III: The actual proof (given the shuffling game result)} \label{cont}
Consider an edge $(u,v)$ that gets hit, i.e., gets deleted by the adversary while it is matched.
As explained in Section \ref{highlight}, such an edge must be bad. 
Moreover, the probability of a matched edge to be bad is at most $4\eps$, 
which coincides with the  shuffling game of Section \ref{shuffling}, where any newly created edge due to $\adder$ is bad with probability $\le 4\eps$.
Following the hit of edge $(u,v)$, its two endpoints $u$ and $v$ become temporarily free,
and they are inserted to the corresponding queue $Q_\ell$.
The sub-scheduler $\temp_\ell$ aims at handling all the temporarily free vertices in $Q_\ell$, one after another.
(Note that $Q_\ell$ contains only vertices that become temporarily free by the adversary.
A vertex that becomes temporarily free due to the update algorithm is handled by appropriate threads instantly, and does not enter $Q_\ell$.)

Denote by $Q^{(j)}_\ell$ the queue $Q_\ell$ at time step $j$.  
For $0 \le t' \le t$, let $B_{t'}$ be the event that $t'$ is the last time step prior to (and including) $t$ where the queue $Q^{(t')}_\ell$ is empty,
that is, the queue $Q^{(t')}_\ell$ at time step $0 \le t' \le t$ is empty, whereas the queue $Q^{(j)}_\ell$ at time step $j$ is non-empty for all $j \in \{t'+1, t'+2,\ldots, t\}$.
(Note that when the queue is empty, no vertex at level $\ell$ is temporarily free due to the adversary.)

We have $\Prob(A_t) = \sum_{t' = 0}^t \Prob(A_t \cap B_{t'})$.
Note that $\Prob(A_t \cap B_{t}) = 0$.
In what follows we fix $t'$ and $t$ such that $t' < t$ and show that $\Prob(A_t \cap B_{t'}) \le n^{-c}$, where
$c$ is a constant that depends on $c'$ and $c''$.
Let $\eta$ be the random variable (r.v.) for the total number of level-$\ell$ matched edges that are present at time $t'$ or get created between times $t'$ and $t$,
and let $\eta_{Q}$ be the r.v.\ for the number of edges among them that got hit between times $t'$ and $t$.
We argue that
\begin{equation} \label{etaprob}
\Prob((\eta \ge (c'/2)(\log^4 n /\eps^3)) \cap (\eta_{Q} \ge (c''/2) \cdot \eps \eta)) ~\le~ n^{-c}.
\end{equation}
To show that Equation (\ref{etaprob}) holds, we break $\eta$ into $\eta^{t'}$ and $\eta^{\ge t'}$, i.e.,  $\eta = \eta^{t'} + \eta^{\ge t'}$,
where the former (respectively, latter) stands for the
r.v.\ for the number of edges (among all $\eta$ edges) that are present at time $t'$ (resp., get created between times $t'$ and $t$).
Similarly, we break $\eta_{Q}$ into $\eta^{t'}_{Q}$ and $\eta^{\ge t'}_{Q}$, i.e, $\eta_{Q} = \eta^{t'}_{Q} + \eta^{\ge t'}_{Q}$.
We reason about the matched edges that are  present at time $t'$ separately from those that get created later on.
For the edges that are present at time $t'$, we apply the shuffling game result (see the end of Section \ref{shuffling}) to get
\begin{equation} \label{etatprob}
\Prob((\eta^{t'} \ge (c'/2)(\log^4 n /\eps^2)) \cap (\eta^{t'}_{Q} \ge (c''/4) \cdot \eps \eta^{t'})) ~\le~ n^{-c}/2.
\end{equation}
 For the edges that get created later on, there is no dependency on coin flips that the algorithm made prior to time $t'$, and so the probability of any of those edges to get hit between times $t'$ and $t$, which is upper bounded by its probability to be bad, is at most $4\eps$.
Moreover, ordering those $\eta^{\ge t'}$ edges by their creation times (i.e., by the time their random samplings are performed),
the probability of any such edge to be bad, which is a necessary condition for it to get hit between times $t'$ and $t$, is at most $4\eps$, independently of whether previously created matched edges are bad or not. We can thus apply a Chernoff bound to get
\begin{equation} \label{etagreatertprob}
\Prob((\eta^{\ge t'} \ge (c'/2)(\log n/\eps)) \cap (\eta^{\ge t'}_{Q} \ge (c''/4) \cdot \eps \eta^{\ge t'})) ~\le~ n^{-c}/2.
\end{equation}
Equation (\ref{etaprob}) thus follows from Equations (\ref{etatprob}) and (\ref{etagreatertprob}) by a union bound, assuming $c'' > 4, \eps < 1/2$.

If the queue is never empty between times $t'$ and $t$, this means that $\temp_\ell$ is never idle.
Following every $T_\ell /\Delta'$ update operations, $\temp_\ell$ removes another vertex from the queue.
Thus $\temp_\ell$ removes at least $x := \lfloor \frac{t-t'}{T_\ell /\Delta'} \rfloor$ vertices from the queue between times $t'$ and $t$.
If the queue is empty at time $t'$,
every vertex in $Q^{(t)}_\ell$ is an endpoint of a matched edge that got hit between times $t'$ and $t$,
hence $|Q^{(t)}_\ell| \le 2\eta_{Q} - x$. Thus if $B_{t'}$ occurs, we have $|Q^{(t)}_\ell| \le 2\eta_{Q} - x$, yielding:
\begin{observation} \label{baseevent}
If both events  $(|Q^{(t)}_\ell| \ge c'' \cdot \eps \eta - x)$ and $B_{t'}$ occur, then $(\eta_{Q} \ge (c''/2) \cdot \eps \eta))$ must occur.
\end{observation}
\ignore{
It follows that
deleted
}


Next, we bound the number of level-$\ell$ matched edges that get deleted by the update algorithm between times $t'$ and $t$.
As a result of such edge deletions, vertices that are matched at level $\ell$ may move to other levels.
We argue that such edges are  deleted by threads of level at least $\ell$.
Moreover, each of these threads may delete just one matched edge at level $\ell$.
This assertion is immediate except for threads that execute Procedure $\evalf$, but for those threads the assertion holds by Corollaries \ref{threadeval}(3) and \ref{falling}(2).
At each level, at most $\alpha$ threads are running at any point in time, where $\alpha$ is some integer constant.
Since the simulation parameter of any thread is either $\Delta'$ or $\Delta = \Delta' / \gamma$,
it follows that the number of level-$\ell$ matched edges deleted by all level-$j$ threads between times $t'$ and $t$ is $\le \alpha \lceil \frac{t-t'}{T_j /\Delta'} \rceil$, for any $j \ge \ell$.
Note that this bound decays by a factor of $\gamma$ with each level, starting from $\alpha \lceil \frac{t-t'}{T_\ell /\Delta'} \rceil \le \alpha(x+1)$.
For any level $j$ where $\alpha \lceil \frac{t-t'}{T_\ell /\Delta'} \rceil< \alpha$, we may use $\alpha$ as a naive upper bound on the number
of level-$\ell$ matched edges deleted by all level-$j$ threads between times $t'$ and $t$.
We conclude that the total number of level-$\ell$ matched edges deleted by the   algorithm between times $t'$ and $t$ is no greater than
$2\alpha (x+1) + \alpha \cdot \log n$, so the number of endpoints of these edges is at most $4\alpha(x+1) + 2\alpha \cdot \log n$.

Note that the two endpoints of any of the $\eta$ level-$\ell$ matched edges that are present at time $t'$ or get created between times $t'$ and $t$
belong to $V_t$, excluding endpoints of edges that got deleted by the update algorithm or by the adversary during this time.
Since all these edges are vertex-disjoint,
it follows that $|V_t| \ge 2\eta - 2\eta_Q - 4\alpha(x+1) - 2\alpha \cdot \log n$.
Thus if the event $|V_t| \le 2\eta - c'' \cdot \eps \eta - 4\alpha(x+1) - 2\alpha \cdot \log n$ occurs, then $(\eta_{Q} \ge (c''/2) \cdot \eps \eta))$ must occur. Moreover, if the event $E_t := \left(\frac{|Q^{(t)}_\ell|}{|V_t|} \ge \frac{c'' \cdot \eps \eta - x}{2\eta-  c'' \cdot \eps \eta - 4\alpha(x+1) - 2\alpha \cdot \log n}\right)$ occurs, then either
$|Q^{(t)}| \ge c'' \cdot \eps \eta - x$ or $|V_t| \le 2\eta-  c'' \cdot \eps \eta -  4\alpha(x+1) - 2\alpha \cdot \log n$ must occur.
Combining these observations with Observation \ref{baseevent} and Equation (\ref{etaprob}), we have
\begin{eqnarray} \label{complete}
\Prob\left((\eta \ge (c'/2)(\log^4 n /\eps^3)) \cap E_t \cap   B_{t'}\right)  ~\le~
\Prob((\eta \ge (c'/2)(\log^4 n /\eps^3))) \cap (\eta_{Q} \ge (c''/2) \cdot \eps \eta)) ~\le~ n^{-c}.
\end{eqnarray}
If the queue at time $t'$ is empty, then any vertex of $V_t$ must be an endpoint of one of the $\eta$ edges, hence $|V_t| \le 2\eta$.
(Recall that $V_t$ does not contain vertices that are temporarily free due to the update algorithm.)
Fix any $\eps < 1/(4(\alpha+1) c'')$.
It is easy to verify that the inequality $\frac{c'' \cdot \eps \eta - x}{2\eta-  c'' \cdot \eps \eta
- 4\alpha(x+1) - 2\alpha \cdot \log n}  \le c''\cdot \eps$ holds  if $\eta \ge (c'/2)(\log^4 n /\eps^3)$.
It follows that
\begin{eqnarray*}
\Prob(A_t \cap B_{t'}) &=& \Prob((A'_t \cap A''_t) \cap B_{t'})
~=~ \Prob\left(\left((|V_t| \ge c'(\log^4 n /\eps^3)) \cap \left(\frac{|Q^{(t)}_\ell|}{|V_t|} \ge c''\cdot \eps\right)\right) \cap B_{t'}\right)
\\ &\le&
\Prob\left(\left((\eta \ge (c'/2)(\log^4 n /\eps^3))  \cap \left(\frac{|Q^{(t)}_\ell|}{|V_t|} \ge c''\cdot \eps\right)\right) \cap B_{t'}\right)
\\ &\le&
\Prob\left(\left((\eta \ge (c'/2)(\log^4 n /\eps^3)) \cap \left(\frac{|Q^{(t)}_\ell|}{|V_t|} \ge \frac{c'' \cdot \eps \eta - x}{2\eta-  c'' \cdot \eps \eta -  4\alpha (x+1) - 2\alpha \cdot \log n} \right)\right) \cap B_{t'}\right)
\\ &=& \Prob\left((\eta \ge (c'/2)(\log^4 n /\eps^3)) \cap E_t \cap   B_{t'}\right)  ~\le~ n^{-c},
\end{eqnarray*}
where the last inequality is by (\ref{complete}).
Since the total number of time steps is $O(n^2)$,
we conclude that $$\Prob(A_t) = \sum_{t' = 0}^t \Prob(A_t \cap B_{t'}) ~=~ \sum_{t' = 0}^{t-1} \Prob(A_t \cap B_{t'}) ~\le~ t \cdot n^{-c} ~=~ O(n^{-c+2}). \inQED$$

\subsection{Corollaries of Lemma \ref{whp}: The almost-maximality guarantee and main results} \label{conclude}

Applying a union bound over all $O(n^2)$ time steps and over all levels,  
we conclude that  at any point in time, with probability at most $O(n^{-c+4} \cdot \log n) = O(n^{-c+5})$,
an $\Omega(\eps)$-fraction of the potential matched vertices at all levels where there are at least $\Omega(\log^4 n /\eps^3)$ vertices are unmatched.
At any level that we ignore, we may lose at most $O(\log^4 n /\eps^3)$ potential matched vertices by definition.
Whenever the size of the maximum matching is $\Omega(\log^5 n /\eps^4)$,
the loss due to the levels that we ignore is negligible, as those levels contain together just an $O(\eps)$-fraction of the potential matched vertices.
Note also  that we disregarded vertices that are temporarily free due to the update algorithm, but the loss due to those vertices is negligible,
as their total number over all levels is bounded by $O(\log n)$.   

Recall that the worst-case update time of our algorithm is $O(\Delta' \cdot \log n) = O(\log^7 n / \eps)$; denote the matching maintained by this algorithm by $\cM_{rand}$.
Set $\delta = \Theta(\log^5 n /\eps^4)$.
As explained above, $\cM_{rand}$ is guaranteed to be almost-maximal w.h.p.\ whenever the maximum matching size is no smaller than $\delta$.
As explained in Section \ref{app:small}, one can   maintain a   matching $\cM_\delta$ with a worst-case update time of $O(\log^5 n /\eps^4)$,
which is guaranteed to be   maximal whenever  $|\cM_\delta| < \delta$, which must hold when the maximum matching size is smaller than $\delta$.
Therefore, we maintain these two matchings $\cM_{rand}$ and $\cM_{\delta}$  throughout the update sequence,
and at any point in time return $\cM_{\delta}$ as the output matching if $|\cM_{\delta}| < \delta$ and $\cM_{rand}$ otherwise.

Summarizing, a worst-case update time of $O(\max\{\log^7 n /\eps,\log^5 n / \eps^4\})$ suffices to maintain a $(1-\eps)$-almost-maximal matching with probability at least $1 - O(n^{-c+5})$ at any point in time (where we never lose more than an $O(\eps)$-fraction of the potential matched vertices).
\begin{theorem} [standard oblivious adversarial model] \label{maint}
Let $\eps \ll 1$ be an arbitrary parameter.
Starting from an empty graph on $n$ fixed vertices, an $(1-\eps)$-maximal matching (and thus a $(2+\eps)$-approximation for both the maximum matching and the minimum vertex cover) can be maintained over any
sequence of edge updates with a (deterministic) worst-case update time of $O(\max\{\log^7 n /\eps,\log^5 n / \eps^4\})$, where the almost-maximality guarantee $1-\eps$ holds both in expectation and w.h.p.
\end{theorem}
{\bf Remarks.}
~(1) The argument above does not bound the expected almost-maximality guarantee directly,
but an expected bound is   implied by the definition of expectation, following similar lines as in the proof of Lemma \ref{whp}.
On the other hand, an expected bound would not yield our high probability guarantee.
\vspace{5pt}
\\
~(2) Note that the output matching may change significantly following a single edge update,
since $\cM_{\delta}$ and $\cM_{rand}$ may be very different. 
Using \cite{Sol18}, however, we can make sure that the number of changes (replacements) to the output matching
is linear in the worst-case update time of the algorithm.
\ignore{
 by maintaining a $(1+\eps)$-approximate maximum matching
over the union of $\cM_{\delta}$ and $\cM_{rand}$ using a \emph{modified version} of the algorithm of \cite{GP13}.
Specifically, for bounded degree graphs, the worst-case update time of the algorithm of \cite{GP13} is $O(1/\eps^2)$,
but the worst-case number of replacements to the matching maintained by the algorithm of \cite{GP13} may be huge.
Consequently, we first strengthen the algorithm of \cite{GP13} to bound the number of replacements to the matching by $O(1/\eps^2)$ in the worst-case,
and use this stronger version of the algorithm instead of the original version.
This will increase the approximation guarantee to $(2+\eps)(1+\eps)$, which is negligible, and the update time will grow by a factor of $O(1/\eps^2)$.
(See Sections \ref{b11}, \ref{b21} and \ref{changes} for a detailed argument.)
}
\vspace{5pt}
\\
~(3) While the worst-case update time in Theorem \ref{maint} holds deterministically, the almost-maximality guarantee is probabilistic.
Using a symmetric argument, we can ensure that the almost-maximality guarantee holds deterministically, but then the bound on the worst-case update time will be probabilistic.

\vspace{6pt}
Consider next the offline model, where in any point in time and for any subset $E_v$ of edges adjacent to any vertex $v$ ($E_v$ need not contain all edges adjacent to $v$),
we can find in time linear in $E_v$ an edge from $E_v$ that will be deleted from the graph after a constant fraction of the edges of $E_v$ have been deleted.
For the offline setting, our algorithm is deterministic. In particular, there is no reason to maintain the matching
$\cM_{\delta}$, thus shaving the second term ($\log^5 n / \eps^4$) from the update time bound.
Also, as there is no need to use $\shuffle$ for the offline setting,
which works faster than the other schedulers by a logarithmic factor, this shaves a logarithmic factor from the update time immediately.
\begin{theorem} [offline model]
Let $\eps \ll 1$ be an arbitrary parameter. In the offline model,
starting from an empty graph on $n$ fixed vertices, an $(1-\eps)$-maximal matching (and thus a $(2+\eps)$-approximation for both the maximum matching and the minimum vertex cover) can be maintained deterministically over any sequence of edge updates with a worst-case update time of $O(\log^6 n /\eps)$.
\end{theorem}
%

\section{Two Simplifying Assumptions} \label{just}
In this section we demonstrate that the assumptions used in Sections \ref{schedulersAnal} and \ref{sec:almost}  do not lose generality.

\subsection{Long Update Sequences} \label{short:app}
If the update sequence is long, then we break it into subsequences of length at most $t_{max} := \Theta(n^2)$ each,
and run two instances of our dynamic matching algorithm on two different dynamic graphs, $G^{old}_i$ and $G^{new}_i$;
denote the maintained matchings by $\cM^{old}_i$ and $\cM^{new}_i$, respectively.
We maintain the invariant that whenever a new update subsequence starts, the two graphs $G^{old}_i$ and $G^{new}_i$ are identical.
Consider such a time $i_0$ when a new update subsequence starts. At this stage we set $\cM^{old}_i = \cM^{new}_i$, and then restart $G^{new}_i$ and $\cM^{new}_i$ as empty.
The matching that we output for the current update subsequence, i.e., for  the next $t_{max}$ update steps $i = i_0, i_0+1,\ldots,i_0 + t_{max}-1$,
is $\cM^{old}_i$.  
During these $t_{max}$ update steps, the graph $G^{old}_i$ is being changed dynamically only by the adversary, i.e., a single edge change per update step.
The graph $G^{new}_{i}$ is being changed at a faster rate during these $t_{max}$ update steps.
Firstly,  
 we add all edges of $G^{old}_{i}$ to $G^{new}_{i}$, $O(1)$ edges per update step, skipping edges that got deleted by the adversary before reaching the update step
in which they are scheduled  to be added.
In addition, the graph $G^{new}_i$ is being changed dynamically   by the adversary, similarly to  $G^{old}_i$,
both by the new edges (that get inserted by the adversary) and by the ones of $G^{old}_{i}$ that got added to $G^{new}_{i}$ and later get deleted by the adversary.

We run our dynamic matching algorithm on both dynamic graphs $G^{old}_i$ and $G^{new}_i$ to maintain $\cM^{old}_i$ and $\cM^{new}_i$, respectively.
Since we make sure to ``grow'' $G^{new}_i$ at a sufficiently fast rate, we will have $G^{new}_i = G^{old}_i$ by the time $i$ reaches $i_0 + t_{max} - 1$.
Then we can again set $\cM^{old}_i = \cM^{new}_i$, restart $G^{new}_i$ and $\cM^{new}_i$ as empty, handle the next update subsequence of length $t_{max}$ in the same way, and repeat.

When a new update subsequence starts,  the old graph and matching are thrown away, which implies that any matching is being maintained throughout $2t_{max}$ edge updates:
During the first $t_{max}$ edge updates as $\cM^{new}_i$ and during the last $t_{max}$ edge updates as $\cM^{old}_i$.
By the above description, both graphs change dynamically at the same (up to constants) rate as that by the adversary, which enables us to view these graphs as two ordinary dynamic graphs.
Finally, running two instances of our algorithm rather than one
will increase the worst-case update time by another constant factor.
We have henceforth formally justified the assumption that the original update sequence is of length $O(n^2)$.

\subsection{Small matching} \label{app:small}
Consider the naive algorithm for maintaining a maximal matching $\cM$.
Following an edge insertion $(u,v)$, we match $u$ and $v$ if they are both  free.
Following an edge deletion $(u,v)$, if that edge was matched, we scan the neighbors of $u$ and $v$.  
If $u$ has a free neighbor, $u$ will be matched with it, and similarly for $v$.
Note that the number of scanned neighbors of $u$ (and of $v$) should not exceed $2|\cM|+1$, as we may stop at the first free vertex found.
Hence, this   algorithm has a deterministic worst-case update time of $O(|\cM|)$.

Consider the variant of the above algorithm, which limits the neighborhood scan of a vertex for up to $\delta$ arbitrary neighbors; we refer to this scan
as a \emph{partial scan}. We apply this partial scan following an edge deletion $(u,v)$, but in contrast to above, we apply it even if edge $(u,v)$ is unmatched,
in which case the partial scan is applied just for the free endpoints of this edge (if any).
In addition, we pick another arbitrary free vertex $w$ of degree at least $\delta$ (if exists) and apply a partial scan for it too.
(Note that it is trivial to maintain a data structure that consists of the free vertices of degree at least $\delta$.)
Following these partial scans for $u,v$ and (possibly) $w$, we match each of them that is free to a free neighbor, if one is found by the respective partial scan.
Clearly, the worst-case update time of the new algorithm is $O(\delta)$.

Denote by $\cM_\delta$ the matching maintained by the new algorithm, and note that the size $|\cM_\delta|$ of $\cM_\delta$ dynamically changes.
Although we have no control whatsoever on the size of $\cM_\delta$, which could be significantly larger than $\delta$, we argue that $\cM_\delta$ is maximal whenever its size is small enough.
\begin{lemma} \label{maxdelta}
Whenever $|\cM_\delta| < \delta$, the matching $\cM_\delta$  is maximal.
\end{lemma}
\begin{proof}
Suppose for contradiction that at some point $t^*$ in time, $|\cM_\delta| < \delta$ yet there are two free neighbors $v$ and $w$.
Following the previous edge insertion between $v$ and $w$, at least one among $v$ and $w$ was matched, as otherwise the algorithm would match them, and they would remain matched until time $t^*$.
Consequently, there exists a subsequent point in time (until $t^*$) when one of these vertices, w.l.o.g.\ $v$, changes from being matched to free while the other one, $w$, is free;
we consider the last such point in time, denoted $t'$, so that both $v$ and $w$ are free at all times between $t'$ and $t^*$.
Note that $v$ may change from being matched to free only due to a deletion of its matched edge   from the graph, following which a partial scan of $v$ does not find any free vertex.
The only way for this scan not to find any free vertex (and thus miss, in particular, $w$) is if $v$ has at least $\delta$ matched neighbors at that time.
We know that at some subsequent point in time (until $t^*$), the size of the matching drops below $\delta$.
Moreover, our algorithm removes an edge from the matching only if it is deleted from the graph.
On the other hand, every time a matched edge gets deleted from the graph, the algorithm performs a partial scan for an arbitrary free vertex of degree at least $\delta$ (if exists).
We claim that whenever the matching size drops below $\delta$, no free vertex of degree at least $\delta$ may exist.
Indeed, after scanning $\delta$ neighbors of such a vertex via a partial scan, we must reach a free vertex, thus the algorithm would match these two free vertices,
and the matching size would return immediately to $\delta$.
Since $v$ remains free until time $t^*$, this claim implies that $v$'s degree must drop below $\delta$.  
However, when the degree of $v$ reduces from $\delta$ to $\delta-1$, which may only happen due to a deletion of an edge adjacent to $v$,
a partial scan of $v$ is performed. A partial scan in this case is actually a full scan, which must find $w$, and then the algorithm would match $v$ with $w$.
Hence $v$ and $w$ become matched between times $t'$ and $t^*$, a contradiction. \QED
\end{proof}

Recall that our randomized algorithm provides an almost-maximal matching with a worst-case update time of $O((\log^7 n) /\eps)$,
whenever the size of the maximum matching $\cM^*$ is at least $\delta := \Theta(\log^5 n / \eps^4)$.
We will maintain two matchings with a worst-case update time of $O(\max\{\log^7 n /\eps,\log^5 n / \eps^4\})$, one via the algorithm above, denoted $\cM_{\delta}$,
and another via the randomized algorithm, denoted $\cM_{rand}$.
Whenever $|\cM_{\delta}| < \delta$, Lemma \ref{maxdelta} guarantees that $\cM_{\delta}$ is maximal.
In the complementary regime $|\cM_{\delta}| \ge \delta$, the size of the maximum matching $\cM^*$ is at least $\delta$,
and then $\cM_{rand}$ is guaranteed to be almost-maximal w.h.p.
Hence, at any point in time, we return $\cM_{\delta}$ as the output matching if $|\cM_{\delta}| < \delta$ and $\cM_{rand}$ otherwise.

\section{Discussion and Open Problems} \label{discuss}
This paper shows that an almost-maximal matching
can be maintained with a \emph{polylog} worst-case update time, namely  $O(\max\{\log^7 n /\eps,\log^5 n / \eps^4\})$.
Our approach has an inherent barrier of $\Omega(\log n \cdot \log |\cM|)$ on the update time, where $\cM$ is the maintained matching, which is $\Omega(\log^2 n)$ in general.
We managed to shave some logs from the update time and get   close to the $\log^2 n$ barrier.
This improvement, however, requires intricate and tedious adaptations, and lies outside the scope of the current paper; this paper's goal is to present the first algorithm
with a polylog update time for the problem. 

Improving the update time below the barrier $\log^2 n$, towards constant, is a challenging goal. Our work suggests a natural approach: achieving this in the offline model as a first step.
In general, we believe that studying additional classic dynamic graph problems in the offline model is a promising avenue of research:
Leveraging the ``future information'' guaranteed by the offline model may 
yield a deeper understanding of stubborn problems,
which may naturally lead to progress 
in the standard models.

\paragraph{Acknowledgements.~}
The second author is grateful to Danupon Nanongkai and Uri Zwick for their suggestion to study dynamic matchings in the offline model.
We also thank the anonymous reviewers of an earlier version of this paper for their insightful comments, one of which led us to write App.\ \ref{3+eps}.  


\clearpage
\pagenumbering{roman}
\appendix
\centerline{\LARGE\bf Appendix}




\section{A Comparison to Results Done Independently} \label{3+eps}
The SODA'17 paper of \cite{BHN17} maintains a $(2+\epsilon)$-approximate vertex cover and also a  $(2+\epsilon)$-approximate fractional matching
in $O(\log^3 n)$ worst-case update time (ignoring $\eps$-dependencies). The algorithm of \cite{BHN17} is deterministic, and both the algorithm and its analysis are intricate.
In a recent follow-up, Arar et al.\ built on the result of \cite{BHN17} to get a randomized algorithm for maintaining $(2+\eps)$-approximate (integral) matching,
within roughly the same update time as \cite{BHN17}, using an elegant randomized reduction from fully-dynamic integral matching algorithms to fully-dynamic “approximately-maximal” fractional matching algorithms.

Our result is independent of both \cite{BHN17} and \cite{ACCSW17}. Chronologically, our paper started to circulate in Nov.\ 2016,
i.e., after  \cite{BHN17} and before \cite{ACCSW17}.

Our approach has some advantages over the approach of \cite{BHN17,ACCSW17}.
First, if the entire update sequence is known in advance and is stored in some efficient data structure (i.e., the ``offline model'' discussed in Section 1.1), we obtain a \emph{deterministic} algorithm for maintaining an almost-maximal matching (hence a $(2+\eps)$-approximate matching)
with polylog worst-case time bounds, while the approach of \cite{BHN17,ACCSW17} cannot exploit knowledge of the future to achieve a deterministic algorithm. (For the relevance of the offline model, see p.\ 3 in the full version.)
Second, while the works \cite{BHN17,ACCSW17} are tailored to the problems of (approximate) matching and vertex cover
(relying on LP duality, kernels, etc.), our technique seems to be of broader applicability.
In particular, the only property of matchings that our technique really relies on is the vertex-disjointness of edges,
hence we believe it can be naturally used for obtaining polylogarithmic worst-case time bounds for additional problems, such as vertex and edge coloring.

\section{A Sufficient Condition for Winning the Balls and Bins Game} \label{ballsbins}
This section provides a simple proof for the assertion that Player I   wins the game if $b < \frac{k}{(\ln N+1)}$.

Suppose for contradiction otherwise, let $t_1$ be the time when some bin $B_1$ contains $k$ balls, and let $t_i$ be the time when the $i$th (from last) bin  got removed
by Player I, denoted by $B_i$.
(Note that bin $B_{1}$ is removed after $B_{2}$, which is removed after $B_3$, etc.,  in other words $t_1 > t_2 > t_3> \ldots$.)

\begin{claim} \label{basicgame}
For all $i \ge 1$, at time $t_i$, we have $|B_i| \ge k - \sum_{\ell = 1}^{i-1} \frac{b}{\ell}$
and $\sum_{\ell=1}^{i} |B_\ell| \ge i \cdot (k - \sum_{\ell = 1}^{i-1} \frac{b}{\ell})$.
\end{claim}
\begin{proof}
The proof is by induction on $i$, with a trivial basis $i=1$.
By the induction hypothesis, at time $t_i$,  we have $\sum_{\ell=1}^i |B_\ell| \ge i \cdot (k - \sum_{\ell = 1}^{i-1} \frac{b}{\ell})$.
Between times $t_{i+1}$ and $t_i$, at most $b$ balls got added to bins $B_1,\ldots,B_i$ by Player II.
Thus at time $t_{i+1}$,   $\sum_{\ell=1}^i |B_\ell| \ge L:= i \cdot (k - \sum_{\ell = 1}^{i-1} \frac{b}{\ell}) - b$.
This means that at least one bin among $B_1,\ldots,B_i$, denoted $B_{max}$, has size at least $L/i = k -  \sum_{\ell = 1}^i \frac{b}{\ell}$ at time $t_{i+1}$.
Since $B_{i+1}$ got removed at time $t_{i+1}$, at that time we have $|B_{i+1}| \ge |B_{max}| \ge L/i$.
It follows that at time $t_{i+1}$,
$$\sum_{\ell=1}^{i+1} |B_\ell| ~\ge~ L + (L/i) ~=~ (i+1) \cdot (L/i) ~=~ (i+1) \cdot \left(k - \sum_{\ell = 1}^i \frac{b}{\ell}\right),$$
which completes the induction step.
\QED
\end{proof}
Consider the first bin that got removed, denoted $B_r$, where $r \le N$.
This bin removal occurs at time $t_r =0$, i.e., when the entire game starts, and then we have $|B_r| = 0$.
On the other hand, applying the inductive claim for $i = r$, we obtain the following contradiction:
\begin{equation} \label{basicbound}
|B_r| ~\ge~ k - \sum_{\ell = 1}^{r-1} \frac{b}{\ell} ~\ge~ k - b(\ln (r-1) +1) ~>~ k - b(\ln N+1) ~>~ 0.\end{equation}

\end {document}